\documentclass[12pt]{article}

\usepackage{geometry, booktabs, subcaption}
\usepackage{titlesec}
\usepackage[T1]{fontenc}
\usepackage[utf8]{inputenc}
\usepackage{authblk}
\usepackage{amssymb,amsmath}

\usepackage{longtable}
\usepackage{pdflscape}
\usepackage{amsmath, amsthm}
\usepackage{mathtools}
\usepackage{times}
\usepackage{natbib}
\usepackage{amsmath,amsfonts,dsfont,latexsym,mathrsfs,xcolor,upgreek}
\usepackage{booktabs}
\usepackage{graphicx}
\usepackage{bbm}
\usepackage{bm}
\usepackage{enumitem}
\usepackage{placeins}
\usepackage{subcaption}
\usepackage{xr}
\usepackage{hyperref}
\usepackage{rotating}

\hypersetup{
    colorlinks=false,    
    urlcolor=cyan,
    pdfpagemode=FullScreen,
}

\usepackage{geometry}
\newtheorem{theorem}{Theorem}[section]
\newtheorem{lemma}[theorem]{Lemma}

\newtheorem{corollary}[theorem]{Corollary}

\newcommand{\iid}{\stackrel{\mbox{\scriptsize iid}}{\sim}}
\newcommand{\ind}{\stackrel{\mbox{\scriptsize ind}}{\sim}}
\newcommand{\indicator}{\ensuremath{\mathbbm{1}}}

\newcommand{\calF}{\mathcal{F}}

\newcommand{\calN}{\mathcal{N}}
\newcommand{\calO}{\mathcal{O}}

\newcommand{\dd}{\mathrm d}

\newcommand{\norm}[1]{\left\lVert#1\right\rVert}
\newcommand{\derij}{\frac{\partial}{\partial \Lambda_{ij}}}

\newcommand{\E}{\mathbb{E}}
\newcommand{\R}{\mathbb{R}}

\newcommand{\e}{{\rm e}}

\renewcommand{\mid}{\ensuremath{\,|\,}}

\newcommand{\Capp}{C^{\mathrm{app}}}
\newcommand{\Dapp}{D^{\mathrm{app}}}

\usepackage{xparse}
\NewDocumentCommand{\evalat}{sO{\big}mm}{%
  \IfBooleanTF{#1}
   {\mleft. #3 \mright|_{#4}}
   {#3#2|_{#4}}%
}

\date{} 

\title{Bayesian clustering of high-dimensional data via latent repulsive mixtures}

\author[1]{Lorenzo Ghilotti}
\author[2]{Mario Beraha}
\author[2]{Alessandra Guglielmi}
\affil[1]{Department of Economics, Management and Statistics, University of Milano-Bicocca,  Milan, Italy}
\affil[2]{Department of Mathematics, Politecnico di Milano, Milan, Italy}

\begin{document}

\maketitle

\begin{abstract}
Model-based clustering of moderate or large dimensional data is notoriously difficult. 
We propose a model for simultaneous dimensionality reduction and clustering by
assuming a mixture model for a set of latent scores, which are then linked to the observations via a
Gaussian latent factor model. This approach was recently investigated by \cite{Chandra20}. The authors use a factor-analytic representation and assume a mixture model for the latent factors. However, performance can deteriorate in the presence of model misspecification. 
Assuming a repulsive point process prior for the component-specific means of the mixture for the latent scores is shown to yield a more robust model that outperforms the standard mixture model for the latent factors in several simulated scenarios.
The repulsive point process must be anisotropic to favor well-separated clusters of data, and its density should be tractable for efficient posterior inference. We address these issues by proposing a general construction for anisotropic determinantal point processes. 
We illustrate our model in simulations as well as a plant species co-occurrence dataset.
\end{abstract}

KEYWORDS: Anisotropic point process; Determinantal point process; Gaussian factor model;  MCMC; Model-based clustering.

\section{Introduction}

This paper concerns  Bayesian cluster analysis for moderate or high dimensional data.
Specifically, we focus on observations $y_1, \ldots, y_n \in \R^p$. 
Although a variety of models have been proposed in the literature \citep{neal_ddt, teh_coalescent, duan_clustering, natarajan2021cohesion}, Bayesian mixtures constitute a direct approach for model-based clustering; see \cite{fruhwirth2019handbook} for a recent review.
In mixture models, it is assumed that data are generated from $m$ (either random or fixed) homogeneous populations. Typically, each population is assumed to be suitably modeled via a parametric density $f_\theta(\cdot)$ for some parameter $\theta \in \Theta$.
Weights $\bm w = (w_1, \ldots, w_m)$ ($w_h \geq 0$, $\sum_{h=1}^m w_h = 1$) specify the relative frequency of each population. In summary, 
the conditional distribution of data, given parameters, under the mixture model, takes the form
\begin{equation}\label{eq:mix_lik}
    y_1, \ldots, y_n \mid \bm w, \bm \theta \iid \sum_{h=1}^m w_h f_{\theta_h}(\cdot)
\end{equation}
Under the Bayesian approach, suitable priors are assumed for $\bm w$, $\bm \theta = (\theta_1, \ldots, \theta_m)$ and $m$.

The poor performance of Bayesian mixtures when data dimension $p$ is larger than a moderate value, e.g. $p > 10$, is a long-standing issue. 
This is partly due to the poor scalability of the algorithms for posterior inference \citep[see, e.g.,][]{MalsinerWalli16, CelouxHandbook}, and partly due to the asymptotic properties of the model when $p$ is large.
Specifically, Theorem~1 together with Corollaries 1 and 2 in \cite{Chandra20} shows the degeneracy of the cluster estimate when $f_\theta$ in \eqref{eq:mix_lik} is the Gaussian distribution, the de-facto standard. As $p \rightarrow +\infty$ for a fixed $n$, if the covariance matrix is cluster-specific, then with (posterior) probability one, all observations are clustered as singletons, while if the covariance matrix is shared among all the clusters, only one cluster is detected. 

The most popular approach among practitioners to cluster high-dimensional data follows a two-step procedure: first, fitting a latent factor model \citep{lopes2014modern}, a $d$-dimensional score $\eta_i$, where $d \ll p$, is associated with each observation. 
Then, traditional clustering algorithms are applied to the $\eta_i$'s. 
However, this two-step procedure does not allow propagation of the variability induced by the dimensionality reduction step into the cluster estimates.
To overcome this critical limitation, the natural option is to consider a model for simultaneous dimensionality reduction and clustering, by assuming a mixture model for the latent scores, which are then linked to the observations via a Gaussian latent factor model.
This approach was recently investigated by \cite{Chandra20} within the Bayesian nonparametric framework.

However,  mixture models  might  overestimate the number of
clusters to recover the “true” sampling model.
Indeed, the typical postulate in Bayesian mixture models is that parameters $\theta_h$ in \eqref{eq:mix_lik} are i.i.d. from a base distribution.
Under this assumption, if the true data generating density does not agree with \eqref{eq:mix_lik}, the number of clusters a posteriori diverges as $n$ increases \citep{Cai21}.
Two sources of misspecification pertain to
the Gaussian latent factor model with a standard mixture model for the latent scores:
first, the factor-analytic representation entails that data lie close to a $d$-dimensional hyperplane; second, the deviation from such a hyperplane is Gaussian distributed. Both these assumptions can be questioned and are unlikely to hold in practice.
Moreover, \cite{Chandra20} assume a Dirichlet process (DP) mixture prior for the latent scores, which might yield inconsistent estimates even for well-specified models \citep{Mill14}.
Promising results towards consistency of DP mixtures have been obtained in \cite{Ascolani_Biom}, but their assumptions do not hold for the model in \cite{Chandra20}.

The present paper proposes APPLAM: Anisotropic (repulsive) Point Process LAtent Mixture model. 
We consider a model for simultaneous dimensionality reduction and clustering as above but assume a repulsive point process as the prior for the component-specific location parameters of the mixture for the latent scores.
Repulsive mixtures, pioneered by \cite{PeRaDu12}, offer a practical solution to the lack of robustness of mixture models with i.i.d. component parameters.
See also \cite{xu2016bayesian, xie2019bayesian, fuquene2019choosing, bianchini2018determinantal, Ber22}. In these works, repulsive priors are used to jointly model the component-specific location parameters, and, in some of them, their number. However, in the context of a latent mixture model, we argue that to have well-separated clusters of data, it is not sufficient to have well-separated clusters at the latent level, but the repulsion should also take into account the factor analytic model that links the latent variables to the observations.
To this end, we propose an anisotropic determinantal point process (DPP) as the prior for the component-specific parameters, where the matrix of factor loadings drives the anisotropy. We derive a general construction of anisotropic DPPs that induce the desired repulsion.

\subsection{Bayesian clustering via latent mixtures}\label{sec:lamb}

Let $\bm y = (y_1, \ldots, y_n)$, $y_i \in \R^p$, $\Lambda \in \R^{p \times d}$ be the matrix of factor loadings,
$\eta_1, \ldots, \eta_n \in \R^d$ be a set of latent factors, 
and $\Sigma = \text{diag}(\sigma^2_1, \ldots, \sigma^2_p)$ $(\sigma_j > 0)$ be a diagonal covariance matrix. 
Let $\calN_p(\bm a,B)$ denote the $p$--dimensional Gaussian distribution with mean $\bm a$ and covariance matrix $B$. A latent factor mixture model assumes
\begin{equation}
\label{eq:lamb_general}
\begin{aligned}
    y_i \mid \eta_i, \Lambda, \Sigma &\ind \calN_p(\Lambda \eta_i, \Sigma), \\ 
    \eta_i \mid \bm w, \bm \mu, \bm \Delta & \iid \sum_{h=1}^m w_h \calN_d(\mu_h, \Delta_h), 
\end{aligned}
\end{equation}
for $i=1, \ldots, n$.
Extensions to other kernel distributions for the mixture model for latent factors and for observations are straightforward.
The mixture model is completed when assigning the prior for $\bm w = (w_1, \ldots, w_m)$, $\bm \mu = (\mu_1, \ldots, \mu_m)$, $\bm \Delta = (\Delta_1, \ldots, \Delta_m)$, $\Lambda$, and the $\sigma_j^2$'s. In the case of repulsive mixtures, as we propose here, $m$ is also random, so the number of mixture components is learned from the data. See Section~\ref{sec:applam} for the full description of our prior.   

Introducing a set of latent cluster indicator variables $c_i$ such that $P(c_i = h \mid \bm w) = w_h$, we can equivalently state the prior for the $\eta_i$'s in \eqref{eq:lamb_general} as $\eta_i \mid c_i = h, \bm \mu, \bm \Delta \ind \calN_d(\mu_h, \Delta_h)$, $i=1, \ldots, n$.
Therefore, as it is standard in Bayesian mixture models, we cluster data $y_i$'s through the latent variables $\eta_i$'s, i.e., $y_i$ and $y_j$ belong to the same cluster if $c_i = c_j$.
Through the $c_i$'s we can identify the clusters with the allocated components, i.e., those components $h \in \{1, \ldots, m\}$ for which there exists $i$ such that $c_i = h$.

The cluster estimate is interpretable and, hence, useful if observations belonging to different clusters are well separated. 
Repulsive mixture models encourage well-separated clusters by assuming a prior for the cluster centers $\mu_h$'s that favors regular (i.e., well-separated) point configurations.
 A straightforward approach to define such a prior is assuming a repulsive point process which governs both the cardinality $m$ of the components and the locations of the $\mu_h$'s.
In particular, it is possible to define such a process by specifying a density with respect to a Poisson point process. For instance, \cite{quinlan2017parsimonious, xie2019bayesian} assume a pairwise interaction process whose density, with respect to a  suitably defined Poisson process, is
\begin{equation}\label{eq:pairwise_pp}
    p(\{\mu_1, \ldots, \mu_m\}) = \frac{1}{Z} \prod_{j=1}^m \phi_1(\mu_j) \prod_{1 \leq h < k \leq m} \phi_2(\|\mu_h - \mu_k\|)
\end{equation}
where $\phi_1$ is a bounded function, $\phi_2$ is a non-decreasing function, and $Z$ is a normalizing constant that is usually intractable. See \cite{DaVeJo2} and \cite{MoWaBook03} for the definition of density with respect to a Poisson point process.  

This choice of the prior would ensure that different clusters are associated with well-separated latent scores $\eta_i$' s, but is this enough to ensure well-separated clusters of data $y_i$'s?
By the properties of the Gaussian distribution, we have
\[
    \{y_i: c_i = h\} \mid \bm c, \bm \mu, \bm \Delta, \bm \Lambda \iid \calN_p(\Lambda \mu_h, \Lambda \Delta_h \Lambda^\top +\Sigma).
\]
Hence, it is clear that it is not sufficient to encourage a priori that $\|\mu_h - \mu_k\|$ is large to obtain well-separated clusters of datapoints, as the distance between cluster centers is $\|\Lambda \mu_h - \Lambda \mu_k\|$. Such intuition is confirmed by our numerical experiments, as discussed in Appendix \ref{app:comparison_iso_aniso}, where we show that failing to induce the \emph{right repulsion} across cluster centers produces extremely poor cluster estimates a posteriori.
Of course, to produce repulsion across the $\Lambda \mu_h$'s in a pairwise interaction point process, we could easily modify \eqref{eq:pairwise_pp}  by considering $\phi_2(x) \equiv \tilde \phi_2(\|\Lambda x\|)$,  where now the normalizing constant $Z$ depends on $\phi_1, \phi_2$ and, most importantly, $\Lambda$. 
The intractability of the normalizing constant in \eqref{eq:pairwise_pp} poses a significant challenge to posterior simulation. 
\cite{Ber22} discuss how to sample from the posterior distribution of the parameters involved in \eqref{eq:pairwise_pp} using the exchange algorithm \citep{moller2006efficient, murray2006mcmc}, which requires perfect sampling from \eqref{eq:pairwise_pp}. However, \cite{Ber22} investigate the efficiency of the algorithm only when sampling a single real-valued parameter.
Here, instead, the point process density does depend on $\Lambda$ and our preliminary investigation showed that updating $\Lambda$ via the exchange algorithm results in extremely poor mixing due to the high-dimensionality of the matrix $\Lambda$ itself.

\subsection{Our contributions}
\label{sec:our_contrib}

We propose an anisotropic determinantal point process (DPP) in Section~\ref{sec:anisotropic_DPP} as the prior for the component-specific latent factor centers $\{\mu_1, \ldots, \mu_m\}$, where the matrix of factor loadings drives the anisotropy. As a first main contribution, we derive a general construction of anisotropic DPPs together with novel existence conditions for our class of DPPs that are similar to those in \cite{Lav15}, and the associated density does not involve intractable terms. 
We show analytically that the anisotropic DPP induces the desired repulsion. Indeed, we prove that our process is equivalent to a standard (isotropic) DPP distribution for the points $\{\Lambda \mu_1, \ldots, \Lambda \mu_m\}$ on the hyperplane spanned by the columns of $\Lambda$.
Though one could, in principle, directly define a standard DPP on such a hyperplane, the resulting process would be analytically intractable (requiring to perform the eigendecomposition of a kernel function defined on a random hyperplane), and thus of limited practical utility.
We further provide an explicit expression for the spectral density of the DPP, which is essential for simulation purposes. This step also requires a standard approximation \citep[see, e.g.,][]{Lav15} of two infinite sums in the expression of the density of the DPP.  We empirically show that the approximation error is negligible.

Our second contribution is to embed the anisotropic DPP in a latent factor mixture model \eqref{eq:lamb_general}  as a prior for the locations $\mu_h$ and their cardinality.
We discuss several aspects of the proposed model, focusing in particular on the role of the parameters governing the repulsiveness of the process. We also provide guidance on setting such parameters via a data-dependent procedure.
Posterior inference presents non-trivial computational challenges that we address by proposing a Metropolis-within-Gibbs algorithm. In particular, the full conditional for $\Lambda$ is updated using the well-known Metropolis adjusted Langevin algorithm, which, to be computationally efficient, requires the analytical expression for the gradient of (the logarithm of) the density of the DPP; we provide such an expression among our theoretical results.
We illustrate the adequacy of our model by applying it to real and simulated datasets. Specifically, we consider data collecting the occurrence of plant species ($p=123$) at different
sites of the Bauges Natural Regional Park (France). The goal is
to infer the clustering structure of the sites ($n=1139$), such that sites belonging to the same cluster show
similar patterns concerning the occurrence of the plant species.
Further insights on our model are obtained via extensive simulation studies with high dimensional ($p$ up to 1000, which is extremely large for traditional clustering methods) and large ($n = 1000$) datasets. Beyond justifying the need to assume an anisotropic repulsive prior, we compare our APPLAM model with the LAMB one of \cite{Chandra20}. We find that APPLAM produces more robust cluster estimates in the presence of model misspecification.

\section{Methodology}
\label{sec:model}

\subsection{A general construction for anisotropic DPPs}
\label{sec:anisotropic_DPP}

A determinantal point process $\Phi$ on $(\R^d, \mathcal{B}(\R^d))$, where $\mathcal{B}(\R^d)$ is the Borel sigma algebra, is a random subset  $\{\mu_1, \ldots, \mu_m\} \subset \R^d$. See \cite{Macchi75, Lav15, BaBlaKa}.
The probability distribution of a DPP is completely characterized by a continuous complex-valued covariance kernel $K: \R^d \times \R^d \rightarrow \mathbb C$ in terms of its $m$-factorial moment measure. Such a general definition might appear cumbersome to the unfamiliar reader and is reported in Appendix~\ref{app:notation_proofs}. To keep the discussion light, we focus here on DPPs restricted on a compact region $R \subset \R^d$, so that we can write its density with respect to a Poisson process.
However, our results hold for general DPPs defined on the whole $\R^d$.
By Mercer's theorem, the restriction of $K$ to $R \times R$ admits the following spectral representation: $K(x, y) = \sum_{j \geq 1} \gamma_j \xi_j(x) \overline{\xi}_j(y)$, where the $\xi_j$'s form an orthonormal basis for $L^2(R; \mathbb C)$ (eigenfunctions) of complex-valued functions and the $\gamma_j$'s are a summable nonnegative sequence, the eigenvalues; $\overline{\xi}_j$ denotes the conjugate value of $ \xi_j$. 
In particular, when restricted to $R$, a DPP satisfying $\gamma_j < 1$ for all  $j=1,2,\ldots$ admits a density  with respect to the unit-rate Poisson process on $R$ as given by
\begin{equation}
\label{eq:DPPprior}
    p(\{\mu_1, \ldots, \mu_m\}) = \e^{|R| - D} \det\{C(\mu_h, \mu_k)\}_{h, k = 1, \ldots, m}, \ \mu_1, \ldots, \mu_m \in R, \ m=0,1,2,\ldots
\end{equation}
where $C(x, y) = \sum_{j \geq 1}\gamma_j / (1 - \gamma_j) \xi_j(x) \overline{\xi}_j(y)$, $|R| = \int_R \dd x$, $D = - \sum_{j \geq 1} \log(1 - \gamma_j)$, and we adopt the convention that $\det\{C(\mu_h, \mu_k)\}_{h, k = 1, \ldots, m} = 1$ if $m=0$. See \cite{Lav15} for a proof of such results. 
As noted in \cite{Lav15}, the continuity of $C$ implies that the determinant in \eqref{eq:DPPprior} tends to zero as the Euclidean distance $\|\mu_h-\mu_k \|$
goes to zero for some $h\neq k$, which shows that a DPP is a repulsive point process.
The probability distribution of a DPP is uniquely characterized by its kernel $K(\cdot)$ on $\R^d$ or, equivalently, by its spectral representation.
Compared to the expression in \eqref{eq:pairwise_pp}, the density of a DPP features an analytically tractable normalizing constant $\e^{|R| - D}$, which significantly simplifies posterior inference for the hyperparameters of the process; see, e.g., \cite{Ber22} for a discussion on the computational complexities inherited from an intractable normalizing constant. 
On the other hand, to work with a DPP density \eqref{eq:DPPprior}, one needs to restrict the point process on a compact region $R$, which is a nuisance not shared by the pairwise interaction point processes. However, as discussed in Section \ref{sec:elicitation}, the choice of $R$ plays a limited role in our model.

Analytic expressions for eigenvalues $\gamma_j$'s are crucial for inferential purposes. Following the so-called ``spectral approach'' by \cite{Lav15}, \cite{bianchini2018determinantal} and \cite{Ber22} assume $K(x, y) = K_0(x - y)= K_0(\|x - y\|)$, i.e., $K$ is a stationary and isotropic function. 
Instead of modeling $K$,  
they fix the $\xi_j$'s as the Fourier basis and assume a parametric model for the $\gamma_j$'s . This approach ensures the positive definiteness of $K$ and the existence of the DPP density but is not flexible enough for the mixture model \eqref{eq:lamb_general}.
In particular, isotropy of the DPP kernel $K$ conflicts with our goal of forcing repulsion across the $\Lambda \mu_h$'s.
In the following, we provide a new general construction for stationary anisotropic DPPs, an explicit expression for the Fourier transform of its kernel $K_0$, and easy-to-check conditions that guarantee the DPP's existence.
\begin{theorem}\label{teo:aniso_dpp}
Let $\Lambda$ be a (fixed) $p \times d$ real matrix with full rank. Let $W$ be a strictly positive random variable and let $h(y)$ be the marginal density of the random variable $Y$ defined as
\begin{equation}\label{eq:model_Y}
    Y \mid W \sim \calN_d (0,W (\Lambda^T \Lambda)^{-1})
\end{equation}
Let $K_0(x) = \rho h(x) / h(0)$ for $x \in \R^d$ and $\rho > 0$.
Then there exists a DPP $\Phi$ on $\R^d$ with kernel $K(x, y) = K_0(x - y)$ for $\rho \leq \rho_{\max}$ defined as
\begin{equation}\label{eq: rho_max}
    \rho_{\max} = |\Lambda^T \Lambda|^{\frac{1}{2}}(2\pi)^{-d/2}\, \E\bigl[ W^{-\frac{d}{2}} \bigr],
\end{equation}
and
\begin{equation}\label{eq: kernel_Y}
    K_0(x) = \frac{\rho}{\E\bigl[W^{-\frac{d}{2}}\bigr]} \, \E\biggl[W^{-\frac{d}{2}} \exp\biggl(-\frac{||\Lambda x||^2}{2W}\biggr)\biggr], \qquad x \in \R^d.
\end{equation}
If $\varphi(x) = \mathcal{F}(K_0)(x)$ denotes the Fourier transform of $K_0$, we have that
\begin{align}
\varphi(x) = \frac{\rho}{h(0)} \,\E\bigl[\exp\bigl(-2\pi^2 W x^T (\Lambda^T \Lambda)^{-1} x\bigr)\bigr], \qquad x \in \R^d. \label{eq: spect_dens_Y}
\end{align}
Moreover, for any compact $R \subset \R^d$, the restriction of $\Phi$ to $R$ has a density with respect to the unit rate Poisson point process on $R$ if $\rho < \rho_{\max}$. 
\end{theorem}
The kernel $K_0$ characterizing the DPP in Theorem~\ref{teo:aniso_dpp} is not isotropic.
Parameter $\rho$ is the intensity of the process, i.e. it controls the distribution of the number of points in the process. In particular, the expected total number of points in $R$ is equal to $\rho|R|$.
We also emphasize that explicit knowledge of the Fourier transform, as described in Theorem~\ref{teo:aniso_dpp}, is essential for simulation purposes, as one typically approximates the density of the DPP using $\varphi$ as described in Section~\ref{sec:dpp_dens_approx}.
 
The following result is an equivalent characterization of the anisotropic DPP defined above.
\begin{theorem}\label{teo:trans_dpp}
Let  
$\Phi$ be an anisotropic  DPP on $\R^d$ defined as in Theorem~\ref{teo:aniso_dpp}. Then $\tilde \Phi = \{\Lambda \mu \, : \,  \mu \in \Phi\}$ is a stationary and isotropic DPP on $B = \Lambda\R^d$ with kernel 
\[
    \tilde K_0(y) = \frac{\rho |\Lambda^\top \Lambda|^{-1/2}}{\E[W^{-d/2}]} \E \left[ W^{-d/2} \exp\left(- \frac{\|y\|^2}{2W}\right)\right].
\]
\end{theorem}
Theorem~\ref{teo:trans_dpp} sheds light on the repulsiveness induced by our anisotropic DPP and why it is suited in the context of latent mixture models.
As discussed in Section~\ref{sec:lamb}, the model produces interpretable clusters if the points $\{\Lambda \mu_1\, \ldots \Lambda \mu_m\}$ are well separated.
From the theorem above, it turns out that our new construction in Theorem~\ref{teo:aniso_dpp} is equivalent to assuming an isotropic DPP prior for the points $\Lambda \mu_h's$, thus matching our goal of inducing separation across clusters of data.
However, let us stress that the construction in Theorem~\ref{teo:trans_dpp} is not ``operational''. Indeed, to work with a process defined on $\Lambda \R^d$, one would need to perform the eigendecomposition of $\tilde K_0$ at every step of the MCMC algorithm, i.e., anytime $\Lambda$ is updated.
Beyond the analytical difficulty of finding the eigendecomposition of $\tilde K_0$, this also introduces a computational burden not shared by the (equivalent) definition in Theorem~\ref{teo:aniso_dpp}.

The type or strength of repulsion is controlled by the random variable $W \sim p(w)$ in Theorem~\ref{teo:aniso_dpp}.
In the rest of the paper, we consider one specific choice for the random variable $W$ in Theorem~\ref{teo:aniso_dpp}, leading to the anisotropic counterpart of the Gaussian DPPs discussed in \cite{Lav15}, which we refer to as Gaussian-like DPP in the following.
 Moreover, to prove the generality of our approach, in Appendix \ref{sec:app_whmat} we show how to construct an anisotropic counterpart of the Whittle-Mat\'ern DPP, although we will not use it in our examples.

\begin{corollary}\label{cor:gauss_dpp}
Using the same notation of Theorem~\ref{teo:aniso_dpp}, let $W$ be a degenerate random variable defined as
$
    W=|\Lambda^T \Lambda|^{\frac{1}{d}} c^{-2/d}$, for $c > 0$,
where $\Lambda$ is fixed.
Then the kernel $K_0$, its Fourier transform $\varphi = \mathcal{F}(K_0)$ and $\rho_{\max}$ follow here:
\begin{align}
K_0(x) & = \rho \exp \biggl( -\frac{||\Lambda x||^2}{2|\Lambda^T \Lambda|^{\frac{1}{d}} \,c^{-\frac{2}{d}}} \biggr), \qquad x \in \R^d \nonumber\\
\varphi(x) &= \rho \,\frac{(2\pi)^{d/2}}{c} \exp \left( -2\pi^2 |\Lambda^T \Lambda|^{\frac{1}{d}} c^{-\frac{2}{d}} x^T (\Lambda^T \Lambda)^{-1} x \right), \qquad x \in \R^d \label{phi_gauss} \\
\rho_{\max} &= c (2\pi)^{-d/2}. \nonumber
\end{align}
\end{corollary}

We conclude this section by investigating the effect that $\Lambda$ has on the repulsiveness of the DPP. First of all, it is clear that $\Lambda$ induces anisotropy. To visualize this, we
consider the pair correlation function $g(x)$ \citep[PCF,][]{Lav15}; for the Gaussian-like DPP we have:
\begin{equation}
\label{eq:pair_cor}
    g(x) = 1 - \frac{K_0(x)}{K_0(0)} = 1 - \exp \biggl( -\frac{||\Lambda x||^2}{2|\Lambda^T \Lambda|^{\frac{1}{d}} \,c^{-\frac{2}{d}}} \biggr)^2, \quad x \in \mathbb R^d
\end{equation}
and set $p=d=2$ for visual purposes.
Figure~\ref{fig:gauss_pcf} shows the PCFs of two Gaussian-like DPPs with different $\Lambda \in \R^{2\times 2}$.
In the left panel, $\Lambda$ has eigenvectors $e_1 = (1,0)^T$, $e_2=(0,1)^T$ and eigenvalues $\lambda_1 = 1$, $\lambda_2 = \lambda ( = 4)$, which induces stronger repulsion along the horizontal axis than along the vertical one. In the right panel, $\Lambda$ has eigenvectors $e_1 = \sqrt{2}/2 (1,1)^T$, $e_2=\sqrt{2}/2 (-1,1)^T$ and eigenvalues $\lambda_1 = 1$, $\lambda_2= \lambda( = 4)$, which induces stronger repulsion along the bisector of the first quadrant than along the orthogonal direction.

\begin{figure}[t]
\centering
\begin{subfigure}{.5\textwidth}
  \centering
  \includegraphics[width=\linewidth]{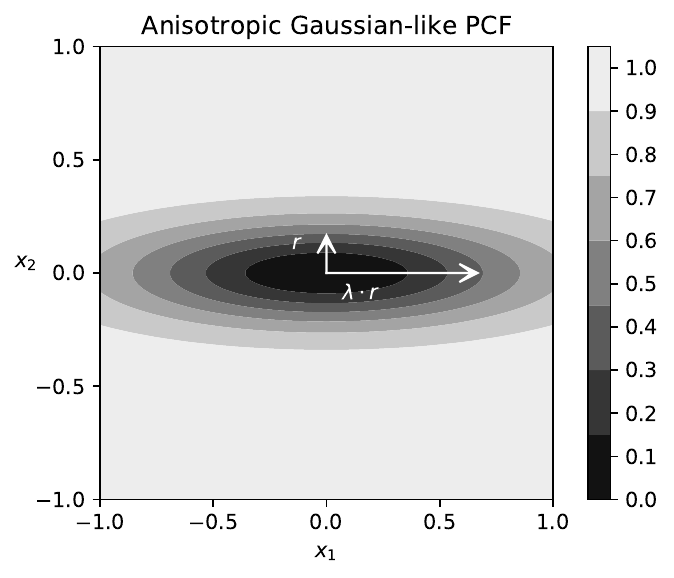}
\end{subfigure}%
\begin{subfigure}{.5\textwidth}
  \centering
  \includegraphics[width=\linewidth]{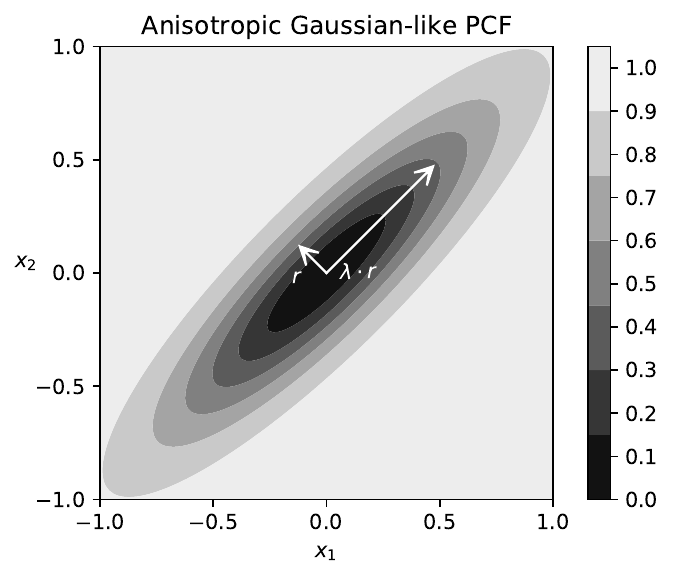}
\end{subfigure}
\caption{Pair correlation function of two Gaussian-like DPPs, one showing strong repulsion along the horizontal direction (left plot) and one showing strong repulsion along the bisector of the first quadrant (right plot).
}
\label{fig:gauss_pcf}
\end{figure}

Although the PCF is useful for visualizing the anisotropy of repulsiveness, it provides little information on \emph{how much} repulsiveness is induced by the DPP.
Several repulsiveness measures have been proposed for spatial point processes. See, e.g., \cite{Lav15, biscio_quant} for a few quantitative indexes.
Here, we focus in particular on the coefficient $p_0$ introduced in \cite{moller2021couplings}
which is a global index of repulsion quantifying the effect of a point on the expected number of points in the process.
For a stationary DPP, $p_0 = \rho^{-1}\int \lvert K_0(x) \rvert^2 \dd x$,
where $\rho:=K_0(0)$.
The larger $p_0$, the more repulsive the DPP is.
For our Gaussian-like DPP in Corollary~\ref{cor:gauss_dpp}, standard computations lead to $p_0 = \rho \pi^{d/2} c^{-1}$,
which shows that the amount of repulsiveness in the process does not depend on $\Lambda$.

\subsection{The APPLAM model}
\label{sec:applam}

The anisotropic repulsive point process latent mixture model proposed here, denoted APPLAM as mentioned before,  assumes likelihood \eqref{eq:lamb_general}. We complete the prior specification as follows. First, following common practice in mixture models, conditional to $\{\mu_1, \ldots, \mu_m\}$, indeed conditionally only to $m$, we  assume
\begin{equation}\label{eq:prior_marks}
    \begin{aligned}
        w_1, \ldots, w_m \mid m &\sim \text{Dirichlet}(\alpha, \ldots, \alpha) \\
        \Delta_1, \ldots, \Delta_m \mid m & \iid \text{IW}_d(\nu_0, \Psi_0)
    \end{aligned}
\end{equation}
where $\text{IW}_d(\nu_0, \Psi_0)$ denotes the $d$-dimensional inverse Wishart distribution, with $\nu_0 > d-1$ degrees of freedom
and mean $\Psi_0 / (\nu_0 - d - 1)$.
Moreover, as  it is common in latent factor models, we assume $\Sigma = \text{diag}(\sigma^2_1, \ldots, \sigma^2_p)$ in \eqref{eq:lamb_general} to be a diagonal matrix and
\begin{equation}\label{eq:prior_sigma}
    \sigma^2_j \iid \text{inv-Gamma}(a_\sigma, b_\sigma), \qquad j=1, \ldots, p,
\end{equation}
i.e., their marginal prior distribution is inverse Gamma with mean $b_\sigma/(a_\sigma-1)$.
As far as the matrix $\Lambda$ is concerned, we impose sparsity through the prior. Indeed, sparse priors are often employed in latent factor models to avoid overparametrization. In particular we assume the Dirichlet-Laplace prior \citep{bhattacharya2015dirichlet} with parameter $a$ as the prior for $\Lambda$, that is, denoting with $\lambda_{jh}$'s the elements of $\Lambda$,
\begin{equation}\label{eq:prior_lambda}
\begin{aligned}
\lambda_{jh}\,|\,\phi, \tau, \psi &\ind \calN(0,\psi_{jh}\phi_{jh}^2\tau^2), \quad j=1, \ldots, p, \, h=1, \ldots, d\\
vec(\phi) \sim \mathrm{Dirichlet}(a,\ldots,a), \quad
\psi_{jh} &\iid \mathrm{Exp}(1/2), \quad
\tau \sim \mathrm{Gamma}(p d  a, 1/2)
\end{aligned}
\end{equation}
where, for any $p\times d$ matrix $A$, $vec(A)$ denotes the vector of dimension $p\times d$ such that $vec(A)_{p(h-1)+ j} =(A)_{j,h} $.  This prior is very popular in the literature on Bayesian factor models because of ease of sampling and good frequentist properties \citep{bhattacharya2015dirichlet}.   
Note also that we have assumed prior independence among the blocks of
parameters above.

The key point in our model is the prior specification for the latent factor centers $\{ \mu_1, \ldots, \mu_m\}$.
To introduce repulsiveness between the points $\{\Lambda \mu_h\}_{h=1}^m$, we first fix a compact $R \subset \R^d$, and assume that the locations $\{\mu_1, \ldots, \mu_m\}$ are distributed as an anisotropic DPP on $R$, conditioned on $m \geq 1$. As shown below, the choice of $R$ plays a very limited role, and we discuss how to fix it in Section \ref{sec:elicitation}. Conditioning on the DPP being non-empty is necessary for our formulation of the model, as \eqref{eq:lamb_general} is not well-defined if $m=0$. Therefore, we assume the following density for the locations $\{\mu_1, \ldots, \mu_m\}$ and their cardinality:
\begin{multline}\label{eq:prior_mu}
     p(\{\mu_1, \ldots, \mu_m\} \mid \Lambda) =  f_{\text{DPP}}(\bm \mu \mid \rho, \Lambda, K_0; R)  \\= \frac{\e^{|R| - D}}{1 - \e^{-D}} \det\{C(\mu_h, \mu_k)\}_{h, k = 1, \ldots, m}, \qquad m \geq 1, \, \mu_1, \ldots, \mu_m \in R
\end{multline}
and $p(\emptyset \mid \Lambda) = 0$, where $D$ and $C$ are defined just after  \eqref{eq:DPPprior}.  The above notation $f_{\text{DPP}}(\bm \mu \mid \rho, \Lambda, K_0; R)$ makes explicit that the density depends on $R$,  the intensity $\rho$, the stationary kernel $K_0$ and anisotropy is driven by $\Lambda$.

\subsection{Prior elicitation for the anisotropic DPP} \label{sec:elicitation}

The prior distributions for $\Sigma$, the $\Delta_h$'s, and $\Lambda$ are shared between Lamb and APPLAM, and we adopt the default values suggested in \cite{Chandra20}. See Appendix \ref{app:elicitation} for further details.
Here, we specifically focus on those hyperparameters defining the DPP density prior in \eqref{eq:prior_mu} when $K_0$ is the anisotropic Gaussian-like kernel in Corollary~\ref{cor:gauss_dpp}, i.e., $R$, $\rho$, and $c$.

First, we discuss the choice of the compact set $R$ of $\R^d$, which is the support of the $\mu_h$'s. Note that, marginally, each $\mu_h$ is uniformly distributed on $R$. We argue that a reasonable choice is to assume $R = [-r, r]^d$, i.e., a hypercube. Indeed, there is no reason to assume a priori that the latent factors exhibit a higher variance along any particular axis. Moreover, such a hypercube should be centered in the origin as in traditional latent factor models, whereby the latent factors are given a standard Gaussian prior.
Clearly, there is non-identifiability between $\Lambda$ and $r$. Indeed, letting $t > 0$ and $\tilde \Lambda = t \Lambda$, $\tilde R = [-r/t, r/t]^d$ leaves the likelihood of the data invariant. Therefore, the choice of $r$ plays a limited role in the model definition.
In particular, our experience shows that posterior inference is robust if $r$ is sufficiently large, i.e., $r \ge 5$. Therefore, we set $R = [-10, 10]^d$ in the rest of the paper.

Parameters $\rho$ and $c$ control the repulsiveness and the intensity of the process. As noted in Section~\ref{sec:anisotropic_DPP}, these quantities are intimately related in DPPs. 
While $\rho |R|$ is the the expected number of points in $R$ of the DPP, note that in \eqref{eq:prior_mu}
we condition on the process being non-empty so that the role of $\rho$ is not as clear as in \eqref{eq:DPPprior}.
Therefore, we suggest fixing those parameters by looking only at the repulsiveness induced by the DPP. In Appendix \ref{app:hyperparams}, we exploit the equivalent representation of the anisotropic DPP given in Theorem~\ref{teo:trans_dpp} to get insights on how such parameters govern the so-called \emph{repulsion range} (i.e., how far should two cluster centers be to influence each other) as well as the strength of the repulsion across close points. We then propose a heuristic approach based on the observed data to set reasonable ranges for hyperparameters $\rho,c$.
In particular, we follow \cite{Lav15, bianchini2018determinantal} and parametrize the DPP by $(\rho, s)$ such that $\rho = s \rho_{\max} \equiv s c (2 \pi)^{-d/2}$, $s \in (0, 1)$, cf. Corollary \ref{cor:gauss_dpp}. In such a way, $s$ takes the interpretation of the strength of repulsion for a fixed $\rho$.

\section{Posterior Simulation}\label{sec:mcmc}

\subsection{Approximation of the DPP density}\label{sec:dpp_dens_approx}

The point process density in \eqref{eq:prior_mu} cannot be evaluated in closed form due to the infinite series in the expression of $D$ and $C$. We follow \cite{Lav15} and approximate it with $f^{\text{app}}_{\text{DPP}}$, which is as in \eqref{eq:prior_mu} where $D$ is replaced by $\Dapp = - \sum_{k \in \mathbb Z_N^d} \log(1 - \varphi(k))$ and $C$ by
\begin{equation}\label{eq:approx_dens}
  \Capp(x, y) = \frac{1}{|R|} \sum_{k \in \mathbb Z_N^d} \frac{\varphi(k)}{1 - \varphi(k)} \exp\left( \frac{2 \pi i}{|R|^{1/d}} \langle k, x - y \rangle\right)
\end{equation}
where $\mathbb Z_N^d = \{-N,\ldots,N\}^d $ and $\langle\cdot, \cdot \rangle$ denotes the scalar product. Note that $\bm \mu \sim f^{\text{app}}_{\text{DPP}}(\rho, \Lambda, C^{\text{app}}; R)$ is still a DPP on $R$, conditioned that it is nonempty.
The truncation level $N$ controls the goodness of the approximation of the density and
sets an upper bound of $(2N + 1)^d$ on the number of points in the DPP \citep[][]{Lav15}.
This approximation can be seen as analogous to the truncation of stick-breaking priors commonly adopted in Bayesian nonparametric mixture models \citep{ishwaran2001gibbs}. For the Dirichlet process, a standard choice is truncating the prior to 50 support points. Similarly, when $d=2$ (the smallest value we consider) and $N=3$, we are setting here an upper bound of 49 points in the DPP, which increases exponentially with the dimension.
Appendix~\ref{app:truncation} reports a simulation showing that the approximation error introduced by the truncation is, in fact, negligible when $\Lambda$ is either fixed to the identity matrix, or random and distributed according to the prior  \eqref{eq:prior_lambda} or the posterior.

\subsection{The Gibbs sampling algorithm}

Including the mixture weights $\bm w$ as finite-Dirichlet distributed (see the first line of \eqref{eq:prior_marks}) in the state space is not convenient, as the sum-to-one constraint on $\bm w$ would lead to complex split-merge reversible jump moves with poor mixing of the chain.
As in \cite{Ber22}, we consider  the following equivalent characterization of the prior for $\bm w$
\begin{equation}\label{eq:prior_jumps}
     \bm w = \left(\frac{S_1}{T}, \ldots, \frac{S_m}{T} \right), \quad T = \sum_{h=1}^m S_h, \quad S_h \iid \text{Gamma}(\alpha, 1). 
\end{equation}
Conditional to $\bm c$ we consider $\bm \mu =\bm  \mu^{(a)} \cup \bm \mu^{(na)}$, $\bm S = [\bm S^{(a)}, \bm S^{(na)}]$ and $\bm \Delta = [\bm \Delta^{(a)}, \bm \Delta^{(na)}]$, divided into allocated and non-allocated components, and denoted with the $(a)$ and $(na)$ superscripts, respectively.  That is, $\bm \mu^{(a)} = \{\mu_{c_i}, \, i=1, \ldots, n\}$, $\bm \mu^{(na)} = \bm \mu \setminus \bm \mu^{(a)}$, and analogously for the terms involving the $S_h$'s and $\Delta_h$'s.

Combining the likelihood \eqref{eq:lamb_general} with the prior as in \eqref{eq:prior_marks}-\eqref{eq:prior_mu}, we can see that the joint distribution of data and parameters has a density,  with respect to some dominating measure on the proper space; see Appendix~\ref{sec:app_measure} for this density and the associated dominating measure. 
Normalization of the weights leads to a term $T^{-n} = (\sum S^{(a)}_{h_1} + \sum S^{(na)}_{h_2})^{-n}$ in the expression of the joint density, which makes it impossible to factorize the density according to the allocated and non-allocated components. 
To overcome this issue, as in \cite{Ber22},  we introduce an auxiliary random variable $u \mid T \sim \text{Gamma}(n, T)$.
We report the joint density of data and parameters in Appendix~\ref{sec:app_measure}.

We propose a Gibbs sampler algorithm, recurring to a Metropolis step when the corresponding full-conditionals cannot be sampled directly. Most of the updates are straightforward, except for two steps, which are outlined below. The full description of the algorithm is in Appendix~\ref{sec:app_Gibbs}.

To update the non-allocated variables $(\bm\mu^{(na)}, \bm S^{(na)}, \bm\Delta^{(na)})$, we disintegrate the joint full-conditional of the non-allocated variables as 
	\[  
	    p(\bm\mu^{(na)}, \bm s^{(na)}, \bm\Delta^{(na)} \mid \text{rest}) = p(\bm\mu^{(na)} \mid \text{rest}) p(\bm s^{(na)} \mid \bm\mu^{(na)}, \text{rest}) p(\bm \Delta^{(na)} \mid \bm\mu^{(na)}, \text{rest}),
	\]
	where ``rest'' identifies all the variables except for $(\bm\mu^{(na)}, \bm s^{(na)}, \bm\Delta^{(na)})$.
	Then $\bm\mu^{(na)} \mid \text{rest}$ is a Gibbs point process with density
	\[
	    p(\{\mu^{(na)}_1, \ldots, \mu^{(na)}_\ell \} \mid \text{rest}) \propto f^{\text{app}}_{\text{DPP}}(\{\mu^{(na)}_1, \ldots, \mu^{(na)}_\ell \}  \cup \bm \mu^{(a)} \mid \rho, \Lambda, K_0; R) \psi(u)^{\ell}
	\]
	where $\psi(u) = \E[\e^{- u S}]$. See also \cite{Ber22}. We employ the birth-death Metropolis-Hastings algorithm in \cite{geyer1994simulation} to sample from this point process density.
	Given $\bm \mu^{(na)}$ it is straightforward  to show
	\begin{align*}
	    \Delta^{(na)}_1, \ldots, \Delta^{(na)}_\ell \mid \text{rest}  &\iid \mathrm{IW}_d(\nu_0,\Psi_0) \\
	    S^{(na)}_1, \ldots, S^{(na)}_\ell \mid \text{rest} &\iid \text{Gamma}(\alpha, 1 + u) \ .
	\end{align*}
	
To update of the latent allocation variables $\bm c$, we found it useful to marginalize over the $\eta_i$'s to get better mixing chains. Hence, we can sample each $c_i$ independently from a discrete distribution over $\{1, \ldots, k+\ell\}$, where $k$ is the number of allocated components, with weights $\omega_{ih}$:
    \begin{align*}
        \omega_{ih} &\propto S_h^{(a)} \mathcal{N}_p(y_i \mid \Lambda \mu_h^{(a)}, \Sigma + \Lambda \Delta_h^{(a)} \Lambda^\top), \quad & h=1, \ldots, k \\
            \omega_{i k + h} & \propto S_h^{(na)} \mathcal{N}_p(y_i \mid \Lambda  \mu_h^{(na)}, \Sigma + \Lambda \Delta_h^{(na)} \Lambda^\top), \quad & h=1, \ldots, \ell \ .
    \end{align*}
    Each evaluation of the $p$-dimensional Gaussian density would require $\calO(p^3)$ operations if care is not taken. However, we take advantage of the special structure of the covariance matrix. Using Woodbury's matrix identity, we have that
    \[
        \left(\Sigma + \Lambda \Delta \Lambda^\top\right)^{-1} = \Sigma^{-1} - \Sigma^{-1} \Lambda \left(\Delta^{-1} + \Lambda^\top \Sigma^{-1} \Lambda \right)^{-1} \Lambda^\top \Sigma^{-1},
    \]
    and hence we need now to compute the inverse of a $d \times d$ matrix. Therefore, evaluating the quadratic form in the exponential requires only $\calO(p)$ computational cost. 
    Moreover, using the matrix determinant lemma, the determinant of the covariance matrix can be computed  as
    \[
        \det(\Sigma + \Lambda \Delta \Lambda^\top) = \det(\Delta^{-1} + \Lambda^\top \Sigma^{-1} \Lambda) \det(\Delta) \det(\Sigma) .
    \]
    This is computed without additional cost by caching operations from the matrix inversion.

\subsection{Updating $\Lambda$ using gradient-based MCMC algorithms}\label{sec:mala_lambda}

As mentioned in Section~\ref{sec:lamb}, sampling from $\Lambda$'s full conditional is non-trivial. 
In particular, we found that random-walk Metropolis-Hastings led to very poor mixing of the MCMC chain, while the adaptive Metropolis-Hastings algorithm in \cite{haario_adaptive} is not feasible here due to the high dimensionality of $\Lambda$. 
Instead, we found the Metropolis adjusted Langevin Algorithm \citep[MALA][]{roberts1996mala} to be more adequate here.
The target density is 
\begin{equation}\label{eq:lambda_fullcond}
\begin{aligned}
p(\Lambda\,|\,\cdots) & \propto p(\bm y\,|\,\Lambda,\bm \eta, \Sigma) p(\Lambda\,|\,\phi,\tau,\psi) \frac{\e^{1 - \Dapp}}{1-\e^{-\Dapp}} \det[\Capp](\mu_1,\ldots,\mu_m),
\end{aligned}
\end{equation}
where $\cdots$ denotes all the variables except $\Lambda$. Although not explicitly stated, $\Dapp$ and $\Capp$ both depend on $\Lambda$.

To sample from \eqref{eq:lambda_fullcond} using MALA, we must evaluate $\nabla \log(p(\Lambda\,|\,\cdots))$.
In a preliminary investigation, we tried using automatic differentiation \citep[AD,][]{griewank1989automatic} to get $\nabla \log(p(\Lambda\,|\,\cdots))$, as it requires only the implementation of a function evaluating $\log(p(\Lambda\,|\,\cdots))$. Unfortunately, we found that this strategy is viable
only in trivial scenarios, i.e. up to $p=50$ and $d=3$ due to RAM memory requirements; see Figure~\ref{fig:ad_vs_grad} in Appendix~\ref{app:analytic-automatic}.  
Therefore, in the following theorem, we provide the analytical expression of $\nabla \log(p(\Lambda\,|\,\cdots))$
when the associated DPP is Gaussian-like. See Appendix~\ref{sec:app_whmat} for the Whittle-Mat\'ern DPP case.
\begin{theorem}\label{teo:grad_lambda}
    Under the Gaussian-like DPP prior, the gradient of the log-full conditional density of $\Lambda$ equals
\begin{equation*}
    \begin{aligned}
        \nabla \log p(\Lambda\,|\,\cdots) & = \Sigma^{-1} \sum_{i=1}^n{(y_i-\Lambda\eta_i)\eta_i^\top} - \frac{1}{(\psi \odot \phi^2) \tau^2} \odot \Lambda +
        \\
        & \qquad + (2 \pi^2 c^{-\tfrac2d})\sum\limits_{k \in \mathbb{Z}_N^d}{g^{(k)} \frac{\varphi(k)}{(1-\varphi(k))^2} \Bigl[ \frac{1-\varphi(k)}{1-\e^{-D^{\mathrm{app}}}} -  v_k^T \left( \Capp\right)^{-1} u_k \Bigr] }
    \end{aligned}
\end{equation*}
where $\varphi(k)$ is defined in \eqref{phi_gauss}, $\odot$ denotes the elementwise (Hadamard) product, 
\[
    g^{(k)} = 2|\Lambda^T \Lambda|^{\tfrac1d} \Lambda (\Lambda^T \Lambda)^{-1} \Bigl[ \frac1d k^T (\Lambda^T \Lambda)^{-1} k \mathds{1}_d - k ((\Lambda^T \Lambda)^{-1} k)^T \Bigr],
\]
$u_k$ and $v_k$ are $m$-dimensional column vectors for each $k \in \mathbb{Z}^d$ with entries
\begin{equation*}
    (u_k)_j = \e^{2\pi i \langle k, T\mu_j\rangle} , \quad (v_k)_j = \e^{-2\pi i \langle k, T\mu_j\rangle} , \quad j=1,\ldots,m
\end{equation*}
and $\Capp := \Capp(\mu_1,\ldots,\mu_n)$.
\end{theorem}

Figure~\ref{fig:ad_vs_grad} in Appendix~\ref{app:analytic-automatic} reports a comparison of the memory requirements and the runtime execution using the AD gradients or our analytical expressions. 
In particular, using the AD gradients requires roughly  100x more memory, which makes a significant difference in practice  since this is the bottleneck of our algorithm.
The runtimes are 10x larger when using AD as well.
The stepsize parameter of the MALA algorithm is tuned running short preliminary chains to get an acceptance rate around $20\%$. In particular, we found that values between $10^{-8}$ and $10^{-10}$ usually give a good mixing of the MCMC chain.

\section{Simulation Studies}\label{sec:simu}

Appendix~\ref{app:simu1} shows several simulated examples aimed at showcasing different aspects of our model. Beyond investigating the effect of truncating the DPP density (Appendix~\ref{app:truncation}) and the usefulness of the analytic expression of the gradients found in Theorem~\ref{teo:grad_lambda} (Appendix~\ref{app:analytic-automatic}), we show the need of assuming the anisotropic DPP, as opposed to an isotropic DPP prior, in the context of latent mixture models in Appendix~\ref{app:comparison_iso_aniso}. Appendix~\ref{app:hyperparams} illustrates 
a procedure to fix hyperparameters.
Moreover, in Appendix~\ref{app:robustness} we demonstrate the robustness of posterior inference to the hyperparameters of the anisotropic DPP prior for the latent cluster centers. 

Appendix \ref{app:comparison-lamb} focuses on the natural competitor of our APPLAM model, namely the Lamb model by \cite{Chandra20}. 
Observe that \cite{Chandra20} already make a compelling argument in favor of latent mixture models compared to more heuristic approaches based on a two-step procedure. 
Therefore, we limit ourselves to showing the robustness of APPLAM to model misspecification and the Lamb model's lack thereof.
Specifically, we consider two simulation settings, where model misspecification occurs either at the latent level, i.e., the latent scores $\eta_i$ are simulated from a mixture of $t$ distributions, or at the likelihood level, i.e., the observations, conditional to all parameters, follow a $t$ distribution.
In both cases, for various choices of the latent dimension $d = \{4, 8\}$, data dimension $p = \{500, 1000\}$ and 100 iid replicates of the datasets, we demonstrate how the Lamb model significantly overestimates the number of clusters, contrary to our APPLAM model. Moreover, we show that APPLAM produces more reliable cluster estimates than Lamb.

In summary, our simulations confirm empirically the intuition that repulsive mixtures are more robust than non-repulsive ones in the presence of model misspecification. Moreover, they also justify the introduction of anisotropic DPPs presented here, as it is clear that forcing isotropic repulsion across cluster centers does not translate in well-separated clusters of data a posteriori. On the contrary, the isotropic models might greatly overestimate the number of clusters.

\section{Joint Species Modeling}

We apply our APPLAM model to data collecting the occurrence of plant species at different sites of the Bauges Natural Regional Park in France.
The data are available from the Alpine Botanical Conservatory (CBNA) and have been previously investigated by \cite{Bystrova21} and \cite{Thuiller18}. 
Specifically, the dataset we analyze refers to $n=1139$ sites and, for each site, the presence-absence of $p=123$ different plant species is reported.
Our goal is to infer the clustering structure of the sites, such that sites belonging to the same cluster show similar patterns concerning the occurrence of the plant species.
Note that the data are binary;
precisely, let $z_i \in \{0, 1\}^p$, $i=1, \ldots, n$,  such that $z_{i,j}= 1$ if the plant species $j$ occurs at the site $i$, $z_{i,j}=0$ otherwise. 
We assume $z_{i,j} = \indicator_{[0,\infty)}(y_{i,j})$ and apply the APPLAM model to the $y_i$'s, assuming the Gaussian-like DPP.
To avoid non-identifiability, we fix the $\Delta_h$'s in \eqref{eq:prior_marks} as the $d \times d$ identity matrices.
Posterior simulation requires minor modifications to the Gibbs sampling algorithm, and we discuss them in Appendix~\ref{sec:app_Gibbsxbinary}.

We rely on model selection to select the hyperparameters. In particular, we examine a total of 144 models where we vary the degree of repulsiveness by tuning $(\rho, c)$, as well as other hyperparameters. See Appendix \ref{app:bauges_hyper_elic} for more details. From Figure~\ref{fig:bauges_nclus_d3} and \ref{fig:bauges_nclus_d4},  it is clear that if $\rho|R|$ is large and/or the strength of repulsion $s$ is mild, the number of estimated clusters increases, which confirms our insights from the simulation studies. In particular, the number of estimated clusters ranges between 3 and 15; see also Tables \ref{tab:bauges_d3} and \ref{tab:bauges_d4}.
We focus only on models for which the posterior mean of the number of clusters is less than eight and select the best one by maximizing the WAIC  index \citep{Watanabe13}, which is standard in model selection. This yields the following choice of hyperparameters: $d=3$, $\rho|R| = 0.1$, $s = 0.9$, $\alpha = 10^{-3}$ and $a_\sigma = 2, b_\sigma = 1$.  Note that $\alpha$ small, as in this case here, implies a \emph{sparse} marginal prior for the mixture weights $\bm w$ in \eqref{eq:prior_marks} \citep{rousseau2011asymptotic}.

We discuss here the estimated clusters, obtained by minimizing the posterior expectation of Binder's loss function \citep{Bin(78)}, under this last hyperparameter selection. The estimated clusters are five,  with cardinalities between 200 and 250.
To interpret the estimated clusters, we look at the patterns of presence-absence of species in each cluster. Specifically, for $j=1, \ldots, p$, let $\Bar{z}_j$ and $\Bar{z}_{c,j}$  be the empirical frequency of the $j$-th species in the whole sample and in the $c$-th cluster, respectively.
For each cluster, we select six species that better represent it by choosing the three species that maximize (resp. minimize) $\delta_{c, j} = (\Bar{z}_{c,j} - \Bar{z}_j)$. In total, we find 20 species that best describe the different clusters. We report the corresponding $\delta_{c,j}$'s in Figure~\ref{fig:relevant_species_by_cluster}. Cluster~3 does not significantly depart from the sample averages, with $|\delta_{3,j}| < 0.1$ for all the selected species, but with an overall prevalence of lower than average species presence.
Cluster~4 features the most extreme deviations from the mean compared to the other clusters, especially higher than average deviations. Indeed, species 27 (the common rock-rose), 46 (the great yellow gentian), 97 (the sweet vernal grass) have a $\delta_{4,j}$ around or higher than 0.2. However, this cluster also records the most extreme lower-than-average deviations, with the presence of species 45 (the common beech) lower than  $0.15$ compared to the associated sample average and species 112 (the silver fir) lower than $0.1$. 
Cluster~5 differs from the other clusters mainly for the higher presence of species 6 (the ivy), 45 (the common beech), 83 (the sweet woodruff), 112 (the silver fir), 117 (the field maple), 121 (a blackberry species). The common beech helps explain the difference between cluster 4 and 5.
Finally, cluster~2 is mainly characterized by the higher presence of species 72 (the Saint Anthony's Turnip), while cluster~1 records a higher presence of species 28 (the common hazel). The bottom right panel of Figure~\ref{fig:relevant_species_by_cluster} reports the posterior expectation of $\Lambda \Lambda^\top$ restricted on the 20 species explaining the five estimated clusters.
A  non-negligible negative correlation value is estimated between species 72 (the Saint Anthony's Turnip) and 117 (the field maple): this can be observed in the empirical proportions of these two species within each of the five estimated clusters.

\begin{figure}[t]
    \centering 
\begin{subfigure}{0.33\textwidth}
  \includegraphics[width=\linewidth]{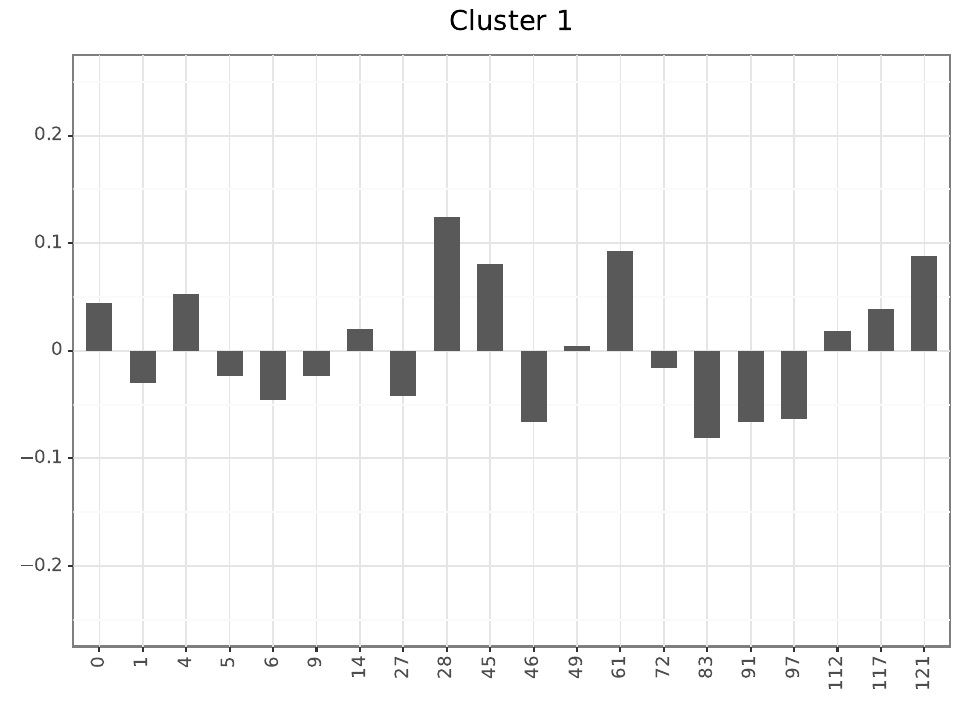}
\end{subfigure}\hfil 
\begin{subfigure}{0.33\textwidth}
  \includegraphics[width=\linewidth]{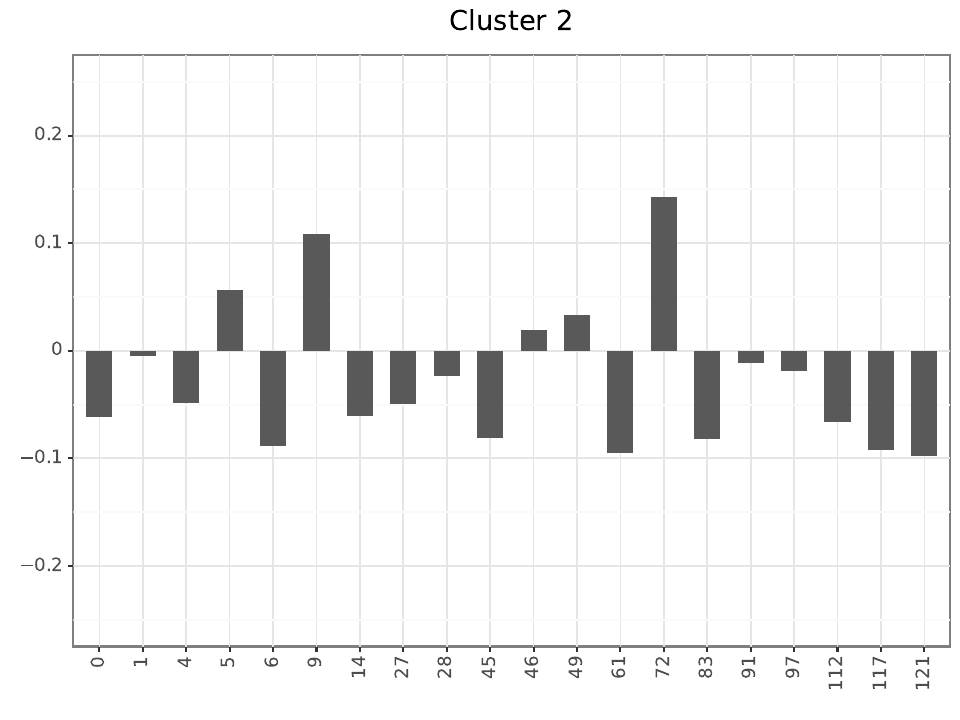}
\end{subfigure}\hfil 
\begin{subfigure}{0.33\textwidth}
  \includegraphics[width=\linewidth]{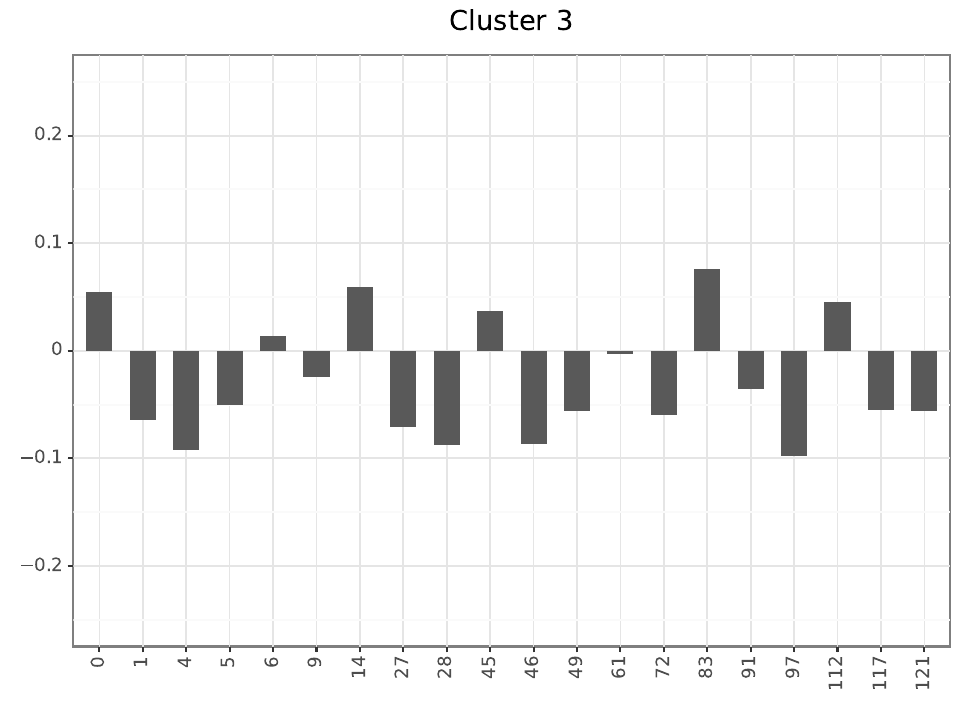}
\end{subfigure}

\medskip
\begin{subfigure}{0.33\textwidth}
  \includegraphics[width=\linewidth]{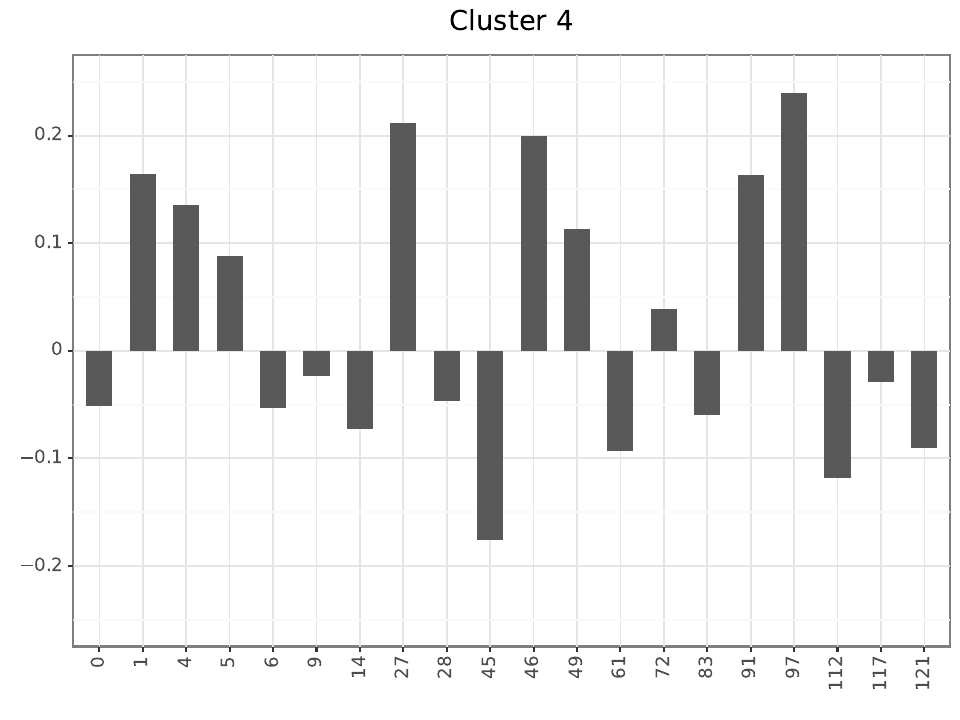}
\end{subfigure}\hfil 
\begin{subfigure}{0.33\textwidth}
  \includegraphics[width=\linewidth]{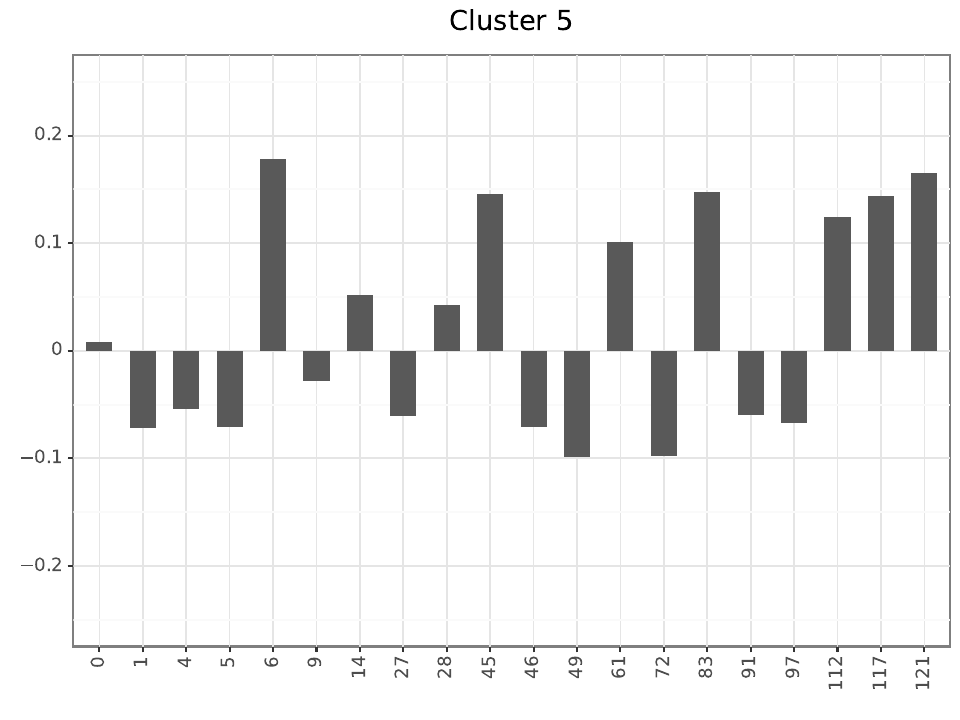}
\end{subfigure}\hfil 
\begin{subfigure}{0.33\textwidth}\centering
  \includegraphics[width=\linewidth]{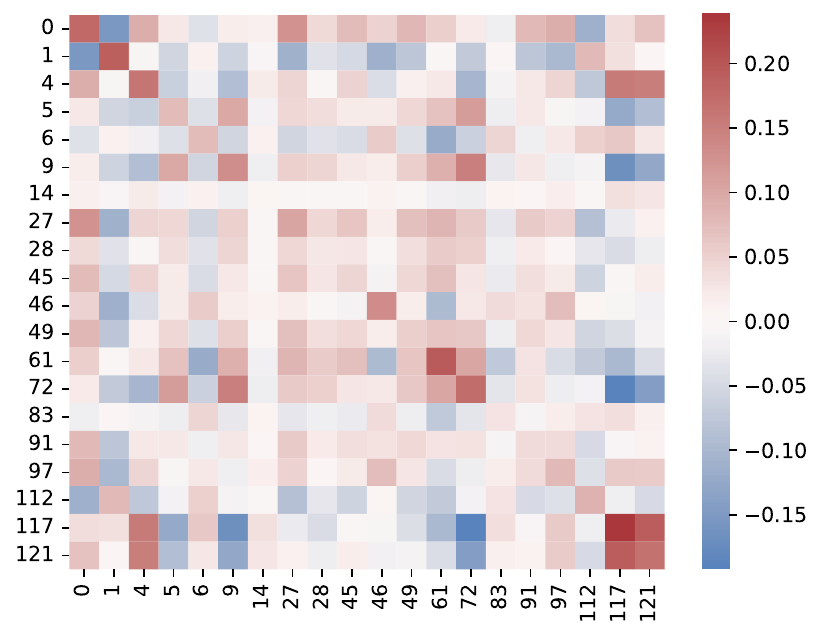}
\end{subfigure}
\caption{Values of $\delta_{c,j}$ for the 20 species identifying the five clusters and (bottom right) posterior estimate of $\Lambda \Lambda^\top$ restricted on the 20 species.
}
\label{fig:relevant_species_by_cluster}
\end{figure}

\section{Discussion}

 Model-based clustering of moderate or large dimensional data is notoriously difficult. 
 We have proposed a model for simultaneous dimensionality reduction and clustering, by
 assuming a mixture model for the latent scores, which are then linked to the observations via a
 Gaussian latent factor model. This approach was recently investigated by \cite{Chandra20}. The authors use a factor-analytic representation and assume a mixture model for the latent factors. However, performance can deteriorate in the presence of model misspecification. 
 Assuming a repulsive point process prior for the component-specific means of the mixture for the latent scores is shown to yield a more robust model that outperforms the standard Lamb on several simulated scenarios.
To favor well-separated clusters of data, the repulsive point process must be anisotropic, and its density should be tractable for efficient posterior inference. We address these issues by proposing a general construction for anisotropic determinantal point processes.

The bottleneck in our sampling algorithm is the spectral approximation of the DPP density, which has a computational cost that scales exponentially in $d$, the dimension of the latent factors. 
It is common practice to set small values for $d$ in latent factor models.
Nonetheless, for moderate or large values of $d$, the approach in \cite{bardenet} would be more efficient, at the price that the parameters in the model are harder to interpret. Another approximation of the DPP density is provided by \cite{poinas}, suitable also for non (hyper-)rectangular domains.
In the case of projection DPPs, i.e., when the eigenvalues of the spectral decomposition are either zero or one, we could sample exactly from the full-conditional of the non-allocated components using the methods in \cite{lav_rubak} instead of employing a birth-death Metropolis-Hastings move.

We have not investigated the estimation of the factor loading matrix $\Lambda$, which could be of interest to explain the correlation structure of the data. Our model inherits the rotational invariance property of classical factor models, but this issue can be dealt with via ex-post procedures for the estimation of $\Lambda$.  See, for instance, \cite{papasta2022} and the references therein.

\section*{Acknowledgements}
M.B. and A.G. acknowledge the support by MUR, grant Dipartimento di Eccellenza 2023-2027.
M.B. received funding from the European Research Council (ERC) under the European Union’s Horizon
2020 research and innovation programme under grant agreement No 817257. A.G. has been partially supported by MUR - Prin 2022 - Grant no. 2022CLTYP4, funded by the European Union –Next Generation EU.
We gratefully acknowledge the DataCloud laboratory (\url{https://datacloud.polimi.it}): some of the numerical experiments have been performed thanks
to the Cloud resources offered by the DataCloud laboratory.

\FloatBarrier

\bibliographystyle{biometrika}
\bibliography{references}

\appendix

\makeatletter
\renewcommand \thesection{S\@arabic\c@section}
\makeatother
\renewcommand{\theequation}{\thesection.\arabic{equation}}
\renewcommand{\thefigure}{S\arabic{figure}}
\renewcommand{\thetable}{S\arabic{table}}

\section{Further notation and proofs}

\subsection{General construction of DPPs}
\label{app:notation_proofs}

The $m$-th factorial measure of $\Phi$ is defined as
\[
\Phi^{(m)}(B_1\times \cdots \times B_m) = \sum_{\mu_1,\ldots,\mu_m \in \Phi}^{\neq} \indicator[\mu_1 \in B_1]  \cdots \indicator[\mu_m \in B_m]
\]
for measurable $B_1, \ldots, B_{m} \subset \R^d$. The summation is intended over all $m$-tuples of pairwise different points in $\Phi$ and it is standard notation.
The $m$-th factorial moment measure is  defined as  $M_{\Phi^{(m)}}(B_1\times \cdots \times B_m) = \E \left[ \Phi^{(m)}(B_1\times \cdots \times B_m) \right]$, where the expectation is taken with respect to $\Phi$.

In order to define a DPP, let $K: \R^d \times \R^d \rightarrow \mathbb{C}$ be a continuous covariance kernel.
Then $\Phi$ is a DPP on $\R^d$ if, for all $m=1,2,\ldots$, its $m$-th factorial moment measure has a density, with respect to the $m$-folded product of the $d$-dimensional Lebesgue measure, $\rho^{(m)}$, defined as
\[
    \rho^{(m)}(\mu_1, \ldots, \mu_m) = \det\{K(\mu_h, \mu_k)\}_{h, k = 1, \ldots, m}, \qquad \mu_1, \ldots, \mu_m \in \R^d.
\]
By Mercer's theorem $K(x, y) = \sum_{j \geq 1} \gamma_j \xi_j(x) \overline{\xi}_j(y)$ where the $\xi_j$'s form an orthonormal basis for $L^2(\R^d; \mathbb C)$ of complex-valued functions and the $\gamma_j$'s are a summable nonnegative sequence. 
Then, existence of a DPP with kernel $K$ is equivalent to $\gamma_j \leq 1$ for all $j$, see \cite{Macchi75}.
When restricted to a compact $R \subset \R^d$, $\Phi$ is still a DPP with kernel $K$ restricted to $R \times R$.
In particular, if $\gamma_j < 1$ for all $j$, $\Phi$ has a density with respect to the unit rate Poisson point process on $R$ given by
\begin{equation*}
    p(\{\mu_1, \ldots, \mu_m\}) = \e^{|R| - D} \det\{C(\mu_h, \mu_k)\}_{h, k = 1, \ldots, m}, \qquad \mu_1, \ldots, \mu_m \in R,
\end{equation*}
where $C(x, y) = \sum_{j \geq 1}\gamma_j / (1 - \gamma_j) \xi_j(x) \overline{\xi}_j(y)$, $|R| = \int_R \dd x$, and $D = - \sum_{j \geq 1} \log(1 - \gamma_j)$. See \cite{Lav15} for a proof of such results.

\subsection{Proof of Theorem \ref{teo:aniso_dpp}}
\begin{proof}
Since $\Lambda$ is full rank, $\Lambda^T \Lambda$ is positive definite and invertible, so that the conditional distribution (\ref{eq:model_Y}) for  $Y$ is well-defined. 
Denoting with $|\Lambda^T \Lambda|$ the determinant of $\Lambda^T \Lambda$, we explicitly compute the marginal density $h(x)$ of $Y$.
We have:
\begin{equation*}
h(x)=\int_0^{+\infty}{p(x\,|\,w)\cdot p(w) dw} = \int_0^{+\infty}{\frac{|\Lambda^T \Lambda|^{\frac{1}{2}}}{(2 \pi)^{\frac{d}{2}} w^{\frac{d}{2}}} \exp \biggl(-\frac{x^T \Lambda^T \Lambda x}{2w}\biggr) \cdot p(w) dw}
\end{equation*}
Therefore, we derive
\begin{equation}
h(x)= \frac{|\Lambda^T \Lambda|^{\frac{1}{2}}}{(2 \pi)^{\frac{d}{2}}} \, \E\biggl[W^{-\frac{d}{2}} \exp\biggl(-\frac{||\Lambda x||^2}{2W}\biggr)\biggr], \qquad x \in \R^d
\label{eq: density_Y}
\end{equation}
Consequently, since $\displaystyle K_0(x)=\rho h(x)/h(0)$, we have
\begin{equation*}
K_0(x)=\frac{\rho}{\E\bigl[W^{-\frac{d}{2}}\bigr]} \, \E\biggl[W^{-\frac{d}{2}} \exp\biggl(-\frac{||\Lambda x||^2}{2W}\biggr)\biggr], \qquad x \in \R^d
\end{equation*}
and, since $\varphi = \calF(K_0)$, we compute
\begin{equation*}
\begin{aligned}
	\varphi(x)&=\int_{\R^d}{\e^{-2\pi i x^T y}\, K_0(y) dy}=\\
	&= \frac{\rho}{h(0)} \int_{\R^d}{\e^{-2\pi i x^T y}\, \frac{|\Lambda^T \Lambda|^{\frac{1}{2}}}{(2 \pi)^{\frac{d}{2}}} \, \E\biggl[W^{-\frac{d}{2}} \exp\biggl(-\frac{||\Lambda y||^2}{2W}\biggr)\biggr] dy}=\\
	&=\frac{\rho \,|\Lambda^T \Lambda|^{\frac{1}{2}}}{h(0)\, (2 \pi)^{\frac{d}{2}}} \int_{\R^d}{\e^{-2\pi i x^T y} \int_0^{\infty}{w^{-\frac{d}{2}} \exp \biggl(-\frac{|| \Lambda y||^2}{2w}\biggr) p(w)dw}\,dy}=\\
	&=\frac{\rho \,|\Lambda^T \Lambda|^{\frac{1}{2}}}{h(0)\, (2 \pi)^{\frac{d}{2}}} \int_0^{\infty}{\int_{\R^d}{w^{-\frac{d}{2}} \exp\biggl(-2\pi i x^T y - \frac{y^T \Lambda^T \Lambda y}{2w}\biggr) dy} \,p(w)dw}=\\
&=\frac{\rho \,|\Lambda^T \Lambda|^{\frac{1}{2}}}{h(0)\, (2 \pi)^{\frac{d}{2}}} \int_0^{\infty}{\int_{\R^d}{w^{-\frac{d}{2}} \exp\biggl(-\frac{1}{2} \biggl[y^T \frac{\Lambda^T \Lambda}{w} y + 4\pi i x^T y \biggr]\biggr) dy}\, p(w)dw}
\end{aligned}
\end{equation*}
The term in the squared brackets above can be written as
\begin{equation*}
\begin{split}
\bigl[...\bigr]=&\bigl(y-(-2\pi i w (\Lambda^T \Lambda)^{-1} x)\bigr)^T \, \frac{\Lambda^T \Lambda}{w} \,\bigl(y-(-2\pi i w (\Lambda^T \Lambda)^{-1} x)\bigr)\\[1ex]
&+ 4\pi^2 x^T w(\Lambda^T \Lambda)^{-1} x \, .  
\end{split}
\end{equation*}
If we plug it in in the last expression of $\varphi(x)$ we have:
\begin{equation*}
\varphi(x)=\frac{\rho}{h(0)} \int_0^{+\infty}{\exp(-2\pi^2 w x^T (\Lambda^T \Lambda)^{-1} x)\, p(w)dw} \, .
\end{equation*}
Summing up, we derive
\begin{equation*}
\varphi(x)=\frac{\rho}{h(0)}\, \E\biggl[\exp\bigl(-2\pi^2 W x^T (\Lambda^T \Lambda)^{-1} x\bigr)\biggr], \qquad x \in \R^d \, .
\end{equation*}
Observe that, since $h(x)$ is the density of a real-valued random variable, using Fourier transform properties, then 
\[
	K_0(x)=\rho \,\frac{h(x)}{h(0)} \in L^1(\R^d), \qquad \varphi = \calF(K_0) \in L^1(\R^d) \, .
\]
Therefore, to guarantee the existence of the anisotropic determinantal point process with kernel $K_0$ and spectral density $\varphi = \calF(K_0)$, we refer to Corollary~(3.3) of \cite{Lav15}: since $K_0 \in L^1(\R^d), \varphi = \calF(K_0)$ and $\varphi \in L^1(\R^d)$, we just need to ensure 
\[
	\varphi(x) \leq 1, \qquad \forall x \in \R^d \, .  
\]
So, we need to assume 
\[
	\max_{x \in \R^d} \varphi(x)= \varphi(0) = \frac{\rho}{h(0)} \leq 1 \, .
\]
This assumption is equivalent to
\begin{equation*}
\rho \leq \rho_{max}
\end{equation*}
where
\begin{equation*}
\rho_{max} = h(0)=\frac{|\Lambda^T \Lambda|^{\frac{1}{2}}}{(2\pi)^{\frac{d}{2}}}\, \E\biggl[ W^{-\frac{d}{2}} \biggr].
\end{equation*}
\end{proof}

\subsection{Proof of Theorem \ref{teo:trans_dpp}}

Let $B := \Lambda \R^d$. Then $\Lambda: \R^d \rightarrow B$ is a bijective map. We perform a change of coordinates on $B$ using the orthonormal basis given by the singular value decomposition of  $\Lambda = U \Omega V^\top$ where $U$ is a $p \times d$ matrix whose columns are an orthonormal basis of $B$, $\Omega$ is a $d \times d$ diagonal matrix and $V$ is a $d \times d$ orthogonal matrix.
Let $\Lambda^* := \Omega V^\top$ which maps $\R^d$ into $B$ expressed with the coordinates given by $U$.

Consider now the transformed points $\tilde \mu_h = \Lambda \mu_h$. These are clearly in a one-to-one relation with the points $\mu^*_h = \Lambda^* \mu_h$, given by $\tilde \mu_h = U \mu^*_h$.
Then, it is sufficient to show that $\Phi^* = \{\mu^*_1, \ldots, \mu^*_m\}$ is an isotropic DPP on $B^* = \Lambda^* \R^d = \R^d$.
Consider a measurable function $h: B^{*k} \rightarrow \R_+$,
\begin{equation*}
    \E \sum_{\bm  {y}_k^* \in \Phi^*}^{\neq} h(y^*_1, \ldots, y^*_k) = \E \sum_{\bm  {x}_k \in \Phi}^{\neq} h(\Lambda^* x_1, \ldots, \Lambda^* x_k),
\end{equation*}
where the summation is intended over all the $k$-tuples of pairwise disjoint points in the support of the point process. 
By the definition of the DPP
\begin{align*}
    \E \sum_{\bm  {y}_k^* \in \Phi^*}^{\neq} h(y^*_1, \ldots, y^*_k)
    &= \E \sum_{\bm{x}_k \in \Phi}^{\neq} h(\Lambda^* x_1, \ldots, \Lambda^* x_k) \\
    &= \int_{(\R^d)^k} \det\{K(x_i, x_j)\}_{i, j = 1, \ldots, k} \, h(\Lambda^* x_1, \ldots, \Lambda^*x_k) \dd x_1 \cdots \dd x_k \\
    &= \int_{B^{*k}} \det\{K((\Lambda^*)^{-1} y^*_i, (\Lambda^*)^{-1} y^*_j)\}_{i, j = 1, \ldots, k} \, h(y^*_1, \ldots, y^*_k) \det(\Lambda^*)^{- k} \dd y^*_1 \cdots \dd y^*_k
\end{align*}
which shows that $\Lambda^* \mu$ is a DPP on $B^*$ with kernel given by $\det(\Lambda^*)^{-1} K((\Lambda^*)^{-1} y^*_i, (\Lambda^*)^{-1} y^*_j)$ and the dominating measure on $B^*$ is the Lebesgue measure. 
Stationarity and isotropy follow by choosing $K \equiv K_0$ as in Theorem~\ref{teo:aniso_dpp}. Indeed,
\begin{align*}
    K((\Lambda^*)^{-1} y^*_i, (\Lambda^*)^{-1} y^*_j) &= K_0((\Lambda^*)^{-1} (y^*_i - y^*_j) )\\
    &= \frac{\rho}{\E[W^{-d/2}]} \E \left[ W^{-d/2} \exp\left(- \frac{\|y^*_i - y^*_j\|^2}{2W}\right)\right]
\end{align*}

The result follows by noting that $U$ is orthogonal (i.e., the map between the $\tilde \mu_j$ and $\mu^*_j$ is an isometric bijection) and that $\det(\Lambda^*) = \det(\Lambda^\top \Lambda)^{1/2}$.

\subsection{Proof of Corollary \ref{cor:gauss_dpp}}
\begin{proof}
Let $c > 0$. In Theorem~\ref{teo:aniso_dpp}, set
\[
 W=|\Lambda^T \Lambda|^{\frac{1}{d}} \cdot c^{-\frac{2}{d}}
\]
Consequently, $W^{-\frac{d}{2}} = |\Lambda^T \Lambda|^{-\frac{1}{2}} \cdot c$, and from  \eqref{eq: density_Y}, we derive
\[
h(x)= \frac{c}{(2\pi)^{\frac{d}{2}}} \cdot \exp\biggl( -\frac{||\Lambda x||^2}{2|\Lambda^T \Lambda|^{\frac{1}{d}} \,c^{-\frac{2}{d}}} \biggr), \qquad x \in \R^d \, .
\]
From  \eqref{eq: kernel_Y} and \eqref{eq: spect_dens_Y}, we have
\begin{equation*}
\begin{split}
&K_0(x)=\rho \cdot \exp \biggl( -\frac{||\Lambda x||^2}{2|\Lambda^T \Lambda|^{\frac{1}{d}} \,c^{-\frac{2}{d}}} \biggr), \qquad x \in \R^d\\[1ex]
&\varphi(x)=\rho \,\frac{(2\pi)^{\frac{d}{2}}}{c} \cdot \exp \biggl( -2\pi^2 |\Lambda^T \Lambda|^{\frac{1}{d}} c^{-\frac{2}{d}} x^T (\Lambda^T \Lambda)^{-1} x \biggr), \qquad x \in \R^d \, .
\end{split}
\end{equation*}
From (\ref{eq: rho_max}), the existence condition requires $\rho \leq \rho_{max}$, with 
\begin{equation*}
\rho_{max} = \frac{c}{(2\pi)^{d/2}} \, .
\end{equation*}
\end{proof}

\subsection{Proof of Theorem~\ref{teo:grad_lambda}}
\begin{proof}
Consider 
\begin{equation}\label{eq_log_full_gauss}
\log p(\Lambda\,|\,\cdots) \propto \log p(\bm y\,|\,\Lambda, \bm\eta, \Sigma)+ \log f^{\text{app}}_{\text{DPP}}(\bm\mu \mid \Lambda) + \log p(\Lambda\,|\,\phi,\tau,\psi)
\end{equation}
We compute the gradient of log-full conditional
density of $\Lambda$ summing up the gradients of the three terms above.
Note that the only term depending on the anisotropic DPP is the second one. Since it is the most complex term, we derive it using the two lemmas below. 
For the first term of \eqref{eq_log_full_gauss},
we have
\begin{equation*}
\begin{aligned}
    \nabla \log p(\bm y\,|\,\Lambda,\bm \eta, \Sigma) &= \nabla \left( - \frac{1}{2} \sum_{i=1}^n (y_i - \Lambda \eta_i)^\top  \Sigma^{-1}(y_i - \Lambda \eta_i) \right)\\
    &=\Sigma^{-1} \cdot \sum\limits_{i=1}^n{(y_i-\Lambda\eta_i)\eta_i^T} \, . 
\end{aligned}
\end{equation*}
The gradient of the last term is
\begin{equation*}
    \nabla \log p(\Lambda\,|\,\phi,\tau,\psi) = - \frac{1}{(\psi \odot \phi^2) \tau^2} \odot \Lambda \, . 
\end{equation*}
The gradient of  the second term of \eqref{eq_log_full_gauss} is,
\begin{align}
     \nabla \log f^{\text{app}}_{\text{DPP}}(\bm\mu \mid \Lambda) &= \nabla \bigl[-D^{\mathrm{app}} - \log(1-\e^{-D^{\mathrm{app}}}) + \log \det[\Capp](T\mu_1,\ldots,T\mu_n)\bigr] \nonumber\\
 &= -\nabla D^{\mathrm{app}} - \nabla \log(1-\e^{-D^{\mathrm{app}}}) + \nabla \log \det[\Capp](T\mu_1,\ldots,T\mu_n) \nonumber\\
 &= -\frac{1}{1-\e^{-D^{\mathrm{app}}}} \nabla D^{\mathrm{app}}+ \nabla \log \det[\Capp](T\mu_1,\ldots,T\mu_n) \label{eq:secondterm}
\end{align}
To handle \eqref{eq:secondterm}, the terms $\nabla D^{\mathrm{app}}$ and $\nabla \log \det[\Capp](T\mu_1,\ldots,T\mu_n)$ are to be computed.

\begin{lemma}
\label{lem:lemma1}
    For the Gaussian-like DPP prior,
    \begin{equation}\label{grad_D_app}
        \nabla D^{\mathrm{app}} = \sum\limits_{k \in \mathbb{Z}_N^d}{\frac{\varphi(k)}{1-\varphi(k)}(-2 \pi^2 c^{-\tfrac2d})g^{(k)}} 
    \end{equation}
    where $\varphi(k)$ refers to \eqref{phi_gauss}.
\end{lemma}
\begin{proof}
    Write
    \begin{equation*}
        \nabla D^{\mathrm{app}} = - \sum\limits_{k \in \mathbb{Z}_N^d}{ \nabla \log (1 - \varphi(k))} =\sum\limits_{k \in \mathbb{Z}_N^d}{\frac{1}{1-\varphi(k)} \nabla \varphi(k)}.
    \end{equation*}
    For the Gaussian-like DPP prior, from \eqref{phi_gauss}, observe that, for $k \in \mathbb Z^d$
    \begin{equation*}
    \begin{aligned}
       \nabla \varphi(k) &= \frac{\rho}{c} (2\pi)^{\frac{d}{2}} \,\nabla \exp\left(  -2\pi^2 |\Lambda^T \Lambda|^{\frac{1}{d}} c^{-\frac{2}{d}} k^T (\Lambda^T \Lambda)^{-1} k \right)\\
       &= \varphi(k) \left( -2\pi^2 c^{-\frac{2}{d}} \right) \nabla \left( |\Lambda^T \Lambda|^{\frac{1}{d}} k^T (\Lambda^T \Lambda)^{-1} k \right) \, .
    \end{aligned}
    \end{equation*}
    Note that 
    \begin{equation*}
    \begin{aligned}
        \nabla \left( |\Lambda^T \Lambda|^{\frac{1}{d}} k^T (\Lambda^T \Lambda)^{-1} k \right) =& \nabla \left( |\Lambda^T \Lambda|^{\frac{1}{d}} \right) k^T (\Lambda^T \Lambda)^{-1} k +\\
        & + |\Lambda^T \Lambda|^{\frac{1}{d}} \nabla \left( k^T (\Lambda^T \Lambda)^{-1} k \right) = \\
        = & \frac{1}{d}|\Lambda^T \Lambda|^{\frac{1}{d}-1} 2|\Lambda^T \Lambda|\Lambda(\Lambda^T \Lambda)^{-1} k^T (\Lambda^T \Lambda)^{-1} k + \\
        & + |\Lambda^T \Lambda|^{\frac{1}{d}} \nabla \left( k^T (\Lambda^T \Lambda)^{-1} k \right)
    \end{aligned}
    \end{equation*}
    where, in the last step, formula (53) of \cite{matrixcook} is used. Then,
    \begin{equation*}
        \begin{aligned}
            \nabla \left( k^T (\Lambda^T \Lambda)^{-1} k \right) &= \nabla \mathrm{tr}\left( k^T (\Lambda^T \Lambda)^{-1} k \right) =\\
            &= \nabla \mathrm{tr}\left( (\Lambda^T \Lambda)^{-1} k k^T \right) = -\Lambda (\Lambda^T \Lambda)^{-1} (2k k^T) (\Lambda^T \Lambda)^{-1}
        \end{aligned}
    \end{equation*}
    where, in the last step, we have applied formula (125) of \cite{matrixcook}. For the Gaussian-like DPP prior, this leads to
    \begin{equation*}
    \nabla \varphi(k) = \varphi(k)(-4 \pi^2 c^{-\tfrac2d})|\Lambda^T \Lambda|^{\tfrac1d} \Lambda (\Lambda^T \Lambda)^{-1} 
    \Bigl[ \frac1d k^T (\Lambda^T \Lambda)^{-1} k \mathds{1}_d - k k^T (\Lambda^T \Lambda)^{-1}\Bigr]
    \end{equation*}
    where $\mathds{1}_d$ is $d\times d$ matrix of 1's.
    This concludes the proof of Lemma~\ref{lem:lemma1}.
\end{proof}

\begin{lemma}
\label{lem:lemma_2}
    For the Gaussian-like DPP prior,
    \begin{equation}\label{grad_ldcapp}
        \nabla \log \det[\Capp] = (-2 \pi^2 c^{-\tfrac2d})\sum\limits_{k \in \mathbb{Z}_N^d}{\frac{\varphi(k)}{(1-\varphi(k))^2} g^{(k)} v_k^T (\Capp)^{-1} u_k}
    \end{equation}
    where $\varphi(k)$ refers to \eqref{phi_gauss}.
\end{lemma}
\begin{proof}
We have
\begin{align}
        \derij \log \det[\Capp] &= \mathrm{tr}\left( \evalat[\Big]{\left(\frac{\partial}{\partial U} \log \det U\right)}{U=\Capp} \cdot \derij \Capp \right) \label{first_line_grad_ldc}\\
        &=  \mathrm{tr}\left( \frac{1}{\det[\Capp]} \evalat[\Big]{\left( \frac{\partial}{\partial U} \det U\right)}{U=\Capp} \cdot \derij \Capp \right) \nonumber\\
        &= \mathrm{tr}\left( (\Capp)^{-1} \cdot \derij \Capp \right) \label{last_line_grad_ldc}
    \end{align}
    where formula (137) of \cite{matrixcook} is used to get \eqref{first_line_grad_ldc} and formula (49) of \cite{matrixcook} for \eqref{last_line_grad_ldc}. Now, note that 
    \begin{equation*}
        \Capp = \sum\limits_{k \in \mathbb{Z}_N^d}{\frac{\varphi(k)}{1-\varphi(k)}u_k v_k^T}
    \end{equation*}
    where $u_k, v_k$ are defined in the statement. It follows that, for the Gaussian-like DPP prior,
    \begin{equation*}
    \derij \Capp = \sum\limits_{k \in \mathbb{Z}_N^d}{\frac{\varphi(k)}{(1-\varphi(k))^2}(-2 \pi^2 c^{-\tfrac2d}) g^{(k)}_{ij} u_k v^T_k}
    \end{equation*}
    where $g^{(k)}_{ij} = (g^{(k)})_{ij}$.
    Then, back to \eqref{last_line_grad_ldc}, for  the Gaussian-like DPP prior, we have
    \begin{align}
        \derij \log \det[\Capp] &= \mathrm{tr}\left( (\Capp)^{-1} \cdot \derij \Capp \right) \nonumber\\
        &= \mathrm{tr}\left( \sum\limits_{k \in \mathbb{Z}_N^d}{\frac{\varphi(k)}{(1-\varphi(k))^2}(-2 \pi^2 c^{-\tfrac2d}) g^{(k)}_{ij} (\Capp)^{-1} u_k v^T_k}\right) \nonumber\\
        &= \sum\limits_{k \in \mathbb{Z}_N^d}{\frac{\varphi(k)}{(1-\varphi(k))^2}(-2 \pi^2 c^{-\tfrac2d}) g^{(k)}_{ij} \,\mathrm{tr}\left( (\Capp)^{-1} u_k v^T_k\right)}\nonumber\\
        &= (-2 \pi^2 c^{-\tfrac2d})\sum\limits_{k \in \mathbb{Z}_N^d}{\frac{\varphi(k)}{(1-\varphi(k))^2} g^{(k)}_{ij} v_k^T (\Capp)^{-1} u_k}\, . \nonumber
    \end{align}
which concludes the proof of Lemma~\ref{lem:lemma_2}.
\end{proof}

Coming back to the proof of Theorem~\ref{teo:grad_lambda},  in the case of the Gaussian-like DPP prior, from \eqref{grad_D_app} and \eqref{grad_ldcapp}, Equation (\ref{eq:secondterm}) results into
\begin{equation*}
    \nabla \log f^{\text{app}}_{\text{DPP}}(\bm\mu \mid \Lambda) = (2 \pi^2 c^{-\tfrac2d})\sum\limits_{k \in \mathbb{Z}_N^d}{g^{(k)} \frac{\varphi(k)}{(1-\varphi(k))^2} \Bigl[ \frac{1-\varphi(k)}{1-\e^{-\Dapp}} -  v_k^T (\Capp)^{-1} u_k \Bigr] }
\end{equation*}
which concludes the proof of Theorem~\ref{teo:grad_lambda}.

\end{proof}

\section{Measure-Theoretic Details}\label{sec:app_measure}

From the discussion in the main text, we have that $\bm y \in \R^{p \times n}$, $\bm \eta \in \R^{d \times n}$, $\Sigma \in \R_+^p$, $\Lambda \in \R^{p \times d}$, $\bm \psi \in \R_+^{p \times d}$, $\bm \phi \in \mathbb S^{p \times d - 1}$ (the $pd -1$ dimensional simplex), and $\tau \in \R_+$.
Moreover, we consider $\bm \mu$ as a random point configuration, which takes values in $\Omega = \cup_{m=0}^\infty \Omega_m$ where $\Omega_m$ denotes the space of (pairwise distinct) $m$-uples of $\R^d$. We endow each $\Omega_m$ with the smallest $\sigma$-algebra which makes the following mapping measurable
\[
    (\mu_1, \ldots, \mu_m) \mapsto \{\mu_1, \ldots, \mu_m\} \, ,
\]
where on the left-hand side we see $\mu_1, \ldots, \mu_m$ as an ordered vector in $R^m$ and on the right-hand side as an unordered collection of points in $R$, where $R \subset \R^{d}$ is the hypersquare where $\mu$ is defined.
The $\sigma$-algebra on $\Omega$ is then the smallest $\sigma$-algebra containing the union of all the $\sigma$-algebras on each $\Omega_m$.
Then, it follows that 
$(\bm \mu, \bm s, \bm \Delta, \bm c) \in \cup_{m=0}^\infty \left\{ \Omega_m \times \R_+^m \times \mathcal{SP}_d^m \times \{1, \ldots, m\}^n \right\}$, where $\mathcal{SP}_d$ denotes the space of $d \times d$ symmetric and positive matrices.

Consider now sets $\mathfrak Y \subset \R^{p \times n}$, $\mathfrak N \subset R^{d \times n}$, $\Xi \subset \R_+^p$, $\Uplambda \subset \R^{p \times d}$,  $\Uppsi \subset \R_+^{p \times d}$, $\Upphi \subset \mathbb S^{p \times d - 1}$ $\mathfrak T \subset \R_+$, $\mathfrak O_m \subset \Omega_m$, $\mathfrak S_m \subset \R_+^m$, $\mathfrak D_m \subset \mathcal{SP}_d^m$, $\mathfrak C_m \subset \{1, \ldots, m\}^n$. The dominating measure for the joint distribution of data and parameters is
\begin{equation*}
    \begin{aligned}
        & \nu\left(\mathfrak Y \times \mathfrak N \times \times \Xi \times  \Uplambda \times \Uppsi \times  \Upphi \times \mathfrak T \times \cup_{m \geq 0} \left\{\mathfrak O_m \times \mathfrak S_m \times \mathfrak D_m \times \mathfrak C_m\right\}  \right) = \\
        & \quad \int_{\mathfrak Y} \dd \bm y \times  \int_{\mathfrak n} \dd \bm \eta \times 
        \int_{\Xi} \dd \bm \Sigma \times \int_{\Uplambda} \dd \Lambda \times \int_{\Uppsi} \dd \bm \psi \times \int_{\Upphi} \dd \bm \phi \times \int_{\mathfrak T} \dd \tau  \, \times \\
        & \qquad \qquad \times \sum_{m=0}^\infty \frac{\e^{- |R|}}{m!} \int_{\mathfrak O_m} \dd \mu_m \times \int_{\mathfrak S_m} \dd s_m \times \int_{\mathfrak D_m} \dd \Delta_m \times \sum_{c_1, \ldots, c_n = 1}^M \indicator[\bm c \in \mathfrak C_m] \, .
    \end{aligned}
\end{equation*}

The density of data and parameters with respect to $\nu$ is given by:
\begin{align*}
    & p(\bm y, \bm c, \bm \eta, \bm \mu, \bm S, \bm \Delta, \Sigma, \Lambda, \bm \psi, \bm \phi, \tau) = \\
    & \qquad \frac{1}{T^n}  \left[\prod_{i=1}^n\calN_p(y_i \mid \Lambda \eta_i, \Sigma)  \right] \left[ \prod_{h=1}^k (S_h^{(a)})^{n_h} \text{Ga}(S^{(a)}_h \mid \alpha, 1) \text{IW}(\Delta^{(a)}_h \mid \nu_0, \Psi_0) \prod_{i: c_i = h} \calN_d(\eta_i \mid \mu^{(a)}_h, \Delta^{(a)}_h) \right] \\
    & \qquad \left[\prod_{h=1}^\ell \text{Ga}(S^{(na)}_h \mid \alpha, 1) \text{IW}(\Delta^{(na)}_h \mid \nu_0, \Psi_0) \right] f^{\text{app}}_{\text{DPP}}(\bm \mu^{(a)} \cup \bm \mu^{(na)} \mid \rho, \Lambda, K_0; R) \\
    & \qquad \left[\prod_{j=1}^p \text{inv-Ga}(\sigma^2_j \mid a_\sigma, b_\sigma) \prod_{h=1}^d \calN(\lambda_{jh} \mid 0, \psi_{jh} \phi^2_{jh} \tau) \text{Exp}(\Psi_{jw} \mid 1/2) \right] \text{Dir}(vec(\phi) \mid a) \text{Ga}(\tau \mid pda, 1/2)
\end{align*}

We now introduce the auxiliary variable $u$ such that $u \mid T \sim \mbox{Ga}(n, t)$ and consider the extended parameter space including $u \in \R_+$. 
Moreover, conditional to $\bm c$ we split $\bm \mu = \mu^{(a)} \cup \mu^{(na)}$, $\bm S = [\bm S^{(a)}, \bm S^{(na)}]$ and $\bm \Delta = [\bm \Delta^{(a)}, \bm \Delta^{(na)}]$ into allocated and non-allocated components (denoted with the $(a)$ and $(na)$ superscript respectively).
The dominating measure $\nu^\prime$ on the extended space can be straightforwardly derived. See, for instance, Equation (17) in \cite{Ber22}.
The joint density of data and parameters with respect to $\nu^\prime$ is then
\begin{align*}
    & p(\bm y, \bm c, \bm \eta, \bm \mu^{(a)}, \bm \mu^{(na)}, \bm S^{(a)}, \bm S^{(na)} \bm \Delta^{(a)}, \bm \Delta^{(na)}, \Sigma, \Lambda, \bm \psi, \bm \phi, \tau, u)  \\
    & \qquad =\frac{u^{n-1}}{\Gamma(n)}  \left[\prod_{i=1}^n\calN_p(y_i \mid \Lambda \eta_i, \Sigma)  \right] \Bigg[ \prod_{h=1}^k \e^{-u S_h^{(a)}} (S_h^{(a)})^{n_h} \text{Ga}(S^{(a)}_h \mid \alpha, 1) \text{IW}(\Delta^{(a)}_h \mid \nu_0, \Psi_0) \\
    & \qquad \qquad \times \prod_{i: c_i = h} \calN_d(\eta_i \mid \mu^{(a)}_h, \Delta^{(a)}_h) \Bigg] 
    \left[\prod_{h=1}^\ell \e^{-u S_h^{(na)}} \text{Ga}(S^{(na)}_h \mid \alpha, 1) \text{IW}(\Delta^{(na)}_h \mid \nu_0, \Psi_0) \right] \\
    & \qquad f^{\text{app}}_{\text{DPP}}(\bm \mu^{(a)} \cup \bm \mu^{(na)} \mid \rho, \Lambda, K_0; R) \prod_{j=1}^p \Bigg[ \text{inv-Ga}(\sigma^2_j \mid a_\sigma, b_\sigma) \prod_{h=1}^d \calN(\lambda_{jh} \mid 0, \psi_{jh} \phi^2_{jh} \tau) \\
    & \qquad \qquad \times \prod_{h=1}^d \text{Exp}(\Psi_{jd} \mid 1/2) \Bigg] \text{Dir}(vec(\phi) \mid a) \text{Ga}(\tau \mid pda, 1/2) \, .
\end{align*}

\section{The Anisotropic Whittle-Mat\'ern DPP}\label{sec:app_whmat}

\begin{corollary}\label{cor:matern_dpp}
Using the same notation of Theorem~\ref{teo:aniso_dpp}, let 
\[
    W \sim \mathrm{Gamma}\biggl(\nu+\frac{d}{2}, \frac{1}{2|\Lambda^T \Lambda|^{\frac{1}{d}}\alpha^2}\biggr), \qquad \nu, \alpha > 0,
\]
where $\Lambda$ is fixed. Then the kernel $K_0$, its Fourier transform $\varphi = \mathcal{F}(K_0)$ and $\rho_{\max}$ 
follow here:
\begin{align}
K_0(x) & = \rho\, \frac{2^{1-\nu}}{\Gamma(\nu)} \norm{\frac{\Lambda x}{\alpha \, |\Lambda^T \Lambda|^{\frac{1}{2d}}}}^{\,\nu} K_\nu \left( \norm{\frac{\Lambda x}{\alpha \,|\Lambda^T \Lambda|^{\frac{1}{2d}}}} \right), \qquad x \in \R^d \nonumber\\
\varphi(x) & =\rho\, \frac{ \Gamma(\nu+\frac{d}{2})}{\Gamma(\nu)} \,\frac{(2\sqrt{\pi} \alpha )^d}{\bigl(1+4\pi^2 \alpha^2 |\Lambda^T \Lambda|^{\frac{1}{d}}\, x^T {(\Lambda^T \Lambda)}^{-1} x\bigr)^{\nu+\frac{d}{2}} }\, , \qquad x \in \R^d \label{phi_whittle}\\
\rho_{\max} &= \frac{\Gamma(\nu)}{\Gamma(\nu+\frac{d}{2}) \,(2\sqrt{\pi} \alpha )^d \nonumber}
\end{align} 
where $K_{\nu}$ is the modified Bessel function of the second kind.
\end{corollary}
\begin{proof}
Let $\nu>0$ and $\alpha>0$. In Theorem~\ref{teo:aniso_dpp}, set
\[
W \sim \mathrm{Gamma}\Biggl(\nu+\frac{d}{2}, \frac{1}{2|\Lambda^T \Lambda|^{\frac{1}{d}}\alpha^2}\,\Biggr) \, .
\]
Applying \eqref{eq: density_Y}, we explicitly compute $h(x)$ as 
\begin{equation*}
\begin{split}
h(x)&=\frac{|\Lambda^T \Lambda|^{\frac{1}{2}}}{(2 \pi)^{\frac{d}{2}}} \int_0^{+\infty}{\Biggl[w^{-\frac{d}{2}} \exp\Biggl( -\frac{||\Lambda x||^2}{2w} \Biggr) \Biggl( \frac{1}{2|\Lambda^T \Lambda|^{\frac{1}{d}}\alpha^2} \Biggr)^{\nu+\frac{d}{2}} \frac{w^{\nu+\frac{d}{2} -1}}{\Gamma\bigl(\nu+\frac{d}{2}\bigr)}} \\[1ex]
&\hspace{7.7cm}\times \exp \Biggl( -\frac{w}{2|\Lambda^T \Lambda|^{\frac{1}{d}}\alpha^2} \Biggr)\Biggr] \,dw \\[1ex]
&=\frac{|\Lambda^T \Lambda|^{-\frac{\nu}{d}}}{(2 \pi)^{\frac{d}{2}} (2\alpha^2)^{\nu+\frac{d}{2}} \Gamma\bigl(\nu+\frac{d}{2}\bigr)} \int_0^{+\infty}{w^{\nu-1} \exp\Biggl[ -\frac{1}{2}\Biggl(\frac{1}{|\Lambda^T \Lambda|^{\frac{1}{d}}\alpha^2} w + \frac{||\Lambda x||^2}{w} \Biggr) \Biggr] \,dw} \, .
\end{split}
\end{equation*}
In this paper, by generalized inverse Gaussian distribution (GIG) on $[0,+\infty)$ with parameters $(p, a, b)$ we mean an absolutely continuous distribution with density 
\begin{equation}
f(x)=\frac{(a/b)^\frac{p}{2}}{2K_p(\,\sqrt{ab}\,)} x^{p-1} \exp\biggl(-\frac{1}{2}\biggl[ax+\frac{b}{x}\biggr]\biggr),\; x>0
\label{eq: GIG_density}
\end{equation}
where $p \in \R$, $a>0$ and $b>0$. In the above expression $K_p$ is the modified Bessel function of the second kind. We indicate such a distribution with $\mathrm{giG}(p,a,b)$.

Coming back to the computation of $h(x)$, we see in the last integral that the integrand is the density of a generalized inverse Gaussian distribution with parameters
\[
	p=\nu,\qquad a=\frac{1}{|\Lambda^T \Lambda|^{\frac{1}{d}}\alpha^2}, \qquad b=||\Lambda x||^2.
\]
Therefore, 
\begin{equation*}
\begin{aligned}
    h(x)&=\frac{|\Lambda^T \Lambda|^{-\frac{\nu}{d}}}{(2 \pi)^{\frac{d}{2}} (2\alpha^2)^{\nu+\frac{d}{2}} \,\Gamma\bigl(\nu+\frac{d}{2}\bigr)} \cdot \biggl(|\Lambda^T \Lambda|^{\frac{1}{d}}\alpha^2 \,||\Lambda x||^2\biggr)^{\frac{\nu}{2}} \cdot 2\,K_\nu \Biggl(\frac{||\Lambda x||}{|\Lambda^T \Lambda|^{\frac{1}{2d}}\, \alpha}\Biggr)\\
    &= \frac{1}{(\sqrt{\pi} \,\alpha )^d \,\,2^{\nu +d-1} \,\,\Gamma\bigl(\nu +\frac{d}{2}\bigr)} \cdot \norm{\frac{\Lambda x}{\alpha\, |\Lambda^T \Lambda|^{\frac{1}{2d}}}}^{\,\nu} \cdot K_\nu \Biggl( \,\norm{\frac{\Lambda x}{\alpha\, |\Lambda^T \Lambda|^{\frac{1}{2d}}}}\, \Biggr) \, .
\end{aligned}
\end{equation*}
Note that, as $x \rightarrow 0$, then $x^{\nu} K_{\nu}(x) \rightarrow 2^{\nu -1} \Gamma(\nu)$. Consequently, the kernel $K_0$ is
\begin{equation*}
K_0(x)=\rho\, \frac{2^{1-\nu}}{\Gamma(\nu)} \cdot \norm{\frac{\Lambda x}{\alpha \, |\Lambda^T \Lambda|^{\frac{1}{2d}}}}^{\,\nu} \cdot K_\nu \Biggl(\, \norm{\frac{\Lambda x}{\alpha \,|\Lambda^T \Lambda|^{\frac{1}{2d}}}} \,\Biggr) \, .
\end{equation*}
To derive the spectral density $\varphi = \calF(K_0)$, from \eqref{eq: spect_dens_Y} we have
\begin{equation*}
\begin{aligned}
\varphi(x)&=\frac{\rho}{h(0)} \int_0^{\infty}{\exp\bigl(-2\pi^2 w x^T (\Lambda^T \Lambda)^{-1} x\bigr) \biggl( \frac{1}{2|\Lambda^T \Lambda|^{\frac{1}{d}}\alpha^2} \biggr)^{\nu+\frac{d}{2}} \frac{w^{\nu+\frac{d}{2} -1}}{\Gamma\bigl(\nu+\frac{d}{2}\bigr)} \exp \biggl( -\frac{w}{2|\Lambda^T \Lambda|^{\frac{1}{d}}\alpha^2} \biggr) \,dw} \\
&=\frac{\rho}{h(0)\,\bigl(2|\Lambda^T \Lambda|^{\frac{1}{d}} \alpha^2\bigr)^{\nu +\frac{d}{2}}} \int_0^{\infty}{\frac{w^{\nu +\frac{d}{2}-1}}{\Gamma\bigl(\nu+\frac{d}{2}\bigr)}\, \exp\Biggl[-\Biggl(2\pi^2 x^T (\Lambda^T \Lambda)^{-1} x + \frac{1}{2|\Lambda^T \Lambda|^{\frac{1}{d}}\alpha^2}\Biggr) \,w\Biggr] \,dw}\\
&=\frac{\rho}{h(0)\,\bigl(2|\Lambda^T \Lambda|^{\frac{1}{d}} \alpha^2\bigr)^{\nu +\frac{d}{2}}} \cdot \frac{1}{\biggl(2\pi^2 x^T {(\Lambda^T \Lambda)}^{-1} x + \frac{1}{2|\Lambda^T \Lambda|^{1/d}\,\alpha^2}\biggr)^{\nu+\frac{d}{2}}}\\
&= \frac{\rho}{h(0)} \cdot \frac{1}{\bigl(1+4\pi^2 \alpha^2 |\Lambda^T \Lambda|^{\frac{1}{d}}\, x^T {(\Lambda^T \Lambda)}^{-1} x\bigr)^{\nu+\frac{d}{2}}} \, , 
\end{aligned}
\end{equation*}
so that the spectral density $\varphi = \calF(K_0)$ is
\begin{equation*}
\varphi(x)=\rho\, \frac{ \Gamma\bigl(\nu+\frac{d}{2}\bigr)}{\Gamma(\nu)} \,\frac{(2\sqrt{\pi} \alpha )^d}{\bigl(1+4\pi^2 \alpha^2 |\Lambda^T \Lambda|^{\frac{1}{d}}\, x^T {(\Lambda^T \Lambda)}^{-1} x\bigr)^{\nu+\frac{d}{2}} } \, .
\end{equation*}
\vspace{4ex}
To sum up, we report the kernel $K_0$ and the spectral density $\varphi$ of this model
\begin{equation*}
\begin{split}
&K_0(x)=\rho\, \frac{2^{1-\nu}}{\Gamma(\nu)} \cdot \norm{\frac{\Lambda x}{\alpha \, |\Lambda^T \Lambda|^{\frac{1}{2d}}}}^{\,\nu} \cdot K_\nu \Biggl(\, \norm{\frac{\Lambda x}{\alpha \,|\Lambda^T \Lambda|^{\frac{1}{2d}}}} \,\Biggr)\\[2ex]
&\varphi(x)=\rho\, \frac{ \Gamma\bigl(\nu+\frac{d}{2}\bigr)}{\Gamma(\nu)} \,\frac{(2\sqrt{\pi} \alpha )^d}{\bigl(1+4\pi^2 \alpha^2 |\Lambda^T \Lambda|^{\frac{1}{d}}\, x^T {(\Lambda^T \Lambda)}^{-1} x\bigr)^{\nu+\frac{d}{2}} } \, .
\end{split}
\end{equation*}
The existence condition, from Equation (\ref{eq: rho_max}), requires $\rho \leq \rho_{max}$, with 
\begin{equation*}
\rho_{max}=\frac{\Gamma(\nu)}{\Gamma\bigl(\nu+\frac{d}{2}\bigr) \,(2\sqrt{\pi} \alpha )^d } \, .
\end{equation*}
\end{proof}

We now show how to compute analytically the gradient of the full conditional for $\Lambda$ also in the case of Whittle-Mat\'ern DPP prior.
\begin{theorem}\label{teo:grad_lambda_whmat}
Under the Whittle-Mat\'ern-like DPP prior, the gradient of the log-full conditional density of $\Lambda$ equals
\begin{equation*}
    \begin{aligned}
        \nabla \log p(\Lambda\,|\,\cdots) & = \Sigma^{-1} \sum_{i=1}^n{(y_i-\Lambda\eta_i)\eta_i^\top} +
        \\
        & \qquad + 4\pi^2 \alpha^2 \left(\nu + \frac{d}{2}\right)\sum\limits_{k \in \mathbb{Z}_N^d}{\frac{g^{(k)}}{a^{(k)}} \frac{\varphi(k)}{(1-\varphi(k))^2} \Bigl[ \frac{1-\varphi(k)}{1-\e^{-D^{\mathrm{app}}}} -  v_k^T \left( \Capp\right)^{-1} u_k \Bigr] } + \\
        & \qquad - \frac{1}{(\psi \odot \phi^2) \tau^2} \odot \Lambda 
    \end{aligned}
\end{equation*}
where $\varphi(k)$  is defined in (\ref{phi_whittle}), $\odot$ denotes the elementwise (Hadamard) product, 
\[
    g^{(k)} = 2|\Lambda^T \Lambda|^{\tfrac1d} \Lambda (\Lambda^T \Lambda)^{-1} \Bigl[ \frac1d k^T (\Lambda^T \Lambda)^{-1} k \mathds{1}_d - k ((\Lambda^T \Lambda)^{-1} k)^T \Bigr],
\]
$u_k$ and $v_k$ are $m$-dimensional column vectors for each $k \in \mathbb{Z}^d$ with entries
\begin{equation*}
    (u_k)_j = \e^{2\pi i k^T T\mu_j} , \quad (v_k)_j = \e^{-2\pi i k^T T\mu_j} , \quad j=1,\ldots,m
\end{equation*}
$\Capp := \Capp(T\mu_1,\ldots,T\mu_n)$, and 
$a^{(k)} = 1+ 4\pi^2\alpha^2 |\Lambda^T \Lambda|^{\frac1d} k^T (\Lambda^T \Lambda)^{-1} k$. 
\end{theorem}
\begin{proof}

The proof follows closely the one of Theorem~\ref{teo:grad_lambda}. Instead of Lemma \ref{lem:lemma1} and Lemma \ref{lem:lemma_2}, we use the two following results

\begin{lemma}
\label{lem:lemma1_whmat}
  For the Whittle-Mat\'ern-like DPP prior,
    \begin{equation}\label{grad_D_app_wm}
        \nabla D^{\mathrm{app}} = \sum\limits_{k \in \mathbb{Z}_N^d}{\frac{\varphi(k)}{1-\varphi(k)}4\pi^2 \alpha^2\left(-\nu - \frac{d}{2}\right)\frac{g^{(k)}}{a^{(k)}}} 
    \end{equation}
    where $\varphi(k)$ refers to (\ref{phi_whittle}).
\end{lemma}
\begin{proof}
    Write
    \begin{equation*}
        \nabla D^{\mathrm{app}} = - \sum\limits_{k \in \mathbb{Z}_N^d}{ \nabla \log (1 - \varphi(k))} =\sum\limits_{k \in \mathbb{Z}_N^d}{\frac{1}{1-\varphi(k)} \nabla \varphi(k)}.
    \end{equation*}
     For the Whittle-Mat\'ern-like DPP prior, from \eqref{phi_whittle}, observe that, for $k \in \mathbb Z^d$
    \begin{equation*}
    \begin{aligned}
       \nabla \varphi(k) &= \rho\, \frac{ \Gamma(\nu+\frac{d}{2})}{\Gamma(\nu)} \,(2\sqrt{\pi} \alpha )^d \nabla \left(\bigl(1+4\pi^2 \alpha^2 |\Lambda^T \Lambda|^{\frac{1}{d}}\, k^T {(\Lambda^T \Lambda)}^{-1} k\bigr)^{-\nu-\frac{d}{2}}\right)\\
       &= \frac{\varphi(k)}{a^{(k)}} 4\pi^2\alpha^2\left(-\nu-\frac{d}{2}\right)  \nabla \left( |\Lambda^T \Lambda|^{\frac{1}{d}} k^T (\Lambda^T \Lambda)^{-1} k \right) \, .
    \end{aligned}
    \end{equation*}
    Note that 
    \begin{equation*}
    \begin{aligned}
        \nabla \left( |\Lambda^T \Lambda|^{\frac{1}{d}} k^T (\Lambda^T \Lambda)^{-1} k \right) =& \nabla \left( |\Lambda^T \Lambda|^{\frac{1}{d}} \right) k^T (\Lambda^T \Lambda)^{-1} k +\\
        & + |\Lambda^T \Lambda|^{\frac{1}{d}} \nabla \left( k^T (\Lambda^T \Lambda)^{-1} k \right) = \\
        = & \frac{1}{d}|\Lambda^T \Lambda|^{\frac{1}{d}-1} 2|\Lambda^T \Lambda|\Lambda(\Lambda^T \Lambda)^{-1} k^T (\Lambda^T \Lambda)^{-1} k + \\
        & + |\Lambda^T \Lambda|^{\frac{1}{d}} \nabla \left( k^T (\Lambda^T \Lambda)^{-1} k \right)
    \end{aligned}
    \end{equation*}
    where, in the last step, formula (53) of \cite{matrixcook} is used. Then,
    \begin{equation*}
        \begin{aligned}
            \nabla \left( k^T (\Lambda^T \Lambda)^{-1} k \right) &= \nabla \mathrm{tr}\left( k^T (\Lambda^T \Lambda)^{-1} k \right) =\\
            &= \nabla \mathrm{tr}\left( (\Lambda^T \Lambda)^{-1} k k^T \right) = -\Lambda (\Lambda^T \Lambda)^{-1} (2k k^T) (\Lambda^T \Lambda)^{-1}
        \end{aligned}
    \end{equation*}
    where, in the last step, {we have applied} formula (125) of \cite{matrixcook}.
    For the Whittle-Mat\'ern-like DPP prior, it follows
    \begin{equation*}
    \nabla \varphi(k) = \frac{\varphi(k)}{a^{(k)}}\, 8\pi^2\alpha^2\left(-\nu-\frac{d}{2}\right)|\Lambda^T \Lambda|^{\tfrac1d} \Lambda (\Lambda^T \Lambda)^{-1} 
    \Bigl[ \frac1d k^T (\Lambda^T \Lambda)^{-1} k \mathds{1}_d - k k^T (\Lambda^T \Lambda)^{-1}\Bigr]
    \end{equation*}
    where $\mathds{1}_d$ is $d\times d$ matrix of 1's.
    This concludes the proof of Lemma~\ref{lem:lemma1_whmat}.
\end{proof}

\begin{lemma}
\label{lem:lemma_2_whmat}
For the Whittle-Mat\'ern-like DPP prior,
    \begin{equation}\label{grad_ldcapp_wm}
        \nabla \log \det[\Capp] = 4\pi^2\alpha^2\left(-\nu-\frac{d}{2}\right) \sum\limits_{k \in \mathbb{Z}_N^d}{\frac{\varphi(k)}{(1-\varphi(k))^2} \frac{g^{(k)}}{a^{(k)}} v_k^T (\Capp)^{-1} u_k}
    \end{equation}
    where $\varphi(k)$ refers to (\ref{phi_whittle}).
\end{lemma}
\begin{proof}
Proceeding as in the proof of Lemma \ref{lem:lemma_2_whmat}, we have 
\begin{align}
        \derij \log \det[\Capp] &= \mathrm{tr}\left( \evalat[\Big]{\left(\frac{\partial}{\partial U} \log \det U\right)}{U=\Capp} \cdot \derij \Capp \right) \label{first_line_grad_ldc_wm}\\
        &=  \mathrm{tr}\left( \frac{1}{\det[\Capp]} \evalat[\Big]{\left( \frac{\partial}{\partial U} \det U\right)}{U=\Capp} \cdot \derij \Capp \right) \nonumber\\
        &= \mathrm{tr}\left( (\Capp)^{-1} \cdot \derij \Capp \right) \label{last_line_grad_ldc_wm}
    \end{align}
    where formula (137) of \cite{matrixcook} is used to get \eqref{first_line_grad_ldc_wm} and formula (49) of \cite{matrixcook} for \eqref{last_line_grad_ldc_wm}. Now, note that 
    \begin{equation*}
        \Capp = \sum\limits_{k \in \mathbb{Z}_N^d}{\frac{\varphi(k)}{1-\varphi(k)}u_k v_k^T}
    \end{equation*}
    where $u_k, v_k$ are defined in the statement. It follows that,for the Whittle-Mat\'ern-like DPP prior,
    \begin{equation*}
    \derij \Capp = \sum\limits_{k \in \mathbb{Z}_N^d}{\frac{\varphi(k)}{(1-\varphi(k))^2}4\pi^2\alpha^2\left(-\nu-\frac{d}{2}\right) \frac{g^{(k)}_{ij}}{a^{(k)}} u_k v^T_k} \, . 
    \end{equation*}
    Then, back to Equation \eqref{last_line_grad_ldc_wm}, for the Whittle-Mat\'ern-like DPP prior,
\begin{equation*}
        \derij \log \det[\Capp] = 4\pi^2\alpha^2\left(-\nu-\frac{d}{2}\right)\sum\limits_{k \in \mathbb{Z}_N^d}{\frac{\varphi(k)}{(1-\varphi(k))^2} \frac{g^{(k)}_{ij}}{a^{(k)}} v_k^T (\Capp)^{-1} u_k}
\end{equation*}
which concludes the proof of Lemma~\ref{lem:lemma_2_whmat}
\end{proof}

Coming back to the proof of Theorem~\ref{teo:grad_lambda_whmat},  in the case of the Whittle-Mat\'ern-like DPP prior, from  \eqref{grad_D_app_wm} and \eqref{grad_ldcapp_wm}, Equation \eqref{eq:secondterm} results into
\begin{equation*}
    \nabla \log f^{\text{app}}_{\text{DPP}}(\bm\mu \mid \Lambda) = 4\pi^2\alpha^2\left(\nu+\frac{d}{2}\right)\sum\limits_{k \in \mathbb{Z}_N^d}{\frac{g^{(k)}}{a^{(k)}} \frac{\varphi(k)}{(1-\varphi(k))^2} \Bigl[ \frac{1-\varphi(k)}{1-\e^{-\Dapp}} -  v_k^T (\Capp)^{-1} u_k \Bigr] }
\end{equation*}
which concludes the proof of Theorem~\ref{teo:grad_lambda_whmat}.
\end{proof}

\section{Further details on the MCMC algorithm}\label{sec:app_mcmc}

\subsection{The Gibbs Sampler}
\label{sec:app_Gibbs}

Here below we summarize the steps of our Gibbs sampler, though the calculation
of the full conditionals is omitted when it is straightforward. We denote by $X|\cdots$ the full conditional of the block $X$ of parameters, i.e., the conditional distribution of $X$, given all the other parameters. 

\begin{enumerate}
	\item \textit{Update of} $(\psi,\tau,\phi)$. Following \cite{bhattacharya2015dirichlet} sample
	\begin{align*}
	   \psi_{jh}\,|\,\lambda_{jh},\phi_{jh},\tau & \ind \mathrm{giG}\biggl(\frac{1}{2},1,\frac{\lambda_{jh}^2}{\phi_{jh}^2\tau^2}\biggr) \qquad  \\
	  \tau\,|\,\phi,\Lambda & \sim \mathrm{giG}\Biggl(p\cdot d\cdot (a-1),1,2\sum_{\substack{j=1:p\\h=1:d}}{\frac{|\lambda_{jh}|}{\phi_{jh}}}\Biggr) \\
	  \phi_{jh} = \frac{T_{jh}}{T}, \quad T_{jh} &\ind \mathrm{giG}(\,a-1,1,2|\lambda_{jh}|\,), \quad \displaystyle T:=\sum_{j, h}{T_{jh}}, 
	\end{align*}
	for $j=1,\ldots,p$, $h=1,\ldots,d$, where $\mathrm{giG}$ denotes the generalized inverse-Gaussian distribution (see \eqref{eq: GIG_density}).
	
    \item \textit{Update of} $\Lambda$. Sample from the full conditional density
    \begin{equation*}
        p(\Lambda\,|\,\cdots) \propto p(\bm y\,|\,\Lambda,\bm \eta, \Sigma) p(\Lambda\,|\,\phi,\tau,\psi) f^{\text{app}}_{\text{DPP}}(\bm \mu \mid \rho, \Lambda, K_0; R)
    \end{equation*}
    using a Metropolis-Hastings step; see Section \ref{sec:mala_lambda} for further details.
	
	\item \textit{Update of} $\Sigma$. Sample each $\sigma^2_j$ independently from
	\[
	  \sigma_j^2\,|\,y^{(j)},\bm\eta,\lambda^{(j)} \ind \mathrm{inv-Gamma}\biggl(\frac{n}{2}+a_\sigma, \frac{1}{2} \sum_{i=1}^n {\biggl(y^{(j)}_i- {\lambda^{(j)}}^T\eta_i\biggr)^2} +b_\sigma \biggr)
	\]
	where ${\lambda^{(j)}}^\top$ is the $j$-th row of $\Lambda$ and $y^{(j)} = (y_{1j},\ldots,y_{nj})^\top$.
	
	\item \textit{Update of the non-allocated variables} $(\bm\mu^{(na)}, \bm s^{(na)}, \bm\Delta^{(na)})$. Following \cite{Ber22}, we disintegrate the joint full conditional of the non-allocated variables as 
	\[  
	    p(\bm\mu^{(na)}, \bm s^{(na)}, \bm\Delta^{(na)} \mid \text{rest}) = p(\bm\mu^{(na)} \mid \text{rest}) p(\bm s^{(na)} \mid \bm\mu^{(na)}, \text{rest}) p(\bm \Delta^{(na)} \mid \bm\mu^{(na)}, \text{rest}),
	\]
	where ``rest'' identifies all the variables except for $(\bm\mu^{(na)}, \bm s^{(na)}, \bm\Delta^{(na)})$.
	Then $\bm\mu^{(na)} \mid \text{rest}$ is a Gibbs point process with density
	\[
	    p(\{\mu^{(na)}_1, \ldots, \mu^{(na)}_\ell \} \mid \text{rest}) \propto f^{\text{app}}_{\text{DPP}}(\{\mu^{(na)}_1, \ldots, \mu^{(na)}_\ell \}  \cup \bm \mu^{(a)} \mid \rho, \Lambda, K_0; R) \psi(u)^{\ell}
	\]
	where $\psi(u) = \E[\e^{- u S}]$. We employ the birth-death Metropolis-Hastings algorithm in \cite{geyer1994simulation} to sample from this point process density.
	Given $\bm \mu^{(na)}$ it is straightforward  to show
	\begin{align*}
	    \Delta^{(na)}_1, \ldots, \Delta^{(na)}_\ell \mid \cdots  &\iid \mathrm{IW}_d(\nu_0, \Psi_0) \\
	    S^{(na)}_1, \ldots, S^{(na)}_\ell \mid \cdots &\iid \text{Gamma}(\alpha, 1 + u)
	\end{align*}
	
	\item \textit{Update of the allocated variables} $(\bm\mu^{(a)}, \bm s^{(a)}, \bm\Delta^{(a)})$. Let $k$ be the number of unique values in the allocation vector $\bm c$ and assume that the active components are the first $k$.  We can sample the allocated variables using a Gibbs scan. It is trivial to show that
	\begin{align*}
	    	S^{(a)}_h\,|\,\cdots &\ind \mathrm{Gamma}(\alpha + n_h, 1 + u)\\
	    	\Delta^{(a)}_h\,|\,\cdots &\ind \mathrm{IW}_d\left(\nu_0 + n_h, \Psi_0+ \sum_{i: c_i=h}{(\eta_i-\mu^{(a)}_h)(\eta_i-\mu^{(a)}_h)^\top}\, \right).
	\end{align*}
	The full conditional of $\bm \mu^{(a)}$ is proportional to
	\[
        p(\bm\mu^{(a)}\,|\,\cdots) \propto f^{\text{app}}_{\text{DPP}}(\bm\mu^{(a)} \cup \bm\mu^{(na)} \mid \rho, \Lambda, K_0; R)  \prod_{h=1}^k{\prod_{i:c_i=h}{\calN_d(\eta_i\,|\,\mu_h^{(a)}, \Delta_h^{(a)})}}
    \]
    and we use a Metropolis-Hastings step to sample from it.
	
    \item \textit{Update of the latent allocation variables} $\bm c$. We found it useful to marginalize over the $\eta_i$'s to get better mixing chains. Hence, we can sample each $c_i$ independently from a discrete distribution over $\{1, \ldots, k+\ell\}$ with weights $\omega_{ih}$:
    \begin{align*}
        \omega_{ih} &\propto S_h^{(a)} \mathcal{N}_p(y_i \mid \Lambda \mu_h^{(a)}, \Sigma + \Lambda \Delta_h^{(a)} \Lambda^\top), \quad & h=1, \ldots, k \\
            \omega_{i k + h} & \propto S_h^{(na)} \mathcal{N}_p(y_i \mid \Lambda  \mu_h^{(na)}, \Sigma + \Lambda \Delta_h^{(na)} \Lambda^\top), \quad & h=1, \ldots, \ell
    \end{align*}
    Each evaluation of the $p$-dimensional Gaussian density would require $\calO(p^3)$ operations if care is not taken. However, we take advantage from the special structure of the covariance matrix. Using Woodbury's matrix identity, we have that
    \[
        \left(\Sigma + \Lambda \Delta \Lambda^\top\right)^{-1} = \Sigma^{-1} - \Sigma^{-1} \Lambda \left(\Delta^{-1} + \Lambda^\top \Sigma^{-1} \Lambda \right)^{-1} \Lambda^\top \Sigma^{-1},
    \]
    and hence we need now to compute the inverse of a $d \times d$ matrix. Therefore, evaluating the quadratic form in the exponential requires only $\calO(p)$ computational cost. 
    Moreover, using the matrix determinant lemma, the determinant of the covariance matrix can be computed  as
    \[
        \det(\Sigma + \Lambda \Delta \Lambda^\top) = \det(\Delta^{-1} + \Lambda^\top \Sigma^{-1} \Lambda) \det(\Delta) \det(\Sigma) .
    \]
    This is computed without additional cost by caching operations from the previous matrix inversion.

	\item \textit{Update of the ancillary variable} $u \sim \text{Gamma}(n, T)$.

	\item \textit{Update of the latent scores} $\bm\eta$. For $i=1,\ldots,n$, sample each $\eta_i$ independently from 
\[
\eta_i \mid \cdots \ind \calN_d(m_i\,,\,S_i)
\]
where
\[
    S_i = \biggl(\Lambda^T \Sigma^{-1} \Lambda + (\Delta_{c_i}^{(a)})^{-1} \biggr)^{-1}, \quad m_i = S_i \biggl( \Lambda^T \Sigma^{-1} y_i + (\Delta_{c_i}^{(a)})^{-1} \mu_{c_i}^{(a)} \biggr)
\]
This follows by standard calculations.

\end{enumerate}

\subsection{The Gibbs Sampler for binary data}
\label{sec:app_Gibbsxbinary}
The Gibbs sampler here requires an additional step to deal with the binary observations $z_{i,j}$. Specifically, the Gibbs updates described for the original APPLAM model, i.e., steps 1-8, except for the parameters here fixed, can be applied, since all the involved parameters are independent of $\bm z = (z_1,\ldots,z_n)$, conditionally to $\bm y = (y_1,\ldots,y_n)$. Then, the update of the $y_i$'s is done as follows: 
\[
 y_{i,j} \mid z_{i,j}, \cdots \ind \begin{cases}
 \calN(y_{i,j} \mid (\Lambda\eta_i)_j, 1) \indicator_{[0,\infty)}(y_{i,j}), \quad &\text{ if } z_{i,j}=1;\\
 \calN(y_{i,j} \mid (\Lambda\eta_i)_j, 1) \indicator_{(-\infty,0)}(y_{i,j}), \quad &\text{ if } z_{i,j}=0,
 \end{cases}
\]
for $i=1,\ldots,n; j=1,\ldots,p$, where $(\Lambda\eta_i)_j$ indicates the $j$-th element of the $p$-dimensional vector $\Lambda\eta_i$. That is, if $z_{i,j}=1$ (resp. $z_{i,j}=0$) the full-conditional distribution of $y_{i,j}$ follows a truncated Gaussian distribution with support $[0,\infty)$ (resp. support $(-\infty,0)$).


\section{Exploring the APPLAM model via numerical simulations}\label{app:simu1}

\subsection{Truncation level in the approximation of the DPP density}\label{app:truncation}

Here we provide a numerical illustration of the magnitude of the error introduced by the approximation procedure described in Section~\ref{sec:dpp_dens_approx}. Such an approximation procedure is standard in the literature on DPPs; see, e.g., \cite{Lav15}. 
The goodness of the approximation is clearly controlled by the truncation level $N$ of the series expansions in \eqref{eq:approx_dens}. We discuss the impact of such a truncation in the evaluation of both $D^{app}$ and $C^{app}$; note that the approximation error depends on the convergence rate to 0 of the $\varphi(k)$'s,  whose general expression is reported in \eqref{eq: spect_dens_Y}. 
We consider two illustrative situations, the first corresponding to $d=2$, the second taking $d=4$. For all the DPPs considered in the following, we set the intensity such that $\rho |R| = 1$ and $c$ such that $\rho = \rho_{\max} / 2$. 
For both cases ($d=2, 4$), we consider the Gaussian-like DPP as in Corollary~\ref{cor:gauss_dpp} when $\Lambda$ is either fixed and equal to $I_d$ (isotropic case), or when $\Lambda$ is random and distributed either according to the prior \eqref{eq:prior_lambda} or to the posterior distribution (for the anisotropic model used in the simulated scenario of Appendix~\ref{app:comparison_iso_aniso} ($d=2$) and the simulation study A of Appendix~\ref{app:comparison-lamb} ($d=4$)).
When $\Lambda$ is random, for all the quantities displayed in the next figures we report medians and  $95\%$ credible intervals computed over 100 draws from the random distribution of $\Lambda$.

Figure~\ref{fig:trunc_d2_D_phi_app} (left) refers to the case $d=2$ and shows values of $\varphi(k) = \varphi(k_1,0)$, i.e., how quickly $\varphi(k)$ converges to 0 along the first component. 
Similarly, Figure~\ref{fig:trunc_d2_D_phi_app} (right) show the values of $D^{app}$ as a function of $N$, while Figure~\ref{fig:trunc_d2_C_app} displays $C^{app}$ for different values of $N$. It is clear that the extremely fast decay of the $\varphi(k)$'s implies that the choice of $N$, as long as $N \geq 3$, has very limited impact in the approximation of $D^{app}$ and $C^{app}$.

\begin{figure}[!h]
    \centering 
 \includegraphics[width=0.495\linewidth]{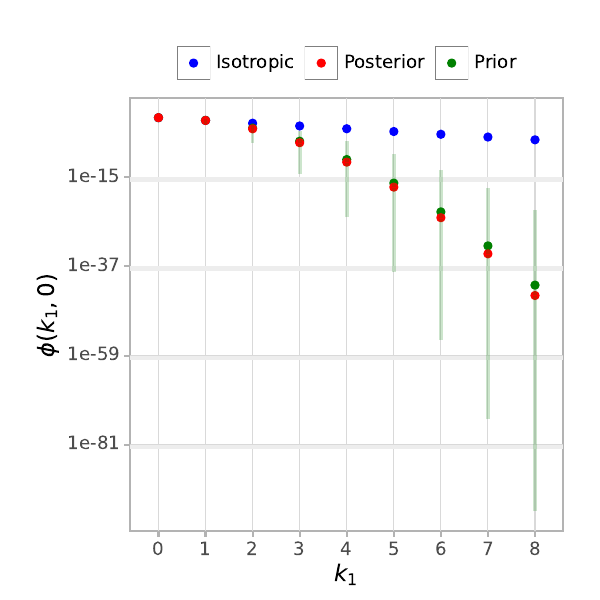}
  \includegraphics[width=0.495\linewidth]{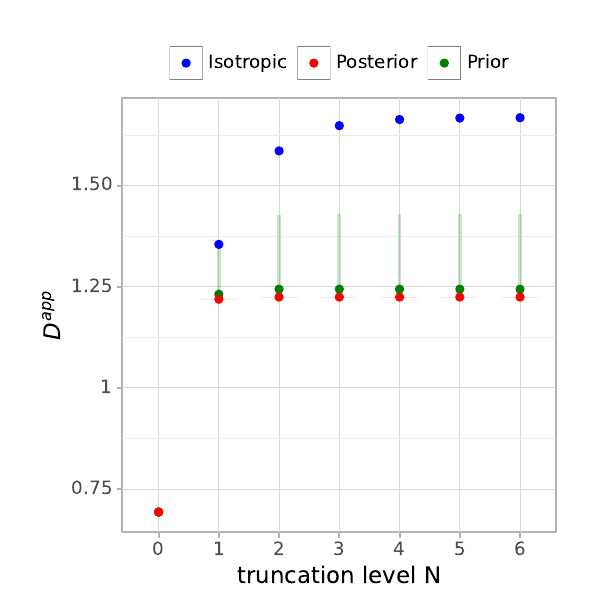}

\caption{Case $d=2$. Values of $\varphi(k_1,0)$ for increasing values of $k_1$ (left);  values of $D^{app}$ as a function of the truncation level $N$ (right). In the case $\Lambda$ is random, either from the prior or the posterior distribution, the plots report medians and $95\%$ credible intervals (in green) of the addressed quantities, obtained from $100$ samples of $\Lambda$.}
\label{fig:trunc_d2_D_phi_app}
\end{figure}

\begin{figure}[!h]
    \centering 
 \includegraphics[width=\linewidth]{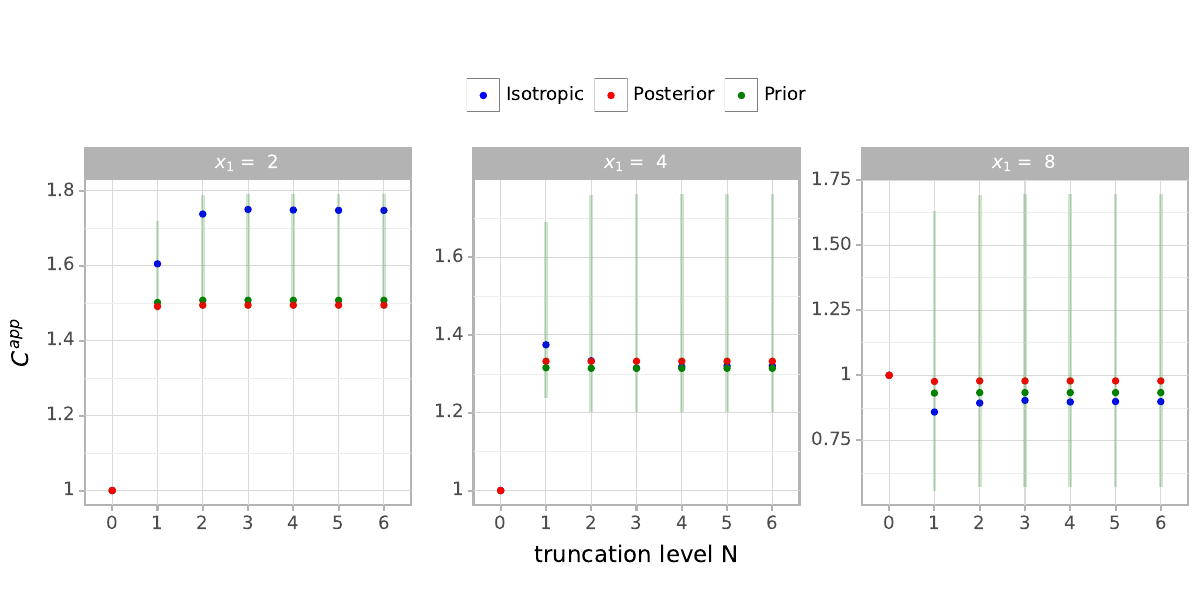}
  
\caption{Case $d=2$. Values of $C^{app}((x_1,0), (0,0))$ as a function of the truncation levels $N$ and for increasing $x_1$. In the case $\Lambda$ is random, either from the prior or the posterior distribution, the plots report medians and $95\%$ credible intervals (in green) of the addressed quantities, obtained from $100$ samples of $\Lambda$.}
\label{fig:trunc_d2_C_app}
\end{figure}

We repeat the same analysis for the case $d=4$, $\Lambda = I_4$ and $\Lambda$ random and distributed either as the prior or the posterior obtained by the MCMC in one of the examples considered in simulation study A of Appendix~\ref{app:comparison-lamb}. See Figures~\ref{fig:trunc_d4_D_phi_app} and \ref{fig:trunc_d4_C_app}, which clearly show that the choice of $N$ has a very limited impact in the approximation of $D^{app}$ and $C^{app}$.

\begin{figure}[!h]
    \centering 
 \includegraphics[width=0.495\linewidth]{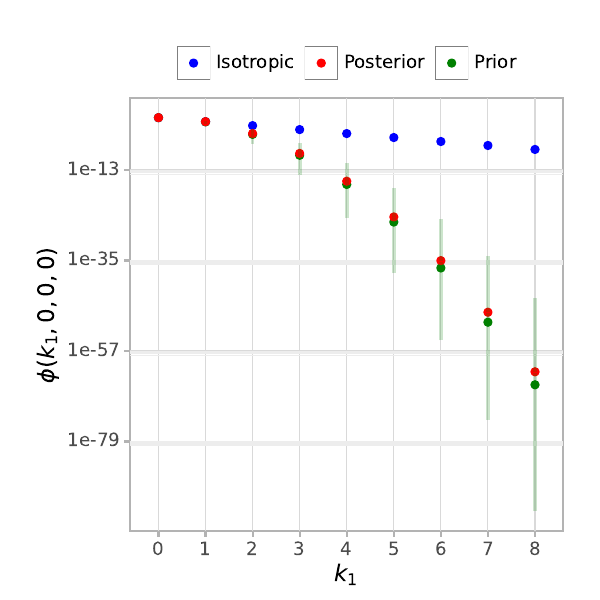}
  \includegraphics[width=0.495\linewidth]{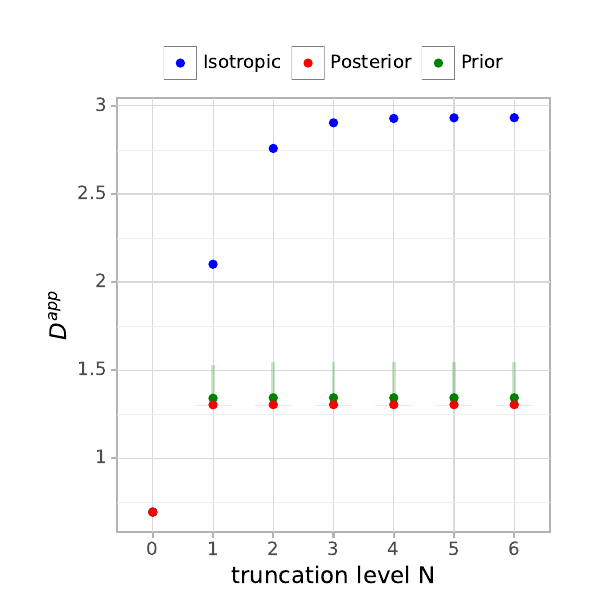}

\caption{Case $d=4$. Values of $\varphi(k_1,0,0,0)$ for increasing values of $k_1$ (left); values of $D^{app}$ as a function of the truncation levels $N$ (right). In the case $\Lambda$ is random, either from the prior or the posterior distribution, the plots report medians and $95\%$ credible intervals (in green) of the addressed quantities, obtained from $100$ samples of $\Lambda$.}
\label{fig:trunc_d4_D_phi_app}
\end{figure}

\begin{figure}[!h]
    \centering 
 \includegraphics[width=\linewidth]{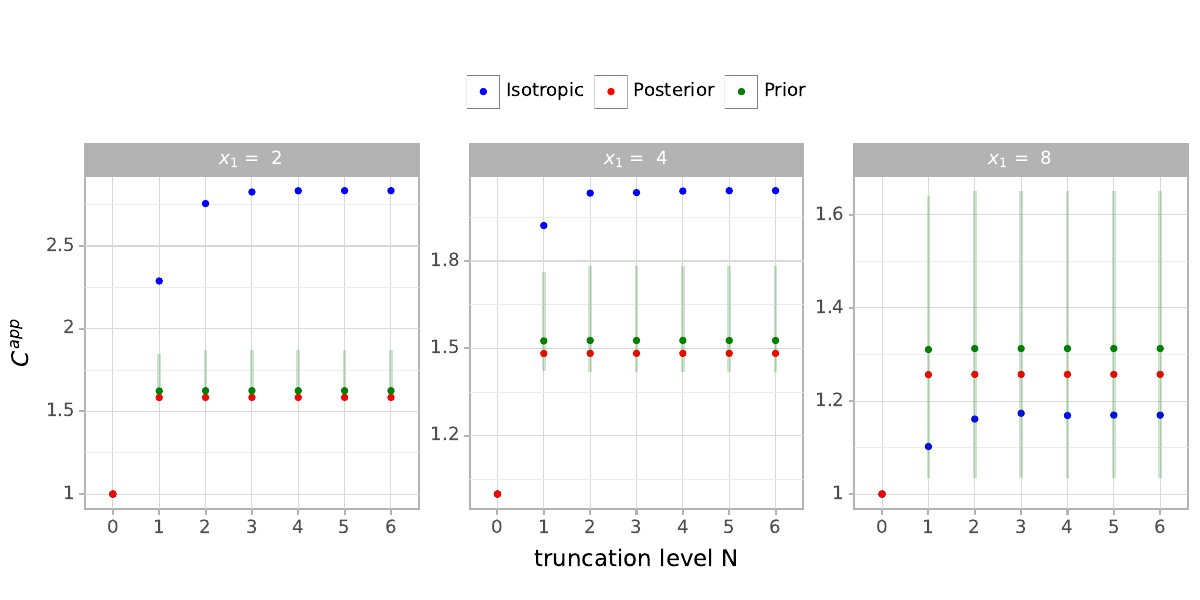}
  
\caption{Case $d=4$. Values of $C^{app}((x_1,0,0,0), (0,0,0,0))$, for different truncation levels $N$ and increasing $x_1$. In the case $\Lambda$ is random, either from the prior or the posterior distribution, the plots report medians and $95\%$ credible intervals (in green) of the addressed quantities, obtained from $100$ samples of $\Lambda$.}
\label{fig:trunc_d4_C_app}
\end{figure}

\FloatBarrier

\subsection{Analytical approach vs Automatic Differentiation}\label{app:analytic-automatic}

Figure~\ref{fig:ad_vs_grad} reports the comparison, in terms of memory usage (measured in Bytes) and execution time per iteration (measured in seconds), between the AD approach and the analytical approach in sampling the high-dimensional matrix of loadings $\Lambda$. Fixing all the other model parameters, we set $N=4$ and $6$ component centers $\mu_h$'s; see also \eqref{eq:approx_dens}. Note that the value of $N$ impacts both the memory usage and the execution time.  We compare the performance in sampling only the matrix $\Lambda$. We use $100$ data points simulated from a $p$-dimensional Gaussian distribution, for $p=100, 200$. 

\begin{figure}[!htb]
    \centering
    \includegraphics[width=0.8\textwidth]{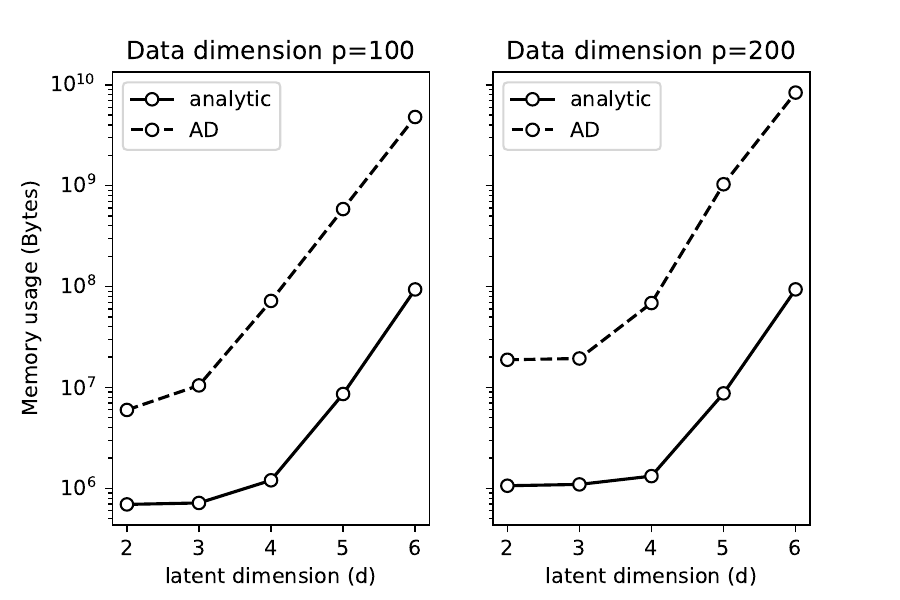}
    \includegraphics[width=0.8\textwidth]{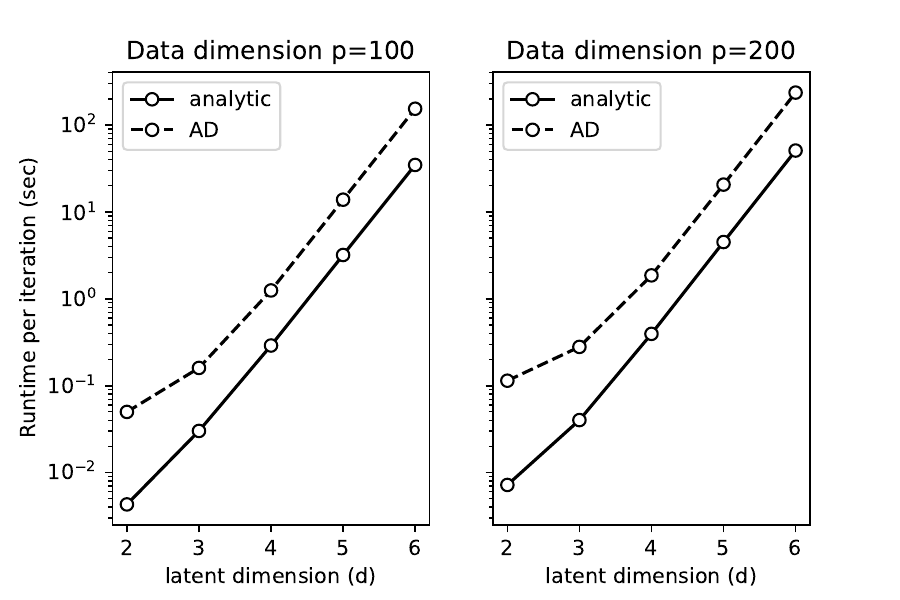}
    \caption{Memory requirement (top row) and run-time execution per iteration of MCMC with $n=100$ samples when the data-dimension is $p=100$ (left plot) and $p=200$ (right plot) as the latent dimension $d$ varies.}
    \label{fig:ad_vs_grad}
\end{figure}

\subsection{Comparison between anisotropic and isotropic processes}
\label{app:comparison_iso_aniso}

We highlight here why the anisotropic  DPP in \eqref{eq:prior_mu} is preferable to the isotropic DPP  in the latent factor model to perform clustering. 
We compare the APPLAM model based on the Gaussian-like DPP in Corollary \ref{cor:gauss_dpp} and its isotropic modification, where the anisotropic point process in \eqref{eq:prior_mu} is replaced by the isotropic Gaussian DPP defined by
\[
K_0(x)  = \rho \exp \biggl( -\frac{|| x||^2}{2 \,c^{-\frac{2}{d}}} \biggr), \qquad x \in \R^d.
\]
For such an isotropic Gaussian DPP, the Fourier transform $\varphi(x)$ of $K_0$ is equal to 
\begin{align*}
\varphi(x) &= \rho \,\frac{(2\pi)^{d/2}}{c} \exp \left( -2\pi^2  c^{-\frac{2}{d}} || x||^2 \right), \qquad x \in \R^d  
\end{align*}
for $\rho<\rho_{\max}$, with 
$\rho_{\max} = c (2\pi)^{-d/2}$.
This is the Gaussian (or squared exponential) point process in Equation (3.8) of \cite{Lav15}, with $\alpha = \sqrt{2} c^{-1/d}$.

\vspace{1ex}
\noindent
\textbf{Data generation.} We consider $p=100$ and $d=2$. We simulate $500$ iid $\eta_i \in \R^2$ from each of $m=2$ Gaussian densities with means $\mu_1 = (-2, 0 )$ and $\mu_2 = (2 ,0)$, and covariance matrices equal to $0.5 I_2$. Then we define $\Lambda = V \Lambda'$, where 
\begin{equation}
    \Lambda' = \begin{pmatrix}
    \lambda_1 & 0 \\
    0 & \lambda_2
\end{pmatrix}
\label{eq:lam-comparison}
\end{equation}
with $\lambda_1 = 3, \lambda_2 = 1/3$, and the columns of $V$ form an orthonormal basis for a $d$-dimensional linear subspace of $\R^p$. Finally, we simulate
\begin{equation*}
    y_i\,|\,\eta_i,\Lambda, \Sigma \ind \calN_p(y_i\,|\,\Lambda\eta_i,\, \Sigma)\,\qquad i=1,\ldots,n
\end{equation*}
where $n=1000$ and $\Sigma = 10 I_p$.

\vspace{1ex}
\noindent
\textbf{Model hyperparameters.} For both models (the APPLAM and the \emph{isotropic} model), we fix $\Sigma$ in \eqref{eq:lamb_general} as the $I_p$  and each $\Delta_h=I_d$.
For both DPPs, we fix $R = [-10,10]^2$. 
To match the two priors, we fix the values of $(\rho, c)$ as follows. We consider the same intensity for the two processes, i.e., $\rho|R|  = 0.1$; this implies the same prior mean for the number of components of the mixtures. We set two different values for $c$ to guarantee that the pair correlation function \eqref{eq:pair_cor} is the same along the $x$-axis of the latent space. This choice is made having full knowledge of the data generating process, given that the latent clusters are separated along the $x$-axis (see the values $\mu_1$ and $\mu_2$), to match as closely as possible the two priors.
This requires $c_{iso} = c_{aniso} \cdot \lambda_1/\lambda_2$, where $\lambda_1$ and $\lambda_2$ are defined below \eqref{eq:lam-comparison}. Then, we set $c_{aniso}$ so that $\rho = 0.9\rho_{max}$.
Following \cite{Chandra20}, we set $a = 0.5$ in \eqref{eq:prior_lambda}. 
For the prior distribution of the weights $\bm w$ in both repulsive mixtures, we assume $\alpha=10^{-3}$ in \eqref{eq:prior_jumps}. Note that $\alpha$ small, as in this case here, implies a marginal prior for $\bm w$ in \eqref{eq:prior_marks} mostly concentrated on the boundary of the simplex in $\R^{m-1}$. 
\begin{figure}[t]
    \centering    
    \includegraphics[width = 0.49\linewidth]{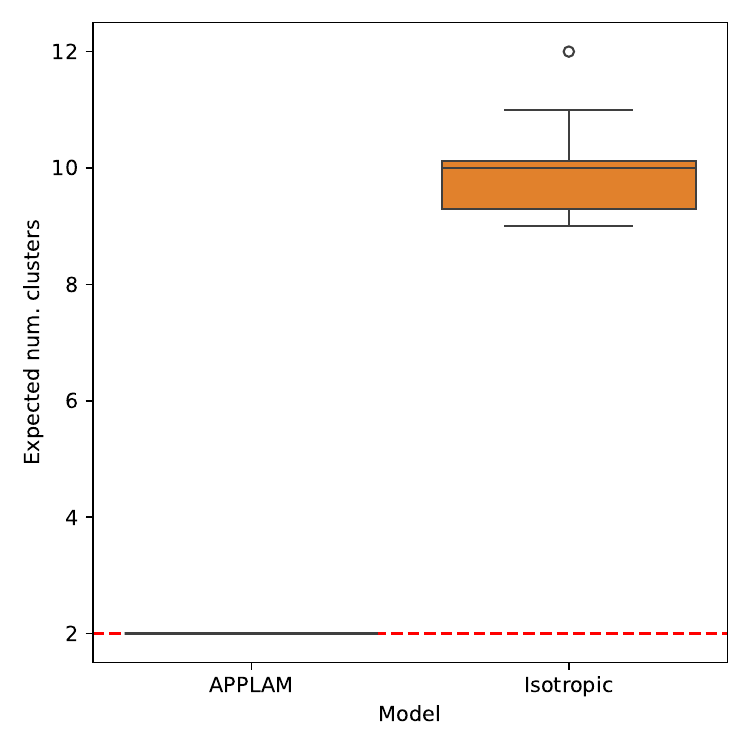}    \includegraphics[width = 0.49\linewidth]{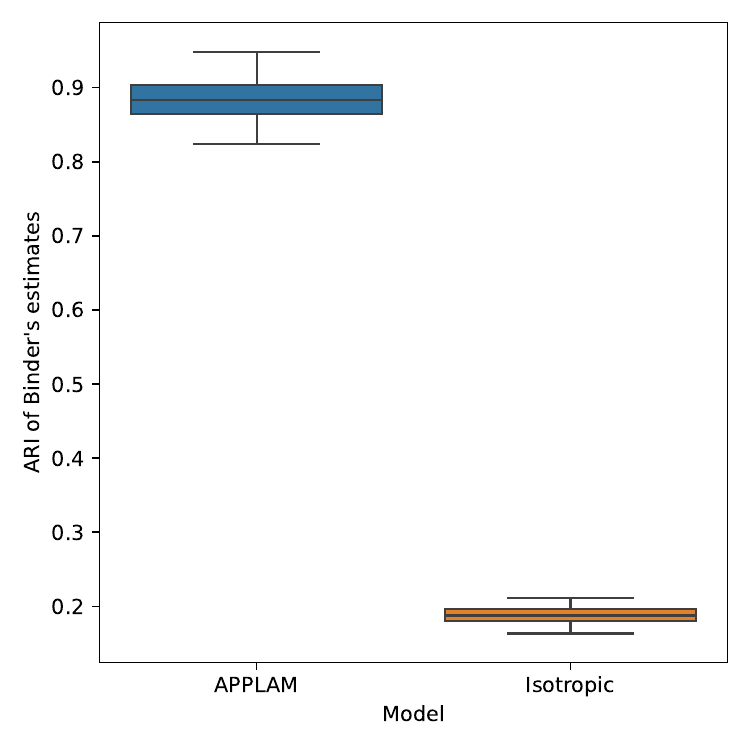}
    \caption{Boxplots of the posterior means of the number of clusters (left); the dotted red line indicates the true number of clusters. Boxplots of the ARI  between the estimated clusters (via Binder's loss function) and the true ones. Each boxplot refers to the output of the 100 replicated datasets.}
    \label{fig:aniso_iso_comp}
\end{figure}

\vspace{1ex}
\noindent
\textbf{Comparison of the estimated clusters.}
To compare the performance of the two models in estimating the clusters, we report the inference produced on 100 replicated datasets, drawn independently according to the procedure previously described. In particular, for each dataset, we repeatedly draw the latent factors $\eta_i$'s, the orthonormal matrix $V$ and then the datapoints. Figure~\ref{fig:aniso_iso_comp} (left) displays the boxplot (over the 100 replicated datasets)  of the posterior mean of the number of clusters. 
Figure~\ref{fig:aniso_iso_comp} (right) reports the boxplots of the values of the adjusted Rand index (ARI) between the cluster estimate and the true partition of the datasets. The cluster estimate has been obtained as the one minimizing the posterior expectation of Binder's loss \citep{Bin(78)}. 
The adjusted Rand index \citep{hubert1985comparing} between two partitions of a dataset may assume negative values, and
it is bounded above by 1, which corresponds to the case of perfect agreement between the two. 
It is clear that our APPLAM model accurately identifies the true clusters of the data (ARIs values close to 1) across all 100 replicated datasets, while the isotropic modification fails.

\subsection{Empirics for hyperparameter selection}\label{app:hyperparams}
 
We describe here a data-dependent method to select appropriate ranges for the values of hyperparameters $\rho$ and $c$ in Corollary~\ref{cor:gauss_dpp}; see Section~\ref{sec:elicitation} for an introductory discussion. These data-dependent estimates of $(\rho,c)$ should be interpreted as suggesting the order of magnitude of the values of $\rho$ and $c$ we should fix, instead of being sharp values of the hyperparameters. 
In particular, we adapt to our setting the arguments of Section 6.1 in \cite{Ber22}, where the authors propose to elicit hyperparameters of a Strauss process prior by controlling ($i$) the \emph{repulsion range} of the process and ($ii$) the strength of repulsiveness induced by the prior. 
Point ($i$) is driven by the assumption that repulsiveness should be induced to eliminate redundant clusters, but it should not severely affect density estimation.
Point ($ii$) instead derives from the argument that the strength of repulsiveness should not be treated as an absolute value but rather as a function of the sample size. See \cite{Ber22} for further details.

In our setting, we exploit ($i$) above to select an appropriate range for the parameter $\rho_{\max}$ in Corollary~\ref{cor:gauss_dpp}. Since $\rho_{\max} = c (2\pi)^{-d/2}$, this is equivalent to selecting an appropriate range for $c$.
It is easiest to work directly with the process $\tilde \Phi = \{\Lambda \mu\colon \mu \in \Phi\}$ defined in Theorem~\ref{teo:trans_dpp}. In particular, $\tilde \Phi$ is defined on (an appropriate subset of) $\R^d$ and is isotropic.
The main idea is to first select a suitable distance $\hat r$ such that we expect any two cluster centers sampled from $\tilde \Phi$, namely $\Lambda \mu_h$ and $\Lambda \mu_k$, to be such that $\|\Lambda \mu_h - \Lambda \mu_k\| > \hat r$. 
To this end, we follow the same argument as in \cite{Ber22} and compute the pairwise distances $d_{i, j}$ between all datapoints $y_i$ and $y_j$. If the sample size is large, we could also consider a subsample of data.
Then, we compute a kernel density estimate of the $d_{i, j}$'s and select $\hat r$ as the smallest local minimum of such a density.  Intuitively, we interpret the kernel density from 0 to $\hat r$  as the density of the distances between datapoints within the same cluster. Choosing $\hat r$ as we do, we guarantee that we do not force separation between locations of different clusters. See Figure~\ref{fig:r_estim} for a graphic display.

Now let $g(\cdot)$ be the pair correlation function of $\tilde \Phi$. Since the process is isotropic, then $g \equiv g(r)$ with $r \in [0, +\infty)$, taking values in $[0, 1)$. 
Having fixed such a value $\hat r$, we aim at choosing $c$ such that $g(\hat r)$ assumes a large value, e.g., $0.99$.  From \cite{Lav15}, $\hat r$ then takes the interpretation of ``range of repulsiveness'', such that two points whose distance is greater than $\hat r$ are effectively independent.
The expression for $g(r)$ is
\begin{equation*}
     g(r) = 1 -\exp\left(-\frac{r^2}{2 |\Lambda^T \Lambda|^{1/d} c^{-2/d} }\right),
\end{equation*}
which depends on $\Lambda$, that is random. Once we estimate $\Lambda$ in the expression above, we can obtain $\hat c$ as a solution of $g(\hat r)=0.99$.
We could assume an a priori plug-in estimate for $\Lambda$ using, e.g., its prior mean. However, this would fail if the posterior distribution is far from the prior.
Instead, our argument is the following. First of all, we note that the term $|\Lambda^T \Lambda|^{1/d}$ is the change of volume between the hypersquare $R$ where the DPP $\Phi$ is defined, and the set $\tilde R = \Lambda R$ where $\tilde \Phi$ is defined.
To get a rough estimate of $|\tilde R|$, we consider again the pairwise distances of the data and take the maximum value, say $r_{\max}$. We approximate $|\tilde R| \approx (r_{\max})^d$. This is analogous to what is proposed in \cite{bianchini2018determinantal} and \cite{Ber22}, where the DPP is defined on the smallest hypersquare containing all the data.  Solving $|\tilde R| = |\Lambda^T \Lambda|^{1/d} |R|$, and as $R = [-10, 10]^d$, we obtain $|\Lambda^T \Lambda|^{1/d} \approx (r_{\max}/ 20)^d$.
Hence, we fix
\[
    \hat c =  (\log 10)^{d/2}\left(\frac{ r_{\max} }{10 \hat r} \right)^d,
\]
and, from Corollary~\ref{cor:gauss_dpp}, $\hat \rho_{\max} = \hat c (2 \pi)^{-d/2}$.

This argument, applied to the datasets analyzed later in Appendix~\ref{app:comparison-lamb}, gives a range for $\hat \rho_{\max} \times |R|$ between $0.5$ and $20.3$, which shows that we should focus on DPPs with (very) small intensity.

Having obtained an estimate for $\hat \rho_{\max}$, we turn to point ($ii$) above, namely, the strength of repulsiveness.
Since the DPP exists for $\rho \leq \rho_{\max}$ (and has a density if the inequality is strict) it is straightforward to set $\rho = s \rho_{\max}$, $s \in (0, 1)$ as done in \cite{Lav15, bianchini2018determinantal, Ber22}. In such a way, the parameter $s$ takes the interpretation of the strength of repulsiveness of the DPP.
We argue that $s$ should be an increasing function of the sample size. Indeed, the point of using a repulsive prior is that if the cardinality of a cluster is small, the repulsiveness should prevail on the likelihood in the cluster, so that the probability of having a small cluster nearby a large one should be small. However, the definition of ``large'' and ``small'' clusters depends on the sample size; e.g., a cluster of size 10 is large in a dataset of $n=30$ individuals, but very small in a dataset with $n=1000$.

In the next subsection, we perform a simulation study to assess the model robustness to its hyperparameters, namely, $\rho$ and $s$. Our simulations confirm the insights obtained here, that is, the intensities should be fairly small. In particular, we suggest the following. First, having fixed $R = [-10, 10]^d$, select a value for $\rho$ such that $\rho |R|$ is small, i.e., smaller than 5 or possibly smaller than 1.
Having fixed $\rho$, one needs only to select a value for the strength for repulsiveness $s$, which should increase with the sample size. This can be done by performing a robustness analysis for different values of $s \in \{0.5, 0.75, 0.9\}$ or fixing $s$ according to users' preferences.

\begin{figure}
    \centering
    \includegraphics[width=0.6\linewidth]{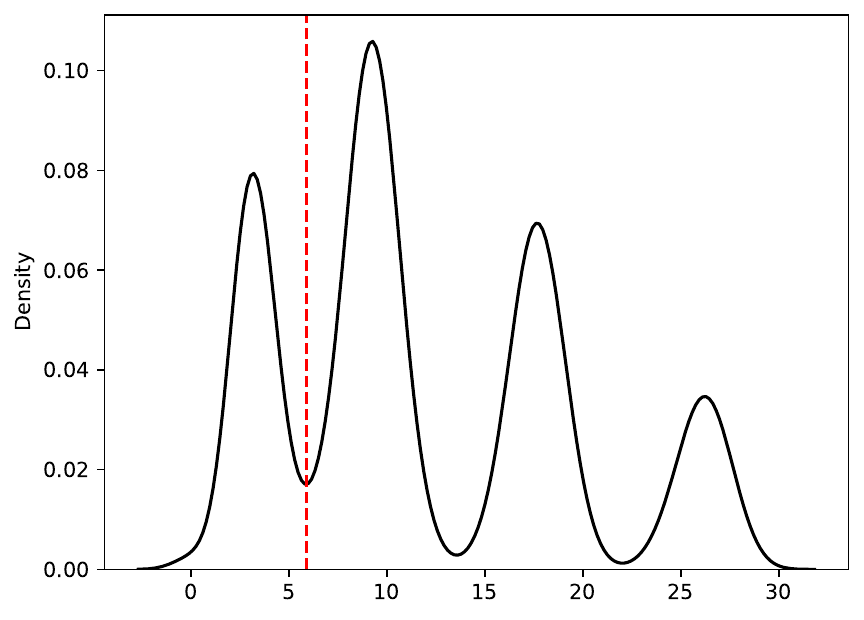}
    \caption{Schematic depiction of the estimation of $\hat{r}$ from the data. The black line is the kernel density estimate of the pairwise distances between datapoints in one of the simulated datasets analyzed in Appendix \ref{app:truncation}. The red dashed line is the value of $\hat r$.}
    \label{fig:r_estim}
\end{figure}

\subsection{Robustness to Model's Hyperparameters}\label{app:robustness}

Here we show the robustness of the APPLAM model with respect to its hyperparameters. To this end, we consider the synthetic data generation mechanism used for the simulation study~A in Appendix~\ref{app:sythetic_mechanism}. Specifically, we set the latent dimension $d=4$ and the data dimension $p=500$. We simulate $n_0 \in \{50, 100, 250\}$ iid latent factors $\eta_i \in \R^d$ from each of $m = 4$ Gaussian kernels with means $\mu_1 = (7.5,\ldots,7.5)\in\R^d$, $\mu_2 = (2.5,\ldots,2.5)$, $\mu_2 = (-2.5,\ldots,-2.5)$
and $\mu_4 = (-7.5,\ldots,-7.5)$, and identity covariance matrices. Then we independently generate the $n = n_0 \times 4 \in \{200, 400, 1000\}$ $p$-dimensional data from the multivariate Student $t$ distributions $t_p(\Lambda\eta_i,10I_4,3)$ (refer to Appendix \ref{app:sythetic_mechanism} for the parametrization of the multivariate Student $t$ distribution), where $\Lambda$ is drawn (once per dataset) from the uniform distribution on the set of $p \times d$ matrices whose columns form an orthonormal basis for a $d$-dimensional linear sub-space of $\R^{p}$.

We explore a fine grid of prior hyperparameters. In particular, we consider $\rho|R| \in \{ 0.05, 0.1, 0.2, 0.5, 1, 5\}$, with $c$ such that $\rho = s\rho_{max}$ (see Corollary~\ref{cor:gauss_dpp}) and $s \in \{ 0.5, 0.75, 0.9 \}$. Remember that 
$s$  is the strength of repulsion for a fixed $\rho$.
 We set $\alpha = 10^{-3}$ in \eqref{eq:prior_jumps} (the marginal prior for $\bm w$ is \textit{sparse}), $a_\sigma = 1, b_\sigma = 0.3$ in \eqref{eq:prior_sigma} and $a = 0.5$ in \eqref{eq:prior_lambda} as suggested in \cite{Chandra20}.
Every time we fit the APPLAM model, we run $2500$ burn-in iterations and $10^4$ iterations with thinning equal to $2$, for a final sample size equal to 5000. Figure~\ref{fig:robustness_nclus_simA} displays the posterior mean of the number of clusters under different choices of the hyperparameters and sample sizes. Figure~\ref{fig:robustness_ari_simA}, instead, shows the values of the adjusted Rand index (ARI) between the cluster estimate and the true partition of the datasets. The cluster estimate has been obtained as the one minimizing the posterior expectation of Binder's loss \citep{Bin(78)}. 
Observe that, except for extreme hyperparameters' values, the posterior inference is robust, though a more significant agreement (larger values of ARI) between the estimated clusters and the true ones is obtained for values of $\rho|R|$ between 0.1 and 0.5. Figure~\ref{fig:robustness_nclus_simA} also points out that, when $\rho |R|$ is large and $s$ is small, the number of estimated clusters increases. Conversely, larger values of $s$ and smaller values of $\rho|R|$ impose stronger repulsiveness, yielding a smaller number of estimated clusters.  This is exactly what we expected from the characteristics of our prior, as discussed in Appendix \ref{app:elicitation}.

\begin{figure}[t]
    \centering
    \includegraphics[width=\linewidth]{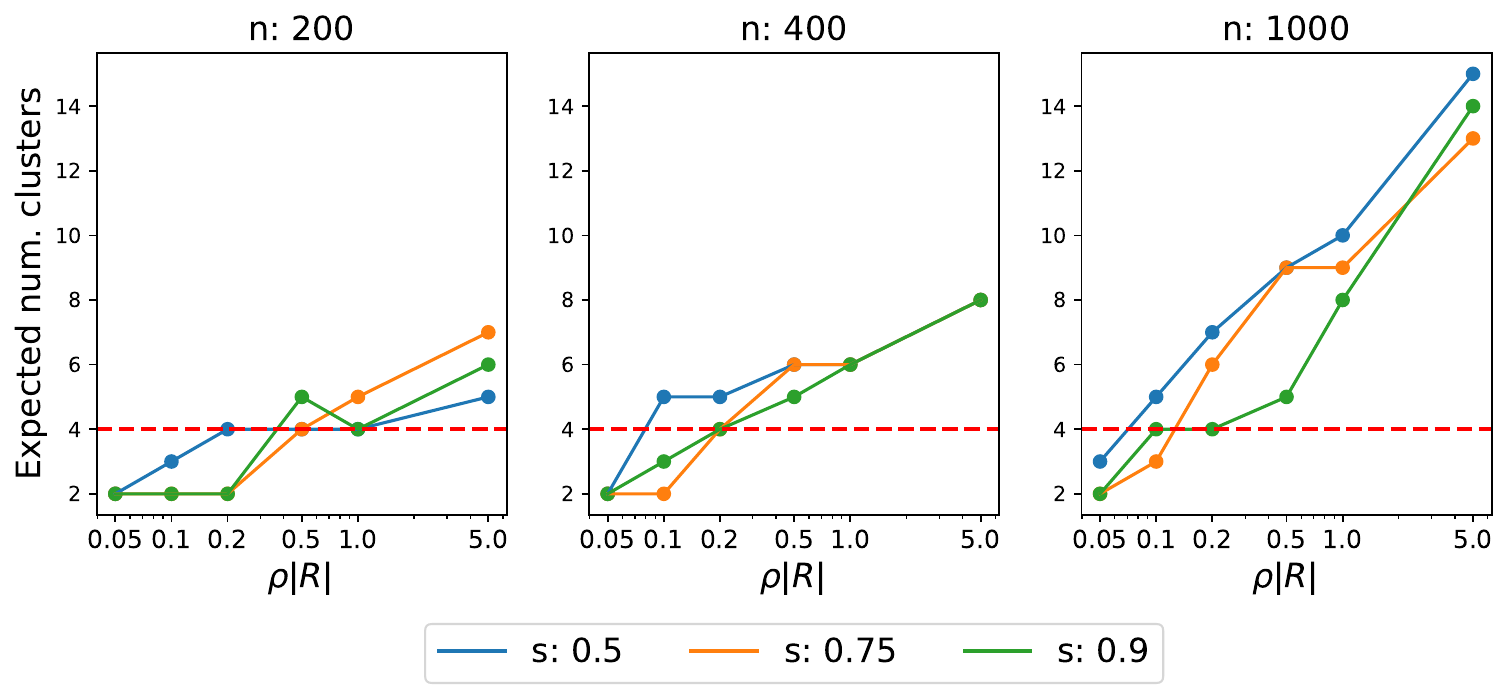}
    
    \caption{Posterior mean of the number of clusters for different hyperparameters and sample sizes.}
    \label{fig:robustness_nclus_simA}
\end{figure}

\begin{figure}[t]
    \centering
    \includegraphics[width=\linewidth]{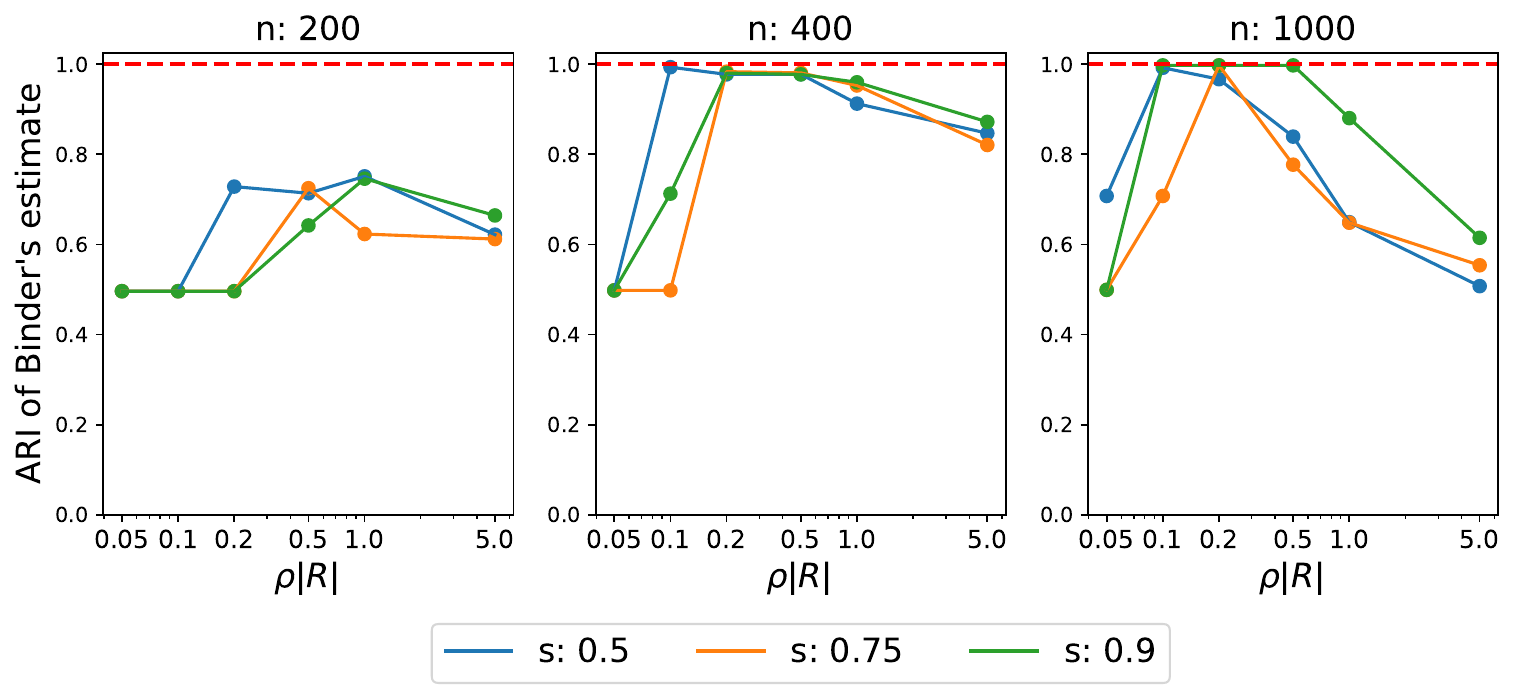}
    
    \caption{Adjusted rand index (ARI) for Binder's loss estimated clusters under different hyperparameters and sample sizes.}
    \label{fig:robustness_ari_simA}
\end{figure}

\FloatBarrier

\section{Comparison between APPLAM and LAMB on simulated datasets}\label{app:comparison-lamb}

\subsection{Synthetic data generation}\label{app:sythetic_mechanism}

We illustrate the comparison between the APPLAM model based on the Gaussian-like DPP in Corollary~\ref{cor:gauss_dpp}
and the Lamb model of \cite{Chandra20}. This assumes likelihood \eqref{eq:lamb_general} and a Dirichlet process mixture with total mass $\alpha_{DP}$ and a Normal-inverse Wishart as a base measure. We consider
two different sets of simulated data, referred to as simulation studies A and B. Both settings imply misspecification of the two models we assume for clustering. 
Let $t_p(q,\Sigma,\nu)$ denote the multivariate $p$-dimensional Student $t$ distribution with location $q$, scale matrix $\Sigma$ and degrees of freedom $\nu$.

In simulation study A we simulate $50$ iid latent factors $\eta_i \in \R^d$  from each of $m=4$ Gaussian kernels with means $\mu_h$, $h=1,\ldots,m$, and identity covariance matrices. Then we simulate
\begin{equation*}
    y_i\,|\,\eta_i,\Lambda, \Sigma \ind t_p(y_i\,|\,\Lambda\eta_i,\, \Sigma,\, 3),\,\qquad i=1,\ldots,n
\end{equation*}
where $\Sigma = 10 I_p$.  
We fix $\mu_1 = (7.5,\ldots,7.5)\in\R^d$, $\mu_2 = (2.5,\ldots,2.5)$, $\mu_2 =( -2.5,\ldots,-2.5)$
and $\mu_4 = (-7.5,\ldots,-7.5)$.  The matrix of factor loadings $\Lambda$ is fixed as one draw from the uniform distribution on the set of $p \times d$ matrices whose columns form an orthonormal basis for a $d$-dimensional linear sub-space of $\R^p$. 
For simulation study B,  50 iid latent factors are generated from each of $m=4$ multivariate Student $t$ distributions $t_d(\mu_h,I_d,3)$, $h=1,\ldots,m$,  with $\mu_h$'s fixed as in simulation study A. In both simulation settings, then 
$n=50\times 4=200$ datapoints are independently generated from the likelihood in (\ref{eq:lamb_general}), with  $\Sigma = 0.1 I_p$ and $\Lambda$ generated as in simulation study A. For the simulation setting A, we consider $p=\{500, 1000\}$ and $d=\{4,8\}$, while we report $p=\{500, 1000\}$ and $d= 4$ for the setting B. For each simulation setting, we sample 100 independent datasets and average posterior estimates over them (see Appendix \ref{app:comparison_simulated}). 

\subsection{Hyperparameters elicitation and MCMC details}\label{app:elicitation}

The elicitation for the common hyperparameters in APPLAM and Lamb follows the default choices of \cite{Chandra20}.  
Specifically, we set $a_\sigma= 1, b_\sigma =0.3 $ in (\ref{eq:prior_sigma}) and $a= 0.5$ in (\ref{eq:prior_lambda}). 
Moreover, Lamb assumes a Dirichlet process location-scale mixture of Gaussian densities for the latent factors. In particular, the Normal-inverseWishart distribution is taken as the base measure, with location $\mu_0$, scale $k$, covariance matrix $\Psi_0$ and degrees of freedom $\nu_0$. The default choice sets $\mu_0$ as the null vector, $k=0.001$, $\Psi_0=\delta \, I_d$, with $\delta = 20$ and $\nu_0 = d + 50$. 
Coherently, for APPLAM we set $\Psi_0 = \delta \, I_d$, $\delta = 20$ and $\nu_0= d + 50$ in (\ref{eq:prior_marks}).  We also set $\alpha=10^{-3}$ in (\ref{eq:prior_jumps}) to assume a \textit{sparse} marginal prior for $\bm w$.

It is difficult to match the prior for the number of clusters in APPLAM and Lamb: while for Lamb it is determined by $\alpha_{DP}$, for APPLAM it is induced by the prior for the number of components $m$, which is controlled by $(\rho, c)$. In particular, $\rho |R|$ is the prior expected number of components if we do not condition on $m \geq 1$.
Therefore, we limit ourselves to evaluate empirically the robustness of the cluster estimates to the choice of these parameters in our simulations.
In particular, we consider $\alpha_{DP}=\{0.1,0.5,1\}$, which corresponds to the prior expected number of clusters of $1.57, 3.63$, and $5.88$, respectively.
For our model instead,  we follow the arguments in Appendix~\ref{app:robustness} and we test different degrees of repulsiveness by setting $\rho |R| \in \{0.5,1,5\}$ and $c = c_{\rho} = 2 \rho (2 \pi)^{d/2}$, i.e., $s=0.5$.  

Any time we fit the APPLAM model, we run $2000$ burn-in iterations and $5000$ iterations with thinning equal to $2$, for a final sample size of 2500. For each run of the Lamb model instead, we simulate $10^6$ burn-in iterations and $5\times 10^4$ iterations with thinning equal to $10$. The poor mixing of the Lamb algorithm in our simulations (see Figure~\ref{fig:lamb_noburn}) demands a very long burn-in phase.

\begin{figure}[ht]
    \centering
    \includegraphics[width=0.8\textwidth]{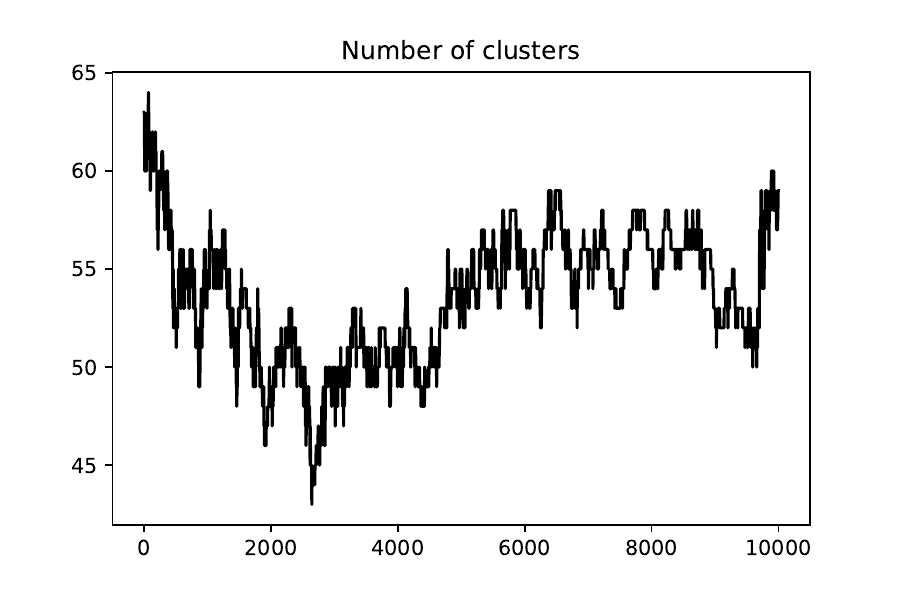}
    \caption{Traceplot of the number of clusters under the MCMC algorithm of the Lamb model for $10^5$ iterations for simulation study B, $p=100$, $d=5$, $\alpha_{DP}=0.5$.}
    \label{fig:lamb_noburn}
\end{figure}

\subsection{Comparison between the two models}
\label{app:comparison_simulated}

We focus here on the cluster estimates of the two models.
We consider the expected number of clusters and a posterior summary statistic of the random partition of the subjects induced by the allocation variables $c_i$'s. As before, the cluster estimates are obtained by minimizing the posterior mean of
Binder's loss function \citep{Bin(78)}, and compute the adjusted rand index (ARI) between the estimated partition and the true one, for both models.
Values of ARI equal to one correspond to perfect agreement between the two partitions. Instead, if ARI is equal to zero, the estimated partition is equivalent to a random labelling of the data.
Figures~\ref{fig:simuA_d4}- \ref{fig:simuA_d8} for simulation study A, Figure~\ref{fig:simuB} for simulation study B,  report 
boxplots of the posterior expectation of the number of clusters (top rows) and of the ARI between the Binder's loss estimated clusters and the true ones (bottom rows) over the 100 datasets.
Tables \ref{table:simuA_d4}-\ref{table:simuA_d8}-\ref{table:simuB} show numerical values, aggregated over the independent datasets, of the posterior mean and posterior modes of the number of clusters, and the ARI's values as before.

The figures generally show that the Lamb model tends to overestimate the true number of clusters ($m=4$).
On the other hand, APPLAM appears rather robust to the choice of $\rho|R|$: 
not only is the estimated number of clusters always moderate (which is good in terms of easier interpretability), but the associated posterior estimates spread around the correct number of clusters. Moreover, APPLAM provides better cluster estimates than Lamb overall in terms of ARI values between the estimated partition and the true one.

\begin{figure}[t]
    \centering
    \includegraphics[width=\linewidth]{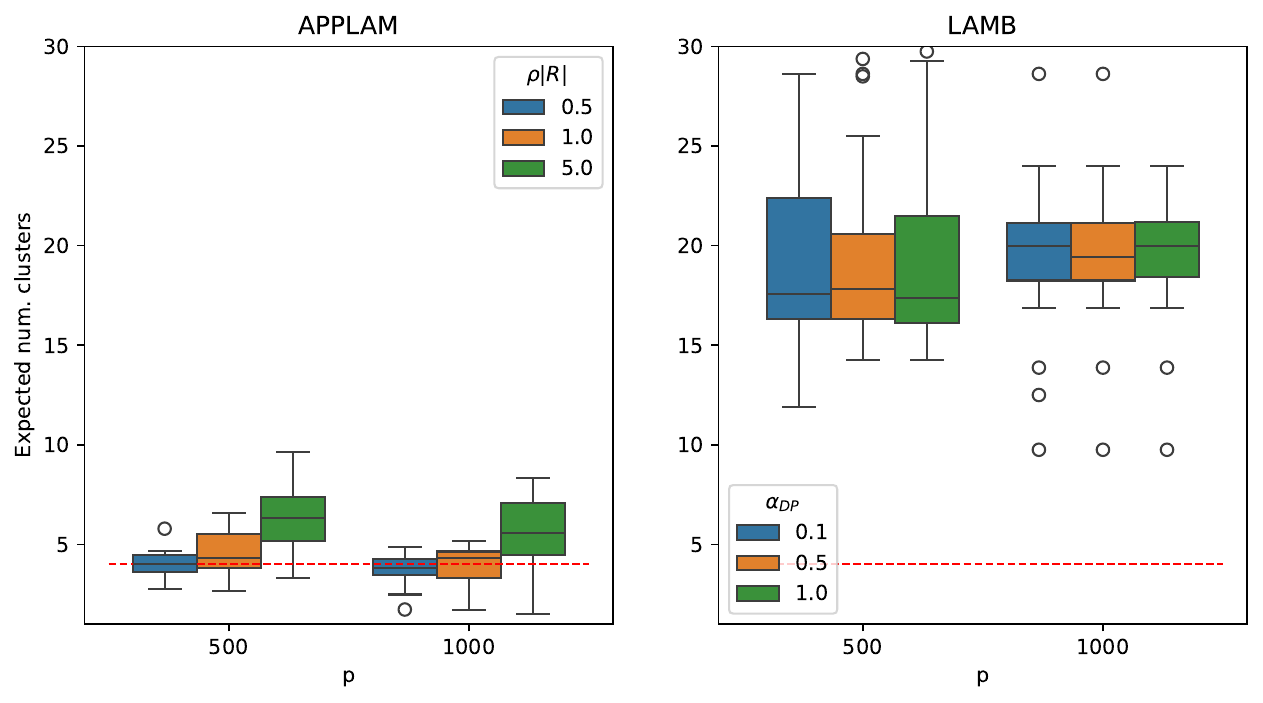}

    \includegraphics[width=\linewidth]{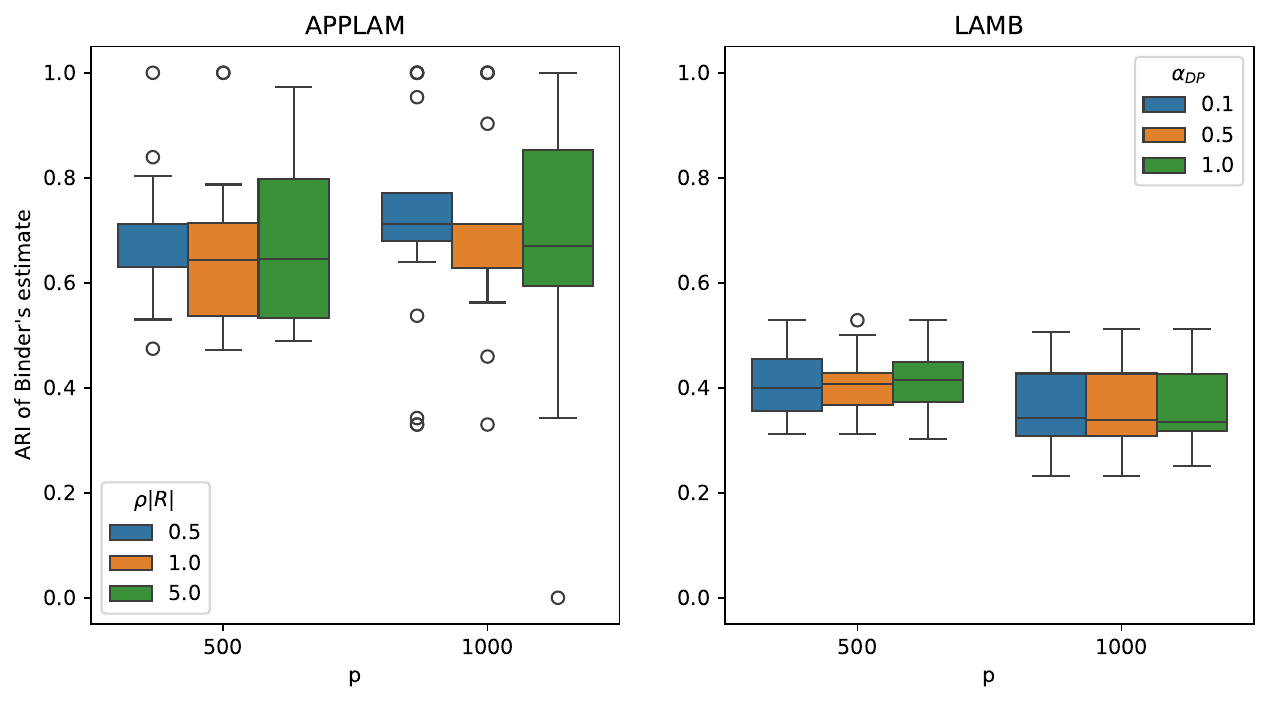}
    \caption{Simulation A, case $d=4$: boxplots of the posterior expectation of the number of clusters (top row) and boxplots of the adjusted rand index (ARI) for the Binder's estimated clusterings (bottom row). Each boxplot refers to the output of the 100 replicated datasets.}
    \label{fig:simuA_d4}
\end{figure}

\begin{table}[h]
\centering
\begin{tabular}{lllrrr}
\toprule
     &      &     &  Mean nclus &  Mode nclus &  ARI Binder's \\
$p$ & Model & Parameter &                &                 &               \\
\midrule
500  & APPLAM & 0.5 &           4.02 &            3.90 &          0.69 \\
     &      & 1.0 &           4.60 &            4.55 &          0.65 \\
     &      & 5.0 &           6.23 &            6.30 &          0.67 \\
     & Lamb & 0.1 &          19.36 &           19.05 &          0.41 \\
     &      & 0.5 &          19.48 &           19.00 &          0.41 \\
     &      & 1.0 &          19.64 &           18.80 &          0.41 \\
1000 & APPLAM & 0.5 &           3.70 &            3.75 &          0.71 \\
     &      & 1.0 &           4.02 &            4.05 &          0.71 \\
     &      & 5.0 &           5.57 &            5.60 &          0.67 \\
     & Lamb & 0.1 &          19.53 &           18.00 &          0.37 \\
     &      & 0.5 &          19.58 &           17.85 &          0.37 \\
     &      & 1.0 &          19.47 &           18.05 &          0.37 \\
\bottomrule
\end{tabular}
\caption{Simulation A, case $d=4$:  posterior mean of the number of clusters (Mean nclus), posterior mode of the number of clusters (Mode nclus) and the ARIs between Binder's loss cluster estimates (ARI Binder's) and the true one, averaged over 100 simulated datasets.  \emph{Parameter} refers to $\rho|R|$ for the APPLAM model and $\alpha_{DP}$ for the Lamb model.}
\label{table:simuA_d4}
\end{table}

\begin{figure}[t]
    \centering
    \includegraphics[width=\linewidth]{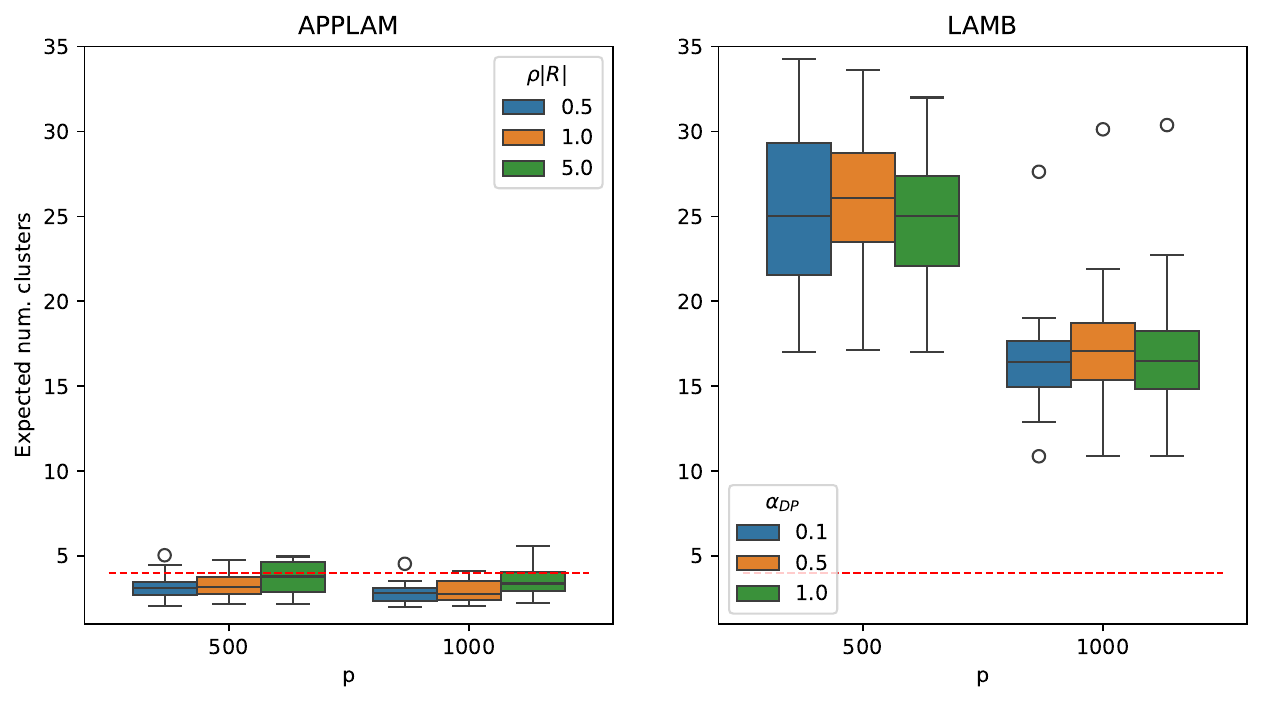}

    \includegraphics[width=\linewidth]{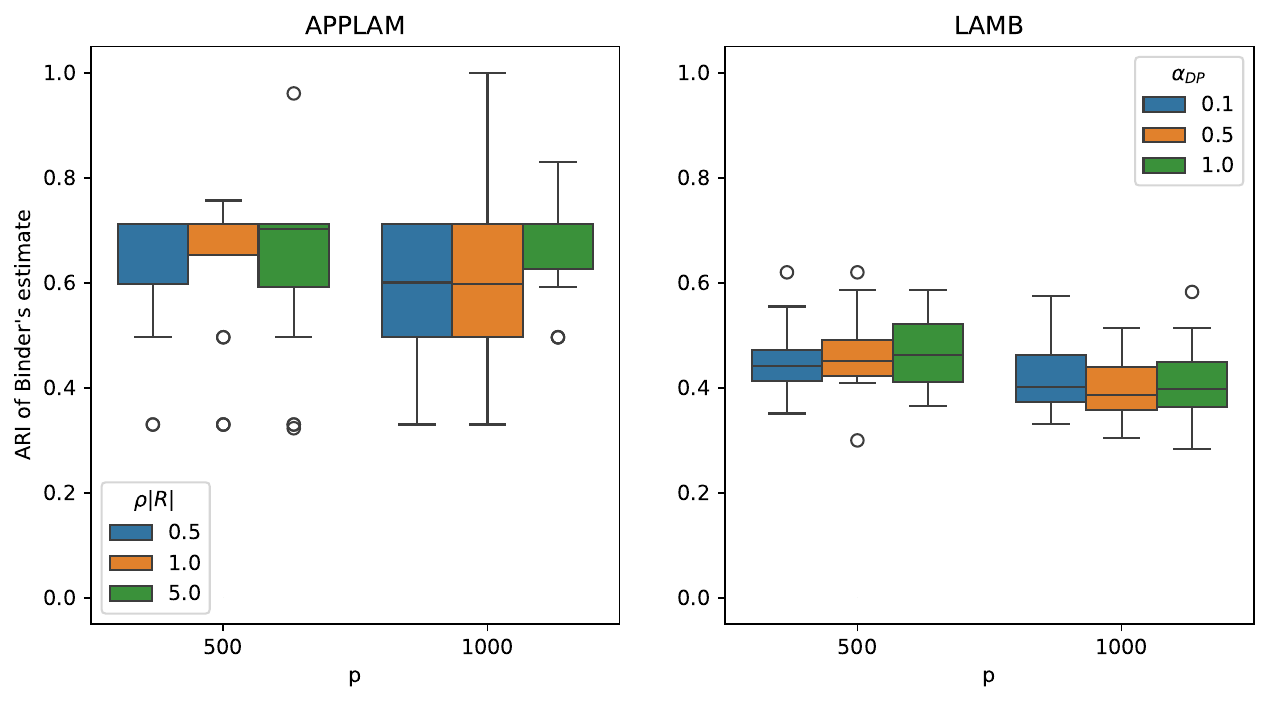}
    \caption{Simulation A, case $d=8$: boxplots of the posterior expectation of the number of clusters (top row) and boxplots of the adjusted rand index (ARI) for Binder's estimated clusters (bottom row). Each boxplot refers to the output of the 100 replicated datasets.}
    \label{fig:simuA_d8}
\end{figure}

\begin{table}[h]
\centering
\begin{tabular}{lllrrr}
\toprule
     &      &     &  Mean nclus &  Mode nclus & ARI Binder's \\
$p$ & Model & Parameter &                &                 &               \\
\midrule
500  & APPLAM & 0.5 &           3.21 &            3.20 &          0.64 \\
     &      & 1.0 &           3.29 &            3.05 &          0.63 \\
     &      & 5.0 &           3.71 &            3.60 &          0.63 \\
     & Lamb & 0.1 &          25.28 &           24.15 &          0.45 \\
     &      & 0.5 &          25.79 &           25.10 &          0.46 \\
     &      & 1.0 &          24.64 &           22.65 &          0.47 \\
1000 & APPLAM & 0.5 &           2.82 &            2.85 &          0.60 \\
     &      & 1.0 &           2.92 &            2.90 &          0.60 \\
     &      & 5.0 &           3.50 &            3.45 &          0.67 \\
     & Lamb & 0.1 &          16.62 &           15.35 &          0.42 \\
     &      & 0.5 &          17.49 &           16.05 &          0.39 \\
     &      & 1.0 &          17.22 &           16.25 &          0.41 \\
\bottomrule
\end{tabular}
\caption{Simulation A, case $d=8$: posterior mean of the number of clusters (Mean nclus), posterior mode of the number of clusters (Mode nclus) and the ARIs between Binder's loss cluster estimates (ARI Binder's) and the true one, averaged over 100 simulated datasets. \emph{Parameter} refers to $\rho|R|$ for the APPLAM model and $\alpha_{DP}$ for the Lamb model.}
\label{table:simuA_d8}
\end{table}

\begin{figure}[t]
    \centering
    \includegraphics[width=\linewidth]{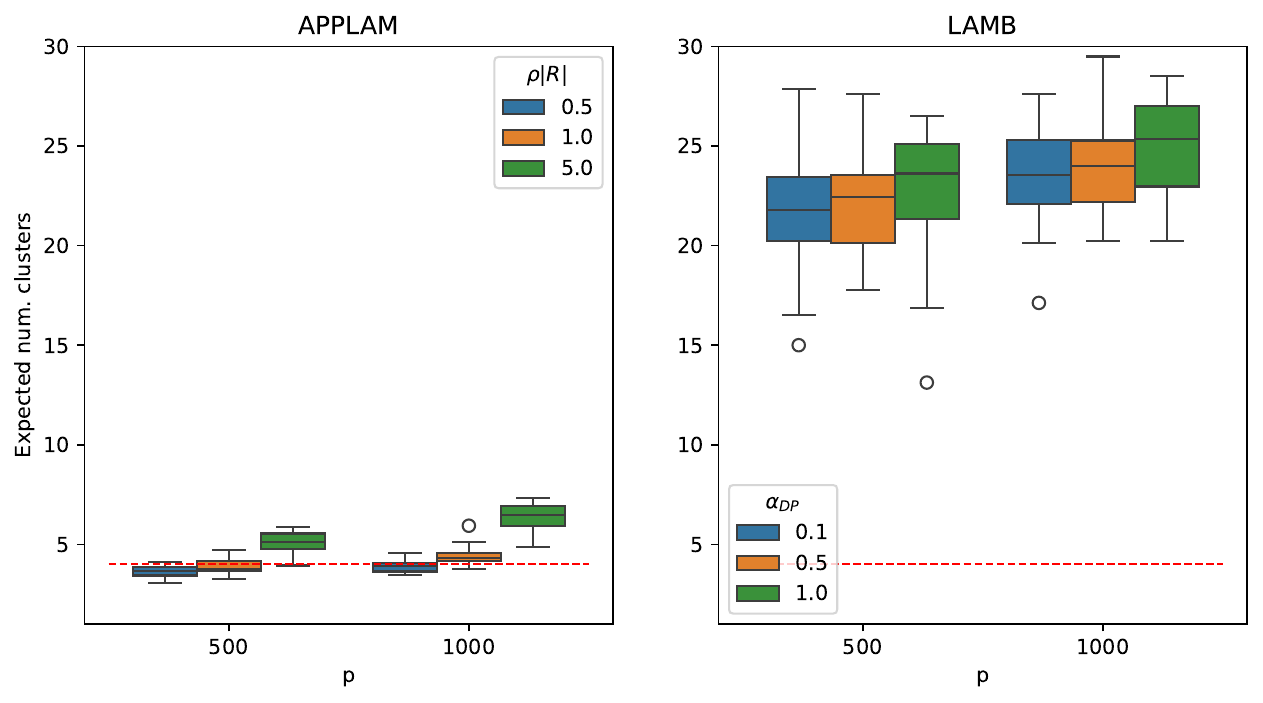}

    \includegraphics[width=\linewidth]{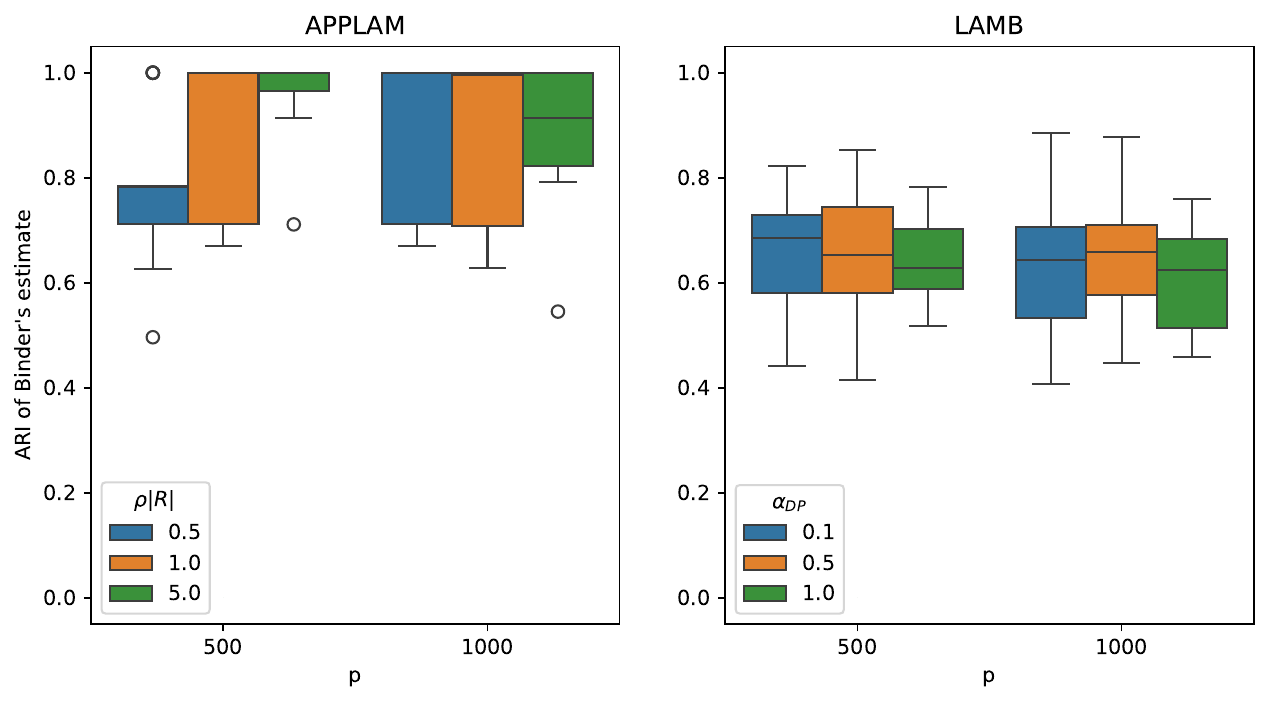}
    \caption{Simulation B: boxplots of the posterior expectation of the number of clusters (top row) and boxplots of the adjusted rand index (ARI) for Binder's estimated clusters (bottom row). Each boxplot refers to the output of the 100 replicated datasets.}
    \label{fig:simuB}
\end{figure}

\begin{table}[h]
\centering
\begin{tabular}{lllrrr}
\toprule
     &      &     &  Mean nclus &  Mode nclus &  ARI Binder's \\
$p$ & Model & Parameter &                &                 &               \\
\midrule
500  & APPLAM & 0.5 &           3.61 &            3.50 &          0.77 \\
     &      & 1.0 &           3.91 &            3.95 &          0.88 \\
     &      & 5.0 &           5.11 &            5.00 &          0.97 \\
     & Lamb & 0.1 &          21.71 &           20.10 &          0.66 \\
     &      & 0.5 &          21.96 &           20.70 &          0.65 \\
     &      & 1.0 &          22.72 &           21.35 &          0.64 \\
1000 & APPLAM & 0.5 &           3.89 &            3.85 &          0.79 \\
     &      & 1.0 &           4.44 &            4.35 &          0.88 \\
     &      & 5.0 &           6.37 &            6.30 &          0.90 \\
     & Lamb & 0.1 &          23.40 &           23.20 &          0.63 \\
     &      & 0.5 &          24.02 &           22.45 &          0.65 \\
     &      & 1.0 &          24.86 &           23.25 &          0.60 \\
\bottomrule
\end{tabular}
\caption{Simulation B: posterior mean of the number of clusters (Mean nclus), posterior mode of the number of clusters (Mode nclus) and the ARIs between Binder's loss cluster estimates (ARI Binder's) and the true one, averaged over 100 simulated datasets. \emph{Parameter} refers to $\rho|R|$ for the APPLAM model and $\alpha_{DP}$ for the Lamb model.}
\label{table:simuB}
\end{table}

\FloatBarrier

\section{Details on the Bauges dataset}\label{app:bauges}

\subsection{Hyperparameter elicitation}\label{app:bauges_hyper_elic}
We examine a broad range of prior hyperparameters. In particular, we consider $d \in \{3, 4\}$, $\rho|R| \in \{0.1, 0.5, 3, 5\}$, while fixing $c$ such that $\rho = s \rho_{\max}$ (see Corollary~\ref{cor:gauss_dpp}), with $s \in \{0.5, 0.75, 0.9 \}$. Moreover, we consider $\alpha \in \{10^{-3}, 10^{-2}, 10^{-1}\}$ in \eqref{eq:prior_jumps} and two choices for the hyperparameters of $\Sigma$ in \eqref{eq:prior_sigma}, specifically $a_\sigma = 3, b_\sigma = 2$ and $a_\sigma = 2, b_\sigma = 1$. Hyperparameter $a$ in \eqref{eq:prior_lambda} is fixed to $0.5$ as suggested in \cite{Chandra20}.  
We have discussed in the previous sections (in particular in Appendices \ref{app:hyperparams}-\ref{app:robustness}) that $\rho |R|$ should be small (e.g. not larger than 5), while $s\in(0,1)$ takes the interpretation of the strength repulsiveness
of the DPP. Hyperparameter $\alpha$ controls the \textit{sparseness} of the prior for the weights $\bm w$ of the mixture distribution of the latent parameters $\eta_i$'s; the smaller is $\alpha$, the more \textit{sparse} is this prior. The values fixed for hyperparameters $(a_\sigma,b_\sigma)$ are standard choices for the marginal prior of a variance parameter: both correspond to the prior mean of the $\sigma_j^2$'s equal to 1, but the prior variance is equal to 1 in the first case and is infinite in the second case. Assuming 3 and 4 ad the values of $d$, the dimension of the latent factors $\eta_i$'s, amounts to impose dimensionality reduction (here $p = 123$). We do not report here posterior inference with $d=2$ since it led to much poorer estimates in terms of WAIC.
We fit the resulting 144 models separately by running the MCMC algorithm for $3000$ burn-in iterations and $5000$ iterations with thinning equal to 5.

\subsection{Inference obtained via APPLAM}

Figure~\ref{fig:waics_bauges} reports the WAIC \citep{Watanabe13} values for all the models we fitted. It is evident that, in general, models with $d=3$ have a better fit (larger WAIC indexes) than those with $d=4$. Tables \ref{tab:bauges_d3} and \ref{tab:bauges_d4} display all the numerical values for the WAIC, together with the number of estimated clusters obtained by minimizing the posterior expected value of Binder's loss.

\begin{figure}[ht]
    \centering
    \includegraphics[width=0.8\textwidth]{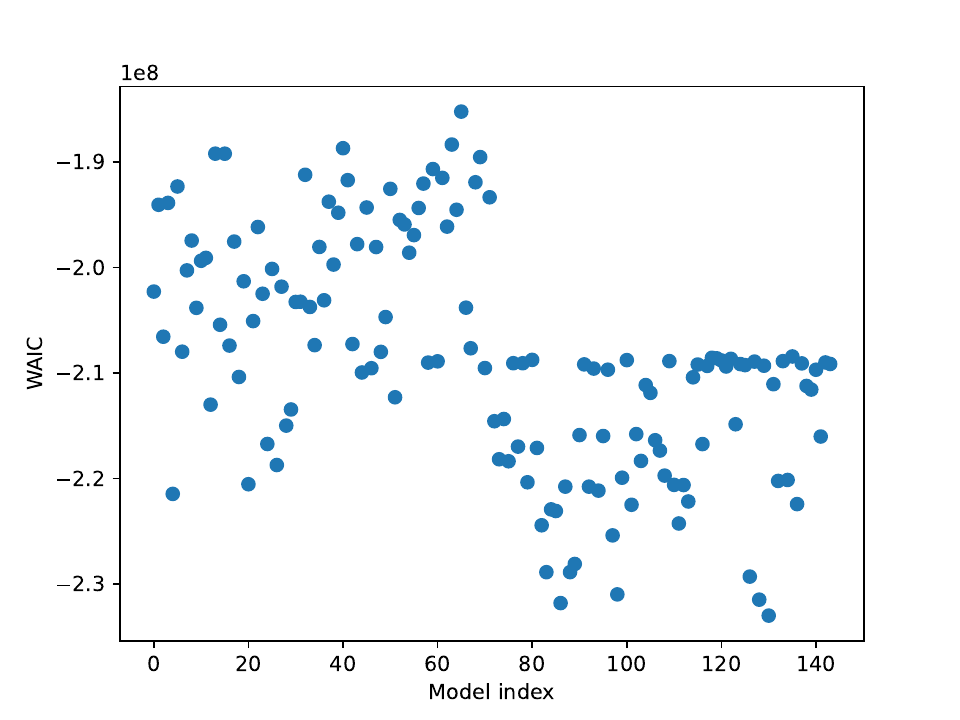}
    \caption{WAIC values for the different fitted models: the first half refers to models with $d=3$, the second half refers to models with $d=4$.}
    \label{fig:waics_bauges}
\end{figure}

Figures~\ref{fig:bauges_nclus_d3} and \ref{fig:bauges_nclus_d4} report the posterior  mean of the number of clusters when $d=3$ and $d=4$, respectively, as a function of $\rho |R| \in \{0.1, 0.5, 3, 5 \}$.

\begin{figure}[ht]
    \centering
    \includegraphics[width=\textwidth]{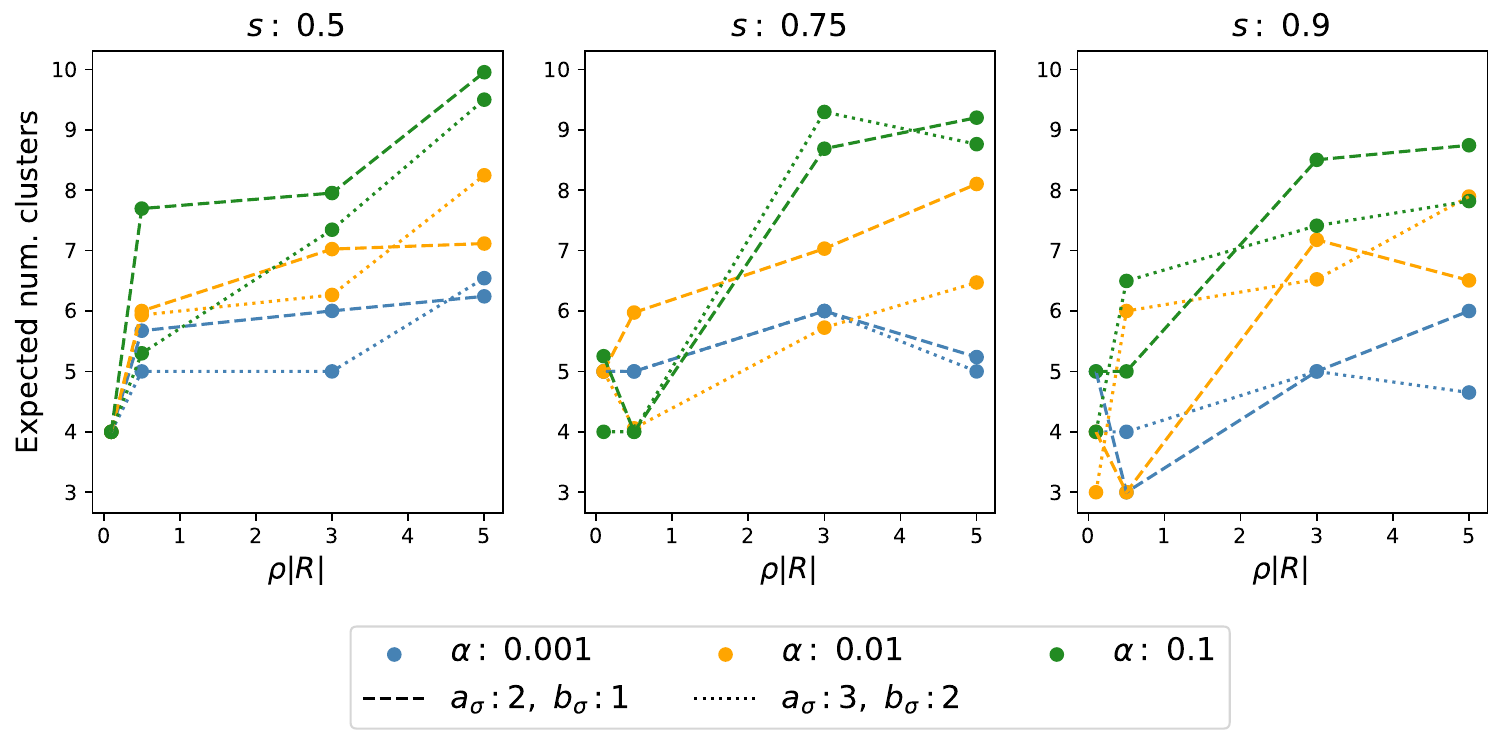}
    \caption{Posterior mean  of the number of clusters for $d=3$ as a function of $\rho |R| \in \{0.1, 0.5, 3, 5 \}$, for different choices of the model's hyperparameters.}
    \label{fig:bauges_nclus_d3}
\end{figure}

\begin{figure}[ht]
    \centering
    \includegraphics[width=\textwidth]{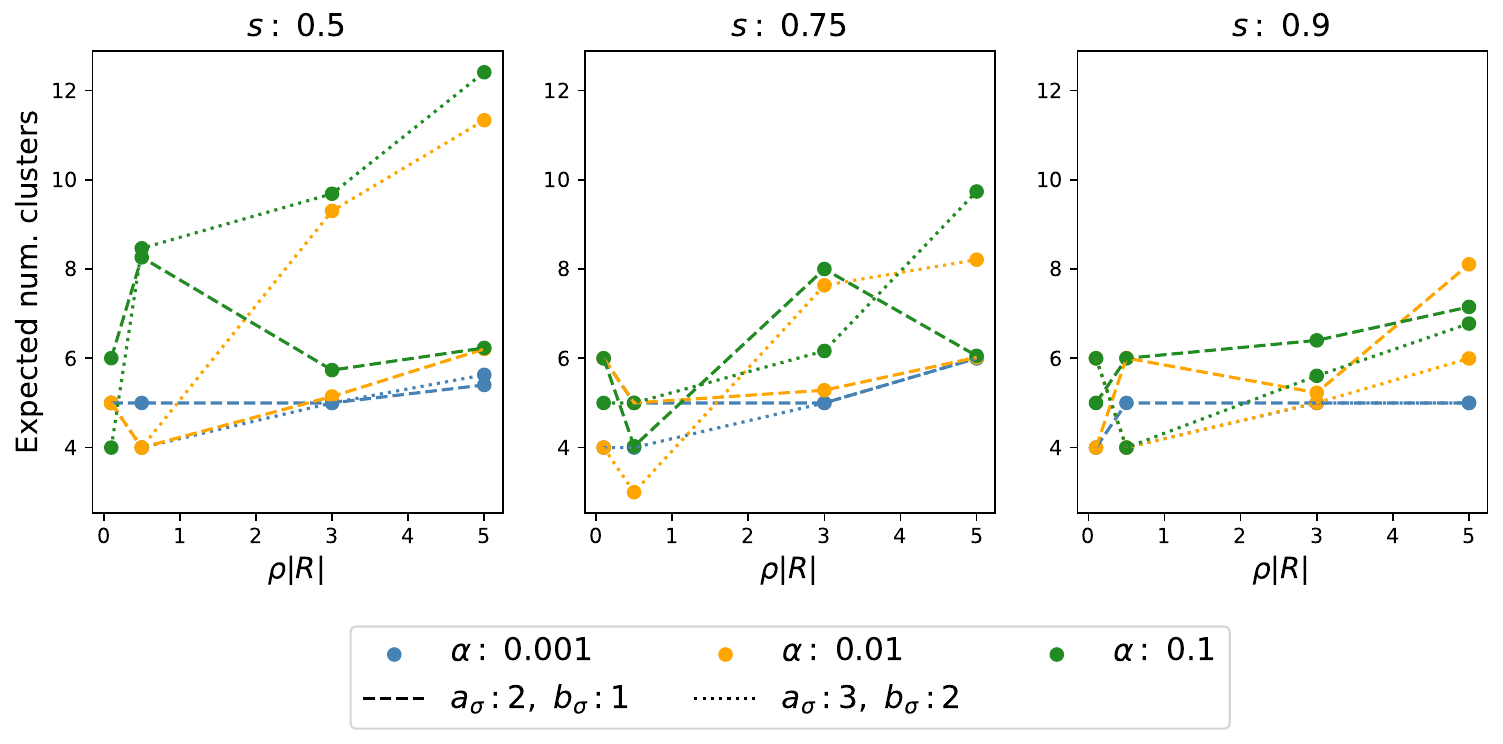}
    \caption{Posterior mean of the number of clusters for $d=4$ as a function of $\rho |R| \in \{0.1, 0.5, 3, 5 \}$, for different choices of the model's hyperparameters.}
    \label{fig:bauges_nclus_d4}
\end{figure}

\begin{landscape}
\begin{table}[h!]
\centering
\scalebox{0.8}{
\begin{tabular}{lrrrrrrl}
\toprule
 & $\rho|R|$ & $s$ & $\alpha$ & $a_\sigma$ & $b_\sigma$ & N. Clus. & WAIC \\
\midrule
0 & 0.1 & 0.5 & 0.001 & 3 & 2 & 4 & -2.02e+08 \\
1 & 0.1 & 0.5 & 0.001 & 2 & 1 & 4 & -1.94e+08 \\
2 & 0.1 & 0.5 & 0.01 & 3 & 2 & 4 & -2.07e+08 \\
3 & 0.1 & 0.5 & 0.01 & 2 & 1 & 4 & -1.94e+08 \\
4 & 0.1 & 0.5 & 0.1 & 3 & 2 & 4 & -2.21e+08 \\
5 & 0.1 & 0.5 & 0.1 & 2 & 1 & 4 & -1.92e+08 \\
6 & 0.1 & 0.75 & 0.001 & 3 & 2 & 5 & -2.08e+08 \\
7 & 0.1 & 0.75 & 0.001 & 2 & 1 & 5 & -2.00e+08 \\
8 & 0.1 & 0.75 & 0.01 & 3 & 2 & 5 & -1.97e+08 \\
9 & 0.1 & 0.75 & 0.01 & 2 & 1 & 5 & -2.04e+08 \\
10 & 0.1 & 0.75 & 0.1 & 3 & 2 & 4 & -1.99e+08 \\
11 & 0.1 & 0.75 & 0.1 & 2 & 1 & 6 & -1.99e+08 \\
12 & 0.1 & 0.9 & 0.001 & 3 & 2 & 4 & -2.13e+08 \\
13 & 0.1 & 0.9 & 0.001 & 2 & 1 & 5 & -1.89e+08 \\
14 & 0.1 & 0.9 & 0.01 & 3 & 2 & 3 & -2.05e+08 \\
15 & 0.1 & 0.9 & 0.01 & 2 & 1 & 4 & -1.89e+08 \\
16 & 0.1 & 0.9 & 0.1 & 3 & 2 & 4 & -2.07e+08 \\
17 & 0.1 & 0.9 & 0.1 & 2 & 1 & 5 & -1.98e+08 \\
18 & 0.5 & 0.5 & 0.001 & 3 & 2 & 5 & -2.10e+08 \\
19 & 0.5 & 0.5 & 0.001 & 2 & 1 & 6 & -2.01e+08 \\
20 & 0.5 & 0.5 & 0.01 & 3 & 2 & 6 & -2.21e+08 \\
21 & 0.5 & 0.5 & 0.01 & 2 & 1 & 6 & -2.05e+08 \\
22 & 0.5 & 0.5 & 0.1 & 3 & 2 & 6 & -1.96e+08 \\
23 & 0.5 & 0.5 & 0.1 & 2 & 1 & 8 & -2.03e+08 \\
24 & 0.5 & 0.75 & 0.001 & 3 & 2 & 5 & -2.17e+08 \\
25 & 0.5 & 0.75 & 0.001 & 2 & 1 & 5 & -2.00e+08 \\
26 & 0.5 & 0.75 & 0.01 & 3 & 2 & 5 & -2.19e+08 \\
27 & 0.5 & 0.75 & 0.01 & 2 & 1 & 7 & -2.02e+08 \\
28 & 0.5 & 0.75 & 0.1 & 3 & 2 & 4 & -2.15e+08 \\
29 & 0.5 & 0.75 & 0.1 & 2 & 1 & 4 & -2.13e+08 \\
30 & 0.5 & 0.9 & 0.001 & 3 & 2 & 4 & -2.03e+08 \\
31 & 0.5 & 0.9 & 0.001 & 2 & 1 & 3 & -2.03e+08 \\
32 & 0.5 & 0.9 & 0.01 & 3 & 2 & 6 & -1.91e+08 \\
33 & 0.5 & 0.9 & 0.01 & 2 & 1 & 3 & -2.04e+08 \\
34 & 0.5 & 0.9 & 0.1 & 3 & 2 & 7 & -2.07e+08 \\
35 & 0.5 & 0.9 & 0.1 & 2 & 1 & 5 & -1.98e+08 \\\bottomrule
\end{tabular}\hspace{1cm}
\begin{tabular}{lrrrrrrl}
\toprule
 & $\rho|R|$ & $s$ & $\alpha$ & $a_\sigma$ & $b_\sigma$ & N. Clus. & WAIC \\
\midrule
36 & 3.0 & 0.5 & 0.001 & 3 & 2 & 5 & -2.03e+08\\
37 & 3.0 & 0.5 & 0.001 & 2 & 1 & 6 & -1.94e+08 \\
38 & 3.0 & 0.5 & 0.01 & 3 & 2 & 7 & -2.00e+08 \\
39 & 3.0 & 0.5 & 0.01 & 2 & 1 & 8 & -1.95e+08 \\
40 & 3.0 & 0.5 & 0.1 & 3 & 2 & 9 & -1.89e+08 \\
41 & 3.0 & 0.5 & 0.1 & 2 & 1 & 10 & -1.92e+08 \\
42 & 3.0 & 0.75 & 0.001 & 3 & 2 & 6 & -2.07e+08 \\
43 & 3.0 & 0.75 & 0.001 & 2 & 1 & 6 & -1.98e+08 \\
44 & 3.0 & 0.75 & 0.01 & 3 & 2 & 7 & -2.10e+08 \\
45 & 3.0 & 0.75 & 0.01 & 2 & 1 & 8 & -1.94e+08 \\
46 & 3.0 & 0.75 & 0.1 & 3 & 2 & 11 & -2.10e+08 \\
47 & 3.0 & 0.75 & 0.1 & 2 & 1 & 11 & -1.98e+08 \\
48 & 3.0 & 0.9 & 0.001 & 3 & 2 & 5 & -2.08e+08 \\
49 & 3.0 & 0.9 & 0.001 & 2 & 1 & 5 & -2.05e+08 \\
50 & 3.0 & 0.9 & 0.01 & 3 & 2 & 7 & -1.93e+08 \\
51 & 3.0 & 0.9 & 0.01 & 2 & 1 & 8 & -2.12e+08 \\
52 & 3.0 & 0.9 & 0.1 & 3 & 2 & 9 & -1.96e+08 \\
53 & 3.0 & 0.9 & 0.1 & 2 & 1 & 10 & -1.96e+08 \\
54 & 5.0 & 0.5 & 0.001 & 3 & 2 & 8 & -1.99e+08 \\
55 & 5.0 & 0.5 & 0.001 & 2 & 1 & 8 & -1.97e+08 \\
56 & 5.0 & 0.5 & 0.01 & 3 & 2 & 9 & -1.94e+08 \\
57 & 5.0 & 0.5 & 0.01 & 2 & 1 & 8 & -1.92e+08 \\
58 & 5.0 & 0.5 & 0.1 & 3 & 2 & 11 & -2.09e+08 \\
59 & 5.0 & 0.5 & 0.1 & 2 & 1 & 12 & -1.91e+08 \\
60 & 5.0 & 0.75 & 0.001 & 3 & 2 & 5 & -2.09e+08 \\
61 & 5.0 & 0.75 & 0.001 & 2 & 1 & 6 & -1.92e+08 \\
62 & 5.0 & 0.75 & 0.01 & 3 & 2 & 7 & -1.96e+08 \\
63 & 5.0 & 0.75 & 0.01 & 2 & 1 & 9 & -1.88e+08 \\
64 & 5.0 & 0.75 & 0.1 & 3 & 2 & 10 & -1.95e+08 \\
65 & 5.0 & 0.75 & 0.1 & 2 & 1 & 11 & -1.85e+08 \\
66 & 5.0 & 0.9 & 0.001 & 3 & 2 & 5 & -2.04e+08 \\
67 & 5.0 & 0.9 & 0.001 & 2 & 1 & 6 & -2.08e+08 \\
68 & 5.0 & 0.9 & 0.01 & 3 & 2 & 9 & -1.92e+08 \\
69 & 5.0 & 0.9 & 0.01 & 2 & 1 & 8 & -1.90e+08 \\
70 & 5.0 & 0.9 & 0.1 & 3 & 2 & 9 & -2.10e+08 \\
71 & 5.0 & 0.9 & 0.1 & 2 & 1 & 11 & -1.93e+08 \\
\bottomrule
\end{tabular}}
\caption{Summary of posterior inference for the Bauges data for different choices of some hyperparameters. In all cases, the latent dimension is $d=3$. ``N. Clus.''  denotes the number of estimated clusters obtained by minimizing the posterior expectation of the Binder's loss function.}
\label{tab:bauges_d3}
\end{table}
\end{landscape}

\begin{landscape}
\begin{table}[]
\centering
\scalebox{0.8}{
    \begin{tabular}{lrrrrrrl}
\toprule
 & $\rho|R|$ & $s$ & $\alpha$ & $a_\sigma$ & $b_\sigma$ &  N. Clus & WAIC \\
\midrule
0 & 0.1 & 0.5 & 0.001 & 3 & 2 & 5 & -2.15e+08 \\
1 & 0.1 & 0.5 & 0.001 & 2 & 1 & 5 & -2.18e+08 \\
2 & 0.1 & 0.5 & 0.01 & 3 & 2 & 5 & -2.14e+08 \\
3 & 0.1 & 0.5 & 0.01 & 2 & 1 & 5 & -2.18e+08 \\
4 & 0.1 & 0.5 & 0.1 & 3 & 2 & 4 & -2.09e+08 \\
5 & 0.1 & 0.5 & 0.1 & 2 & 1 & 6 & -2.17e+08 \\
6 & 0.1 & 0.75 & 0.001 & 3 & 2 & 4 & -2.09e+08 \\
7 & 0.1 & 0.75 & 0.001 & 2 & 1 & 6 & -2.20e+08 \\
8 & 0.1 & 0.75 & 0.01 & 3 & 2 & 4 & -2.09e+08 \\
9 & 0.1 & 0.75 & 0.01 & 2 & 1 & 6 & -2.17e+08 \\
10 & 0.1 & 0.75 & 0.1 & 3 & 2 & 5 & -2.24e+08 \\
11 & 0.1 & 0.75 & 0.1 & 2 & 1 & 6 & -2.29e+08 \\
12 & 0.1 & 0.9 & 0.001 & 3 & 2 & 6 & -2.23e+08 \\
13 & 0.1 & 0.9 & 0.001 & 2 & 1 & 4 & -2.23e+08 \\
14 & 0.1 & 0.9 & 0.01 & 3 & 2 & 6 & -2.32e+08 \\
15 & 0.1 & 0.9 & 0.01 & 2 & 1 & 4 & -2.21e+08 \\
16 & 0.1 & 0.9 & 0.1 & 3 & 2 & 6 & -2.29e+08 \\
17 & 0.1 & 0.9 & 0.1 & 2 & 1 & 5 & -2.28e+08 \\
18 & 0.5 & 0.5 & 0.001 & 3 & 2 & 4 & -2.16e+08 \\
19 & 0.5 & 0.5 & 0.001 & 2 & 1 & 5 & -2.09e+08 \\
20 & 0.5 & 0.5 & 0.01 & 3 & 2 & 4 & -2.21e+08 \\
21 & 0.5 & 0.5 & 0.01 & 2 & 1 & 4 & -2.10e+08 \\
22 & 0.5 & 0.5 & 0.1 & 3 & 2 & 10 & -2.21e+08 \\
23 & 0.5 & 0.5 & 0.1 & 2 & 1 & 10 & -2.16e+08 \\
24 & 0.5 & 0.75 & 0.001 & 3 & 2 & 4 & -2.10e+08 \\
25 & 0.5 & 0.75 & 0.001 & 2 & 1 & 5 & -2.25e+08 \\
26 & 0.5 & 0.75 & 0.01 & 3 & 2 & 3 & -2.31e+08 \\
27 & 0.5 & 0.75 & 0.01 & 2 & 1 & 5 & -2.20e+08 \\
28 & 0.5 & 0.75 & 0.1 & 3 & 2 & 5 & -2.09e+08 \\
29 & 0.5 & 0.75 & 0.1 & 2 & 1 & 4 & -2.22e+08 \\
30 & 0.5 & 0.9 & 0.001 & 3 & 2 & 4 & -2.16e+08 \\
31 & 0.5 & 0.9 & 0.001 & 2 & 1 & 5 & -2.18e+08 \\
32 & 0.5 & 0.9 & 0.01 & 3 & 2 & 4 & -2.11e+08 \\
33 & 0.5 & 0.9 & 0.01 & 2 & 1 & 6 & -2.12e+08 \\
34 & 0.5 & 0.9 & 0.1 & 3 & 2 & 4 & -2.16e+08 \\
35 & 0.5 & 0.9 & 0.1 & 2 & 1 & 6 & -2.17e+08 \\\bottomrule
\end{tabular}\hspace{1cm}%
\begin{tabular}{lrrrrrrl}
\toprule
 & $\rho|R|$ & $s$ & $\alpha$ & $a_\sigma$ & $b_\sigma$ &  N. Clus & WAIC \\
\midrule
36 & 3.0 & 0.5 & 0.001 & 3 & 2 & 5 & -2.20e+08 \\
37 & 3.0 & 0.5 & 0.001 & 2 & 1 & 5 & -2.09e+08 \\
38 & 3.0 & 0.5 & 0.01 & 3 & 2 & 10 & -2.21e+08 \\
39 & 3.0 & 0.5 & 0.01 & 2 & 1 & 7 & -2.24e+08 \\
40 & 3.0 & 0.5 & 0.1 & 3 & 2 & 12 & -2.21e+08 \\
41 & 3.0 & 0.5 & 0.1 & 2 & 1 & 7 & -2.22e+08 \\
42 & 3.0 & 0.75 & 0.001 & 3 & 2 & 5 & -2.10e+08 \\
43 & 3.0 & 0.75 & 0.001 & 2 & 1 & 5 & -2.09e+08 \\
44 & 3.0 & 0.75 & 0.01 & 3 & 2 & 9 & -2.17e+08 \\
45 & 3.0 & 0.75 & 0.01 & 2 & 1 & 6 & -2.09e+08 \\
46 & 3.0 & 0.75 & 0.1 & 3 & 2 & 7 & -2.09e+08 \\
47 & 3.0 & 0.75 & 0.1 & 2 & 1 & 8 & -2.09e+08 \\
48 & 3.0 & 0.9 & 0.001 & 3 & 2 & 5 & -2.09e+08 \\
49 & 3.0 & 0.9 & 0.001 & 2 & 1 & 5 & -2.09e+08 \\
50 & 3.0 & 0.9 & 0.01 & 3 & 2 & 5 & -2.09e+08 \\
51 & 3.0 & 0.9 & 0.01 & 2 & 1 & 6 & -2.15e+08 \\
52 & 3.0 & 0.9 & 0.1 & 3 & 2 & 7 & -2.09e+08 \\
53 & 3.0 & 0.9 & 0.1 & 2 & 1 & 7 & -2.09e+08 \\
54 & 5.0 & 0.5 & 0.001 & 3 & 2 & 6 & -2.29e+08 \\
55 & 5.0 & 0.5 & 0.001 & 2 & 1 & 6 & -2.09e+08 \\
56 & 5.0 & 0.5 & 0.01 & 3 & 2 & 13 & -2.31e+08 \\
57 & 5.0 & 0.5 & 0.01 & 2 & 1 & 8 & -2.09e+08 \\
58 & 5.0 & 0.5 & 0.1 & 3 & 2 & 15 & -2.33e+08 \\
59 & 5.0 & 0.5 & 0.1 & 2 & 1 & 8 & -2.11e+08 \\
60 & 5.0 & 0.75 & 0.001 & 3 & 2 & 6 & -2.20e+08 \\
61 & 5.0 & 0.75 & 0.001 & 2 & 1 & 6 & -2.09e+08 \\
62 & 5.0 & 0.75 & 0.01 & 3 & 2 & 9 & -2.20e+08 \\
63 & 5.0 & 0.75 & 0.01 & 2 & 1 & 7 & -2.08e+08 \\
64 & 5.0 & 0.75 & 0.1 & 3 & 2 & 11 & -2.22e+08 \\
65 & 5.0 & 0.75 & 0.1 & 2 & 1 & 7 & -2.09e+08 \\
66 & 5.0 & 0.9 & 0.001 & 3 & 2 & 5 & -2.11e+08 \\
67 & 5.0 & 0.9 & 0.001 & 2 & 1 & 5 & -2.12e+08 \\
68 & 5.0 & 0.9 & 0.01 & 3 & 2 & 7 & -2.10e+08 \\
69 & 5.0 & 0.9 & 0.01 & 2 & 1 & 10 & -2.16e+08 \\
70 & 5.0 & 0.9 & 0.1 & 3 & 2 & 8 & -2.09e+08 \\
71 & 5.0 & 0.9 & 0.1 & 2 & 1 & 8 & -2.09e+08 \\
\bottomrule
\end{tabular}}
\caption{Summary of posterior inference for the Bauges data for different choices of some hyperparameters. In all cases, the latent dimension is $d=4$. ``N. Clus.''  denotes the number of estimated clusters obtained by minimizing the posterior expectation of the Binder's loss function.}
\label{tab:bauges_d4}
\end{table}
\end{landscape}

\FloatBarrier

\subsection{Scientific names of the plant species}

We report the scientific names of the species displayed in Figure~\ref{fig:relevant_species_by_cluster}.

\begin{itemize}
    \item  ``0'': Acer pseudoplatanus L. 
    \item ``1'': Astrantia major L.
    \item ``4'': Laserpitium latifolium L.
    \item ``5'': Pimpinella major (L.) Hudson
    \item ``6'': Hedera helix L.
    \item ``9'': Cirsium palustre (L.) Scop.
    \item ``14": Prenanthes purpurea L.
  \item ``27'': Helianthemum nummularium (L.) Miller
  \item ``28": Corylus avellana L.
 \item ``45": Fagus sylvatica L.
 \item ``46": Gentiana lutea L.
 \item ``49'': Geranium sylvaticum L.
 \item ``61": Fraxinus excelsior L.
 \item ``72": Ranunculus tuberosus Lapeyr.
 \item ``83'': Galium odoratum (L.) Scop.
 \item ``91'': Carex sempervirens Vill. subsp. sempervirens
 \item ``97": Anthoxanthum odoratum L.
 \item ``112": Abies alba Miller  
 \item ``117": Acer campestre L.
 \item ``121": Rubus fruticosus groupe
\end{itemize}

\end{document}